\documentclass[aps,twocolumn,
superscriptaddress,
footinbib,
pra, 10pt, longbibliography, floatfix]{revtex4-2}

\usepackage{graphicx}
\usepackage{enumitem}
\usepackage{mwe}
\usepackage{mathtools}
\usepackage{bbold}
\usepackage{amsmath,amssymb,bm}
\usepackage{mathrsfs}
\usepackage{MnSymbol}
\usepackage{graphicx}
\usepackage{epstopdf}
\usepackage[normalem]{ulem} 
\usepackage[caption=false]{subfig}
\usepackage{comment}
\usepackage[usenames,dvipsnames]{color}
\usepackage{natbib}
\usepackage{xcolor, colortbl}
\usepackage{physics}
\usepackage{multirow}
\usepackage{stackengine}
\usepackage{dsfont}
\usepackage{relsize}
\usepackage{framed}
\usepackage[utf8]{inputenc} 
\let\longtable*\relax
\usepackage{longtable}
\usepackage[acronym]{glossaries}
\usepackage{ragged2e}
\usepackage{algorithm}
\usepackage{algpseudocode} 
\usepackage{amsthm} 
\usepackage{thm-restate}
\usepackage{booktabs}
\usepackage{tabularx}
\usepackage{subfig}
\usepackage{tikz}
\usetikzlibrary{quantikz2}
\usepackage{diagbox}
\usepackage{hyperref}
\usepackage[nameinlink]{cleveref}

\newtheorem{prop}{Proposition}

\newtheorem*{lemma*}{Lemma}

\begin{document}
\title{Mitigation of correlated readout errors without randomized measurements}

\date{\today}

\author{Adrian Skasberg Aasen}
\email{adrian.aasen@uni-jena.de}
\affiliation{Kirchhoff-Institut f\"{u}r Physik, Universit\"{a}t Heidelberg, Im Neuenheimer Feld 227, 69120 Heidelberg, Germany}
\affiliation{Institut für Festk\"{o}rpertheorie und -optik,  Friedrich-Schiller-Universit\"{a}t Jena, Max-Wien-Platz 1, 07743 Jena, Germany}

\author{Andras Di Giovanni}
\affiliation{Physikalisches Institut, Karlsruher Institut für Technologie, 76131 Karlsruhe, Germany}

\author{Hannes Rotzinger}
\affiliation{Physikalisches Institut, Karlsruher Institut für Technologie, 76131 Karlsruhe, Germany}
\affiliation{Institut für QuantenMaterialien und Technologien, Karlsruher Institut für Technologie, 76344 Eggenstein-Leopoldshafen, Germany}

\author{Alexey V. Ustinov}
\affiliation{Physikalisches Institut, Karlsruher Institut für Technologie, 76131 Karlsruhe, Germany}
\affiliation{Institut für QuantenMaterialien und Technologien, Karlsruher Institut für Technologie, 76344 Eggenstein-Leopoldshafen, Germany}

\author{Martin G\"{a}rttner}
\email{martin.gaerttner@uni-jena.de}
\affiliation{Institut für Festk\"{o}rpertheorie und -optik,  Friedrich-Schiller-Universit\"{a}t Jena, Max-Wien-Platz 1, 07743 Jena, Germany}

\begin{abstract}
Quantum simulation, the study of strongly correlated quantum matter using synthetic quantum systems, has been the most successful application of quantum computers to date. It often requires determining observables with high precision, for example when studying critical phenomena near quantum phase transitions. Thus, readout errors must be carefully characterized and mitigated in data postprocessing, using scalable and noise-model agnostic protocols. We present a readout error-mitigation protocol that uses only single-qubit Pauli measurements and avoids experimentally challenging randomized measurements. The proposed approach captures a very broad class of correlated noise models and is scalable to large qubit systems. It is based on a complete and efficient characterization of few-qubit correlated positive operator-valued measures, using overlapping detector tomography. To assess the effectiveness of the protocol, observables are extracted from simulations involving up to 100 qubits employing readout errors obtained from experiments with superconducting qubits.

\end{abstract}

\maketitle

\section{Introduction}
\label{sec:introduction}
A significant challenge in scaling up quantum devices to regimes where they can demonstrate a quantum advantage is their sensitivity to noise \cite{Altman2021}. Various error mitigation protocols for near-term quantum hardware have been developed to extend the reach of non-fault-tolerant devices \cite{Temme2017,BonetMonroig2018, vandenBerg2023, Cai2023}. Although the value of these short-term solutions has recently been subject to debate \cite{Quek2024}, there are arguments suggesting that they will remain crucial for reaching practical quantum advantage \cite{Zimboras2025}. What is more, the mitigation of \textit{readout} errors takes a distinct and important role in many applications, as these errors persist even when fault tolerance is implemented. As quantum hardware scales up, unintended readout correlations between qubits become a mounting issue that has only recently been discussed \cite{Tuziemski2023, Zhou2023, Geller2021, Sarovar2020}. In superconducting qubits, readout correlations have been observed in multiplexed dispersive readout between pairs of qubits \cite{Heinsoo2018}. Early efforts to reduce such readout correlations focused on fine-tuning quantum control \cite{deGroot2010} or mitigating them postmeasurement with specialized hardware \cite{Lienhard2022}. 


We have identified four properties which are desirable in practical near-term readout error-mitigation protocols:
\begin{enumerate}[label=\textbf{\Alph*}.]
    \item scalability to large qubit numbers,
    \item robustness against correlated and coherent errors,
    \item access to relevant observables,
    \item practicality of experimental implementation.
\end{enumerate}
Among the most studied readout errors are classical errors, which amount to a statistical redistribution of measurement outcomes. These errors are commonly modeled using a confusion matrix \cite{Maciejewski2020, Nation2021} or addressed by unfolding \cite{Nachman2020}. The methods used to mitigate these types of errors have the advantage of being extremely simple to implement and typically account for the largest part of readout errors observed in experiments. These error-mitigation methods can also be made scalable \cite{Tuziemski2023}, which makes them satisfy properties \textbf{A, C} and \textbf{D}, but fail to address coherent errors. Randomized measurement strategies, such as classical shadows \cite{Huang2020}, are of particular interest because many nonclassical errors are transformed to classical ones under randomization. Classical shadows have been made resistant to noise by the introduction of robust shadow tomography \cite{Chen2021} and related methods \cite{Jnane2024,Koh2022, Arrasmith2023}. These techniques satisfy properties \textbf{A-C}, but are arguably difficult to implement due to the need for randomized global Clifford measurements. Recent efforts have been made to alleviate some of the difficulties of randomized measurements \cite{Hu2025, Onorati2024}, but they often require entangling gates (CNOT) to be effective. Other methods that avoid randomized measurement have been proposed \cite{Aasen2024}, but are not scalable, satisfying only properties \textbf{B-D}.  Correlated readout errors have been avoided entirely with compression readout \cite{Ding2023, Huang2025}, which involves only measurements of a single ancilla qubit, but requires relatively complicated compression circuits, satisfying properties \textbf{A-C}.  To the authors' knowledge, there is no readout error-mitigation method that satisfies properties \textbf{A-D} simultaneously.

In this work, we extend the readout error mitigated state tomography (REMST) protocol introduced in Ref.~\cite{Aasen2024} to be scalable while maintaining desirable properties with only small caveats to \textbf{B} and \textbf{C}. The REMST protocol is based on the calibration of the measurement device using detector tomography \cite{Lundeen2008} and the integration of the reconstructed positive operator-valued measure (POVM) directly into a likelihood-based state estimator. 
Here, we develop a scheme that uses this protocol for the readout error-mitigated extraction of few-body observables that is scalable if readout-induced correlations are limited to small subsystems of qubits. Under this assumption, we can create a tiling of the whole qubit system with at most $n_\text{corr}$ qubits in each tile. The tiling for a qubit device can be constructed efficiently by performing detector overlapping tomography \cite{Tuziemski2023, Cotler2020} and extracting readout correlation coefficients \cite{Tuziemski2023, Maciejewski2021} for all pairs of qubits. Using the correlation coefficients, one can create correlated noise clusters using a heuristic-free hierarchical clustering algorithm.

To demonstrate the utility and necessity of the protocol, we consider different realizations of readout error mitigation without randomized measurements and compare how well they perform in relevant tasks, such as reconstructing few-body reduced density matrices and extracting arbitrary two-point observables. Comparisons are made by numerically simulating measurements from noisy POVMs for system sizes between 16 and 100 qubits. To demonstrate the robustness of our scheme, we consider both artificial POVMs with correlated noise and noisy POVMs extracted from superconducting qubits by performing detector tomography. We find that our protocol performs well across all scenarios tested. When no knowledge of the prepared state is available, our protocol outperforms other alternatives by up to an order of magnitude in mean-squared errors of arbitrary two-point observables.

Our protocol will yield crucial advantages in quantum simulation experiments \cite{Georgescu2014} where precise extraction of few-body observables is necessary but may be limited by readout errors. For example, preparing ground states for studying quantum phases of matter using variational quantum circuits requires measuring the expectation value of the system Hamiltonian which, for quantum spin models or Hubbard models, typically consists of sums of low-weight operators. In addition, for quantum simulations of the critical behavior around phase transitions \cite{Keesling2019, Zhang2025, Schuckert2025}, two-point correlators need to be precisely measured to extract critical exponents. 
When studying nonequilibrium quantum dynamics, the accurate determination of two-point correlators is key for determining the rate of information spreading and probing Lieb-Robinson bounds \cite{Richerme2014, Jurcevic2014}.  Furthermore, the tools and approach presented can be used separately for characterization of the readout correlation structure, which can be used to optimize experimental readout parameters \cite{DiGiovanni2025B}. 
Thus, we expect our scalable readout error-mitigation scheme to be of wide use to quantum simulation experiments, opening up new opportunities for studying quantum many-body phenomena.




\section{Preliminaries}

\subsection{Generalized quantum measurements}
A generalized quantum measurement, also called a positive operator-valued measure (POVM), is a set of Hermitian matrices $\mathbf{M} = \{M_i\}$ that satisfies the following properties
\begin{equation}
    M_i\geq 0, \quad\quad M_i^\dagger = M_i, \quad\quad \text{and} \quad\quad \sum_i M_i = \mathbb{1}.
    \label{eq:POVM_def}
\end{equation}
Each POVM element $M_i$ is associated with an outcome $i$ from a measurement process. The Born rule provides the probabilities for the different possible outcomes when measuring quantum state $\rho$,
\begin{equation}
\label{eq:Borns_rule}
    \Tr(\rho M_i) = \langle M_i \rangle = p_i.
\end{equation}
This provides an operational interpretation of the expectation value of POVM elements as a probability distribution. Thus, the three properties in Eq.~\eqref{eq:POVM_def} are equivalent to positivity, realness, and normalization of the probability distribution, respectively.

The benefit of working with POVMs rather than just projective measurements is that it allows us to capture the process of noise at readout, both classical noise, which manifests itself as a redistribution of the eigenvalues of the POVM elements, and coherent errors, which occur in the off-diagonal elements. 

\subsection{Readout error mitigated state tomography}
In Ref.~\cite{Aasen2024}, a readout error-mitigation protocol was introduced that allows correction of beyond-classical readout errors in quantum state reconstruction. This work is a direct continuation and will provide a scalable protocol.
The basic outline of the approach is as follows:
Quantum detector tomography is used as a calibration step, where the effective POVM $\mathbf{M}$ implemented by the measurement device is reconstructed \cite{Fiurek2001}. The reconstructed POVM is then integrated into the likelihood function 
\begin{equation}
    \mathcal{L}(\rho) \propto \Pi_i \Tr(\rho M_i)^{n_i},
\end{equation}
where $n_i$ is the number of occurrences of outcome $i$. Using a Bayesian mean estimator \cite{BlumeKohout2010} or a maximum likelihood estimator \cite{Lvovsky2004}, a most probable state can be found that is guaranteed to be physical.

\subsection{Correlation coefficients and traced-out POVMs}
\label{sec:correlation_coefficients}

To account for correlated readout errors, it is essential to have a precise method to assess how the readout of one qubit affects the readout of another. To do this, one can check how separable a POVM is between two subsystems. In contrast to computing reduced quantum states, extracting a traced-out POVM from a larger POVM lacks a unique definition and necessitates knowing the quantum state being measured. To create an unambiguous notion of a reduced POVM we follow Ref.~\cite{Tuziemski2023}. Consider a bipartite quantum system, A and B, with POVM $\mathbf{M}^{\text{AB}} = \{M^{\text{AB}}_{i_\text{A}i_\text{B}}\}$, where $i_\text{A}$ and $i_\text{B}$ label the outcome in subsystems A and B, respectively. To derive the POVM that acts on the subsystem A, we define the traced-out POVM $M^{\text{A},\rho_\text{B}}$ conditioned on subsystem B being prepared in state $\rho_\text{B}$ as
\begin{equation}
    M^{\text{A},\rho_\text{B}}_{i_\text{A}} = \sum_{i_\text{B}} \Tr_\text{B}(M^{\text{AB}}_{i_\text{A}i_\text{B}} \left(\mathbb{1}_\text{A} \otimes \rho_\text{B} \right)).
    \label{eq:traced_out_POVM}
\end{equation}
By specifying the state of subsystem B, the relative occurrence of the different possible reduced POVM elements on subsystem A is fixed, and therefore unambiguous.

As initially outlined in Refs.~\cite{Maciejewski2021, Tuziemski2023}, correlation coefficients can be formulated using traced-out two-qubit POVMs. These coefficients are defined as the maximum distance between two reduced POVMs, obtained by preparing the traced-out subsystem in different pure states. Following to the two-qubit "worst-case" definition from Ref.~\cite{Tuziemski2023}, the correlation coefficient obtained from the computational basis POVM is explicitly given by
\begin{equation}
\label{eq:correlation_coefficients}
    c_{\text{B} \rightarrow \text{A}} = \sup_{\rho_\text{B}, \sigma_\text{B}} ||M_0^{\text{A}, \rho_\text{B}}- M_0^{\text{A}, \sigma_\text{B}}||_\infty,
\end{equation}
where it suffices to consider the 0 element of the reduced POVM of subsystem A, as it fully specifies the one-qubit POVM due to the normalization condition $M_0+M_1=\mathbb{1}$.
This coefficient quantifies how strongly the readout of subsystem B impacts the readout of subsystem A. Since this definition of the correlation coefficient is not symmetric, we define a symmetric version 
\begin{equation}
    c_{\text{A} \leftrightarrow \text{B}} =\frac{c_{\text{A} \rightarrow \text{B}} + c_{\text{B} \rightarrow \text{A}}}{2}
\end{equation}
which we will refer to as the correlation coefficient from now on. Certain limitations of the correlation coefficients are addressed in Appendix~\ref{App:Limitation_correlation_coeff}.

\subsection{Overlapping detector tomography and perfect hash families}
\label{sec:overlapping_tomography}
Our protocol requires the extraction of all pairwise correlation coefficients between qubits. To achieve this efficiently for large-scale qubit systems while minimizing the experimental cost, we employ overlapping detector tomography \cite{Tuziemski2023}. Various forms of overlapping tomography have been introduced \cite{Cotler2020,GarcaPrez2020,BonetMonroig2020}, each grounded in the principles of covering arrays or perfect hash families (PHF). For simplicity, this discussion will focus on the variant using PHFs proposed by Ref.~\cite{Cotler2020}, as it offers the greatest flexibility within our framework. For our purposes, a hash family can be denoted as $\Phi_{N,k,v} = \{\phi^{k,v}_i\}_{i=1}^N$, which is a set of $N$ hash functions $\phi$ that assign to $k$ possible input values one of $v$ possible output values \cite{WalkerII2007, dougherty2019}. For a hash family to be perfect, it requires that for any arbitrary subset of $t\leq v$ input values $\{x_1, x_2, \dots ,x_t\}$, there exists a hash function $\phi_i \in \Phi_{N,k,v}$ such that $\{\phi_i(x_1), \phi_i(x_2), \dots, \phi_i(x_t)\}$ has $t$ unique output values. We refer to these as $t$-local PHFs.
In other terms, a $t$-local PHF contains hash functions that ensure that for any collection of $t$ input values, there exists at least one hash function that assigns distinct output values to each of the $t$ inputs. For the rest of this work, we will consider PHFs that have $v=t$. In this work, it is only necessary to use a PHF with $t=2$, which can be generated analytically \cite{Cotler2020},
\begin{equation}
    \phi_i^{k,2}(x) = i\text{th digit of the binary expansion of }(x-1),
    \label{eq:2-local_PHF}
\end{equation}
where $x$ denotes the input value. 

In overlapping detector tomography, the PHFs are used to guarantee that all possible two-qubit POVMs can be reconstructed, which means that, for all possible two-qubit pairs, an informationally complete set of calibration states needs to be measured between them. Specifically, a minimal informationally complete set of calibration states for a single qubit comprises four calibration states, $\{\psi_1, \psi_2, \psi_3, \psi_4\}$, and for two qubits, all $16$ combinations of two single-qubit calibration states are necessary. As an example of how PHFs are used to generate a complete set of calibration states for all qubits, consider the 2-local PHF described in Eq.~\eqref{eq:2-local_PHF}. In this context, the input values $x$ are the indices of each qubit, $x \in \{1,2,3,\dots, n_\text{qubits}\}$.  For each hash function $\phi_i^{k,2}$, each qubit is assigned one of two possible output values 0 or 1. An informationally complete set of calibration states is prepared between the qubits that were assigned 0 and 1.  Specifically, all possible unique combinations of the single-qubit calibration states are prepared between the two groups, with no repeated calibration state, that is, for qubits with label (0,1)  the calibration states $\{(\psi_1, \psi_2), (\psi_1, \psi_3), (\psi_1, \psi_4), (\psi_2, \psi_1), \dots, (\psi_4, \psi_3) \}$ are prepared and measured, where all cases with equal calibration states prepared on each qubit group are removed. Once this is done for all hash functions, four final calibration states are measured with all qubits prepared in $\psi_1$, $\psi_2$, $\psi_3$, and $\psi_4$, respectively. This procedure is straightforwardly generalized to a $t$-local PHF, which ensures that any $t$-qubit POVM can be reconstructed. If the calibration states are replaced by single-qubit basis measurements $\{X,Y,Z\}$, one recovers the overlapping tomography prescription in Ref.~\cite{Cotler2020}. 

The total number of measurements required to reconstruct any $t$-qubit POVM from a $k$-qubit system is proportional to the number of hash functions in the PHF. The expected number of hash functions scales as $e^{\mathcal{O}(t)}\log{k}$ \cite{Cotler2020}, that is, exponential in the subsystem size but logarithmic in the total size of the system, making overlapping detector tomography scalable.

There exist situations where larger $t$-local PHF are useful and we do use a PHF with $t=3$ for fairer comparisons later. Creating an optimal PHF is a hard problem, and few analytical methods are available to generate them \cite{dougherty2019, WalkerII2007}. We provide some PHFs generated by a density algorithm for different $k$ and $t$ in Ref.~\cite{Aasen2025}. An example of the $\Phi_{16,15,3}$ PHF is presented in Appendix~\ref{App:perfect_hash_families}. Although the proposed measurement strategy may not be resource optimal, it offers significant adaptability. Additional optimizations could be achieved through the use of covering arrays \cite{GarcaPrez2020,Kiara2024} 

\section{Method}
\label{sec:method}
This section introduces our main result, a protocol specifically devised to address correlated readout errors in any arbitrary two-qubit observable, with the potential to extend to higher-order observables in certain subsystems. As discussed in Sec.~\ref{sec:introduction}, high-precision measurements of low-order observables in the presence of correlated readout errors are crucial for many applications in quantum simulation. The correlated quantum readout error-mitigation (QREM) protocol is depicted in Fig.~\ref{fig:protocol}. The protocol can be broadly split into two separate parts, readout correlation characterization and noise-cluster-based REMST. The first part is used to sample efficiently extract the correlated structure of the qubit readout which enables construction of correlated noise clusters. The obtained correlated noise structure is used to extract error mitigated observables in the subsequent part.  


\begin{figure*}
 \centering
 \includegraphics[width=0.95\linewidth]{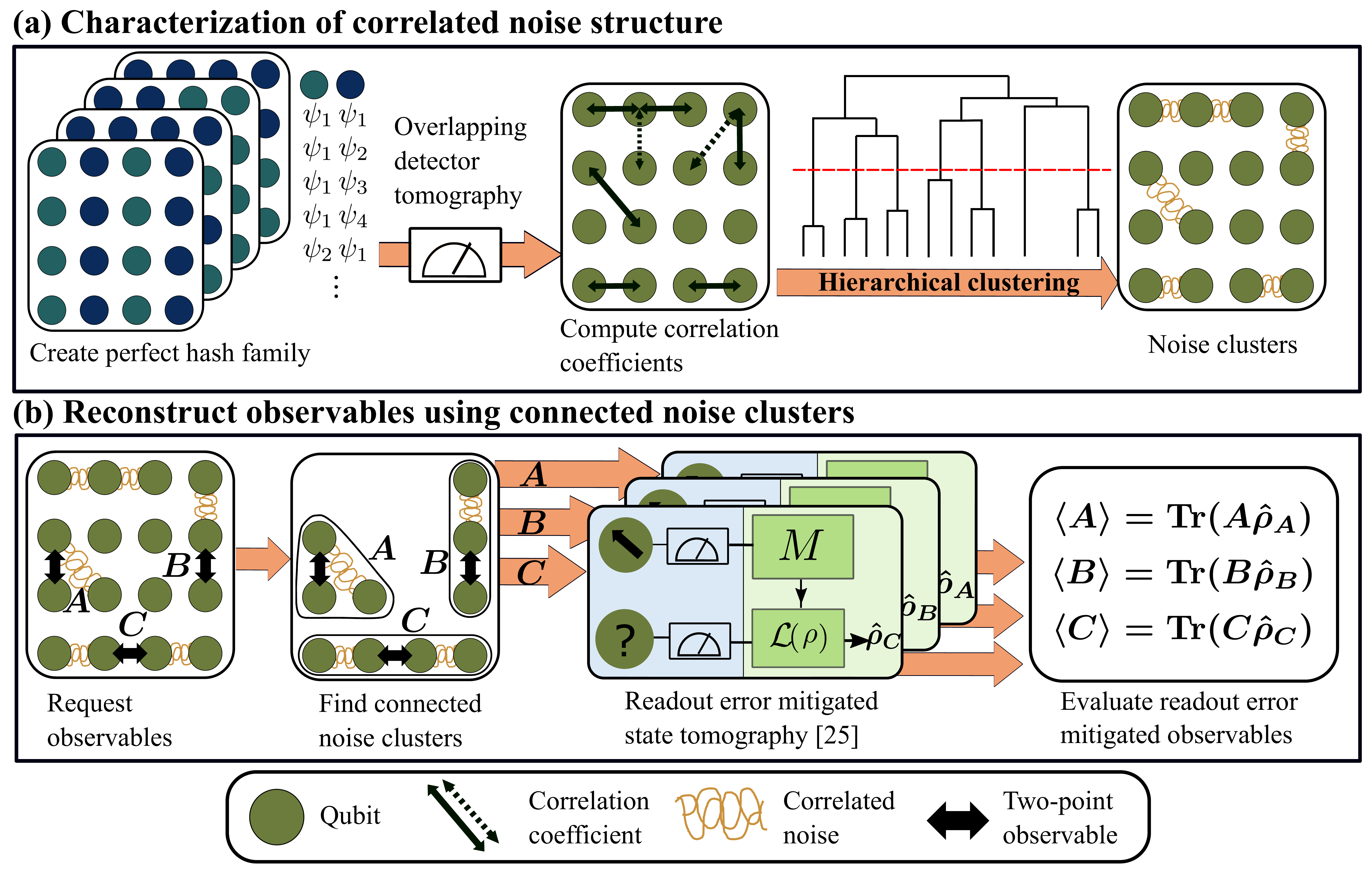}
 \caption{Schematic overview of the correlated readout error-mitigation protocol. \textbf{(a)} The correlated noise is efficiently characterized for the relevant qubit system. A perfect hash family is generated in such a way that all possible two-qubit reduced POVMs can be reconstructed by utilizing parallel measurements. Every hash function assigns one of two distinct labels to each qubit, indicating which calibration state is assigned to each group. Each group label receives an informationally complete set of calibration states between them which is repeatedly measured. Overlapping detector tomography can then be performed where every possible two-qubit POVM is reconstructed. For each two-qubit POVM, the symmetric correlation coefficients are calculated, see Eq.~\eqref{eq:correlation_coefficients}. These coefficients are represented with bidirectional arrows. The dotted bidirectional arrow indicates weaker correlation coefficients. By employing a hierarchical clustering algorithm, the dominant readout noise clusters within the system can be identified. \textbf{(b)} Reconstructing observables with mitigated readout errors. Each requested observable is matched with readout noise clusters to form connected noise clusters, which must be fully characterized. On these connected noise clusters, readout error-mitigated state tomography is carried out to reconstruct the error-mitigated states for each connected cluster. The observables are then computed from the readout error-mitigated states.}
 \label{fig:protocol}
\end{figure*}

\subsection{Noise-cluster characterization}
In general, readout correlations may exist among all qubits in any quantum system, rendering any complete characterization nonscalable. Therefore, in order to efficiently characterize the noise in the entire system, assumptions about the correlation structure have to be made. Fortunately, the extent of correlation in contemporary systems is typically limited, and readout correlations occur only within smaller groups, such as for superconducting qubits that share the same readout microwave line. Therefore, a maximum number of qubits that can be correlated $n_{\text{corr}}$ can be set. 

To find which qubits should be grouped together into noise clusters, we rely on two-qubit correlation coefficients, which are computed from two-qubit POVMs, see Eq.~\eqref{eq:correlation_coefficients}. Reconstructing all possible two-qubit POVMs would require an exhaustive amount of measurements if done naively by measuring all two-qubit POVMs one after another. It can be done significantly more efficiently by employing overlapping detector tomography and perfect hash families to measure all qubits in parallel, see Sec.~\ref{sec:overlapping_tomography}. 

A $2$-local PHF is generated to ensure that each set of two qubits is measured with an informationally complete set of calibration states. Higher $k$'s can be used to create a ``measure first, ask questions later" approach, where the initial calibration measurement can be reused in noise cluster characterization later. In the representation in Fig.~\ref{fig:protocol}\textbf{(a)}, there are 16 qubits, for which four hash functions are sufficient. In each sheet all the qubits of the same colors received the same measurement instructions such that between the two groups there exists an informationally complete set of measurements. Once all four sheets are measured, one can reconstruct all possible two-qubit POVMs. With the two-qubit POVMs the correlation coefficients can be computed as described in Eq.~\eqref{eq:correlation_coefficients}. 

The correlation coefficients form a fully connected graph, facilitating the use of clustering algorithms to identify groups of qubits with the strongest readout correlations. Hierarchical clustering techniques consistently identified qubit clusters across different noise models and experimentally extracted POVMs. These methods provide the benefit of eliminating the need for heuristic parameters, which often differ significantly between different experimental setups or noise models. Although further manual cluster tuning is possible with extra insight into the noise sources, it is not essential for the protocol to function effectively. For more details on the clustering algorithm, see Appendix~\ref{App:Hierarchical_clustering}.

\subsection{Reconstructing readout error-mitigated observables}
The identified noise clusters are used to perform readout error mitigation, cf. Fig.~\ref{fig:protocol} \textbf{(b)}. For a given two-qubit observable to be measured, all noise clusters containing those qubits are grouped to form a \emph{connected} noise cluster. The readout error-mitigated state is reconstructed using the REMST protocol~\cite{Aasen2024} on the entire connected noise cluster. In Appendix ~\ref{App:proving_obs_12} we show that, in the case of a general prepared state, state reconstruction on connected noise clusters is required if both correlated and coherent readout errors are to be mitigated. 
The reconstruction of the noisy POVMs can be carried out on each noise cluster individually.

The observables initially requested are calculated using the readout error-mitigated states received from the REMST protocol. The connected noise clusters can be reused to compute further higher-order observables, even if not originally requested. For details on the implementation, see Appendix~\ref{App:protocol}.




\subsection{Extraction of experimental POVMs}
\label{Sec:Experimentally_extracted_POVMs}
To extract realistic POVMs to test the protocol, detector tomography was performed on a chip with four individual transmon (Xmon) qubits. Each qubit is dispersively coupled to a separate readout resonator. The qubit control and the resonator readout is done with frequency division multiplexing \cite{Jerger2012, Krantz2019}. The chip is measured in a dilution cryostat below $ 15 \, \text{mK}$. For more information on the device used, see Ref.~\cite{DiGiovanni2025A}. Some of the qubits have larger readout correlations as a result of the smaller frequency spacing of the resonators, leading to potential crosstalk.

\section{Results}
\label{sec:results}
To demonstrate that all the desired properties stated in the Introduction are satisfied by the protocol, measurements were simulated on up to 100 qubits. These simulations were based on readout noise extracted from a superconducting qubit system. The properties examined were as follows:
\begin{enumerate}[label=\textbf{\Alph*}.]
    \item scalability to large qubit numbers,
    \item robustness against correlated and coherent errors,
    \item access to relevant observables.
\end{enumerate}
The property \textbf{D} is satisfied by construction as only single-qubit Pauli measurements were used. To make the sampling and noise simulation numerically feasible, we used tensor products of four-qubit states as test states for the QREM protocols and applied noise channels that potentially overlap multiple four-qubit states. This should not be viewed as a limitation of the protocol because the results and protocol complexity would not change if general 16- and 100-qubit states were prepared, see Appendix~\ref{App:proving_obs_12} for more details.


The protocol is tested by reconstructing two-qubit observables. The target states are tensor products of four-qubit Haar-random states, $\rho = \bigotimes_{i=1}^{n_{\text{chunks}}} \rho_i$, where $n_\text{chunks} = n_{\text{qubits}}/4$. The Haar-random pure states $\rho_i$ are generated by drawing a random unitary matrix $U_i$ from the unitary group $U(2^4)$ and applying it to the all-zero state $\rho_0 = \ket{0000}\bra{0000}$, resulting in $\rho_i = U_i \rho_0 U_i^\dagger$. The considered observables are products of Pauli operators  $O = A \otimes B$, where $A, B \in\{X, Y, Z\}$. In Appendix \ref{app:100_qubits_mixed_states} we include a simulation with mixed target states.

Multiple layers of averaging were performed to achieve representative performance estimates. Specifically, each simulation involved numerical measurement samples from $N_s=10$ randomly chosen target states. For each state, $N_p$ random pairs of qubits were selected, and for each selected pair, all two-qubit Pauli observables were reconstructed.
As figures of merit, the mean-squared error (MSE) and state reconstruction infidelity were used. The mean-square error is computed between the reconstructed observables $\hat{O} = \Tr(O \hat{\rho})$ and their ideal values $\langle O\rangle$ from the Haar-random target states,
\begin{equation}
    \label{eq:MSE}
    \text{MSE} = \frac{1}{9} \sum_{i=1}^{9} \left(\langle O_{i} \rangle  -\hat O_i\right)^2,
\end{equation}
where the sum is performed over nine possible two-qubit Pauli observables. Unless stated otherwise, the MSE is also averaged over the $N_s$ random target states and $N_p$ pairs of qubits, which is denoted as the averaged MSE. 
State reconstruction infidelity is computed between the reconstructed error-mitigated state $\hat \rho$ and the target state $\rho$, defined as 
\begin{equation}
    \text{Infidelity} =  1 - F(\rho, \hat \rho) = 1-\left[ \Tr\left(\sqrt{\sqrt{\rho} \hat \rho \sqrt{\rho}}\right)\right]^2.  
\end{equation}
Since only pure target states $\rho$ are used, infidelity simplifies to 
\begin{equation}
    \text{Infidelity} = 1 - \Tr(\rho \hat \rho).
\end{equation}
As for the MSE, infidelity is averaged over $N_s$ random states and $N_p$ random pairs unless otherwise stated.  

To evaluate the statistical robustness of the protocol, each simulation was repeated 10 times and averaged. The standard deviations of these repetitions are shown as error bars in each plot. For all simulations, the maximal size of the noise cluster was set to $n_{\text{corr}}=3$. 

\subsection{16-qubit example with experimental noise}
\label{sec:16-qubit_example}

\begin{figure}
 \centering
 \includegraphics[width=0.95\linewidth]{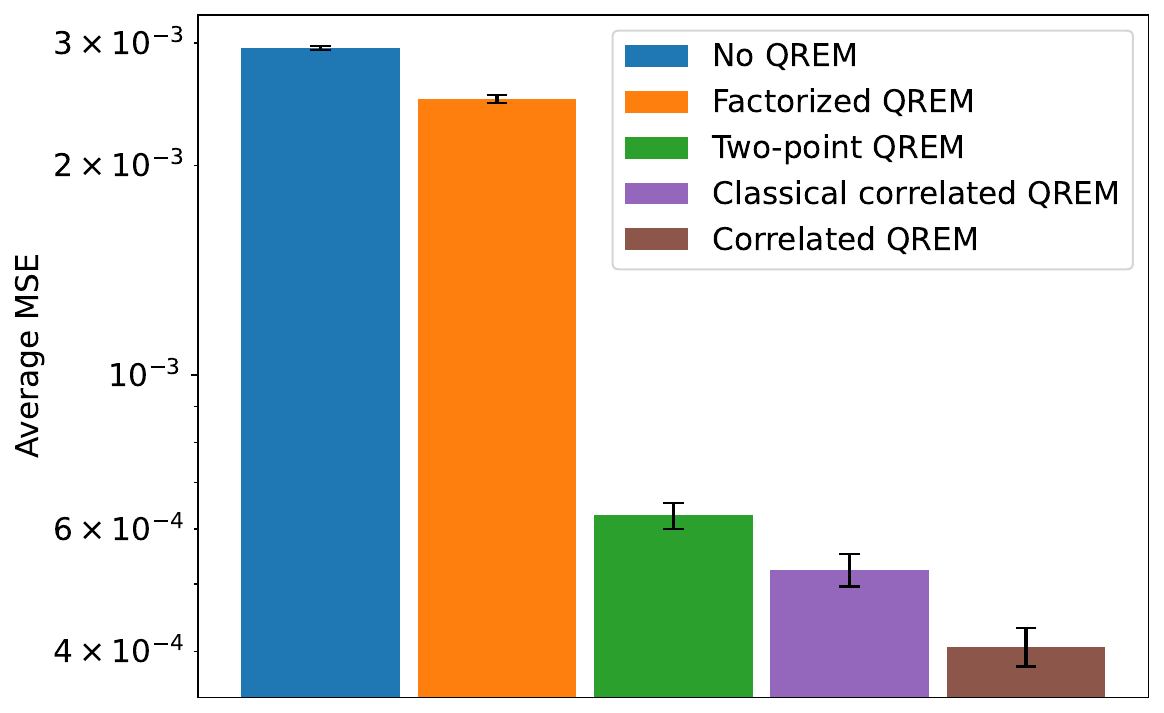}
 \caption{Comparing different scalable readout error-mitigation protocols. Averaged mean-squared error of two-qubit Pauli observables is plotted for five different protocols. The averages are performed over $N_p=20$ randomly chosen pairs for $N_s=10$ random realizations of the 16-qubit states. The error bars are one standard deviation from 10 repeated simulations of the entire protocol with the same readout noise, random states, and pairs chosen.}
 \label{fig:16_qubit_example}
\end{figure}

\begin{figure*}
 \centering
 \includegraphics[width=0.8\linewidth]{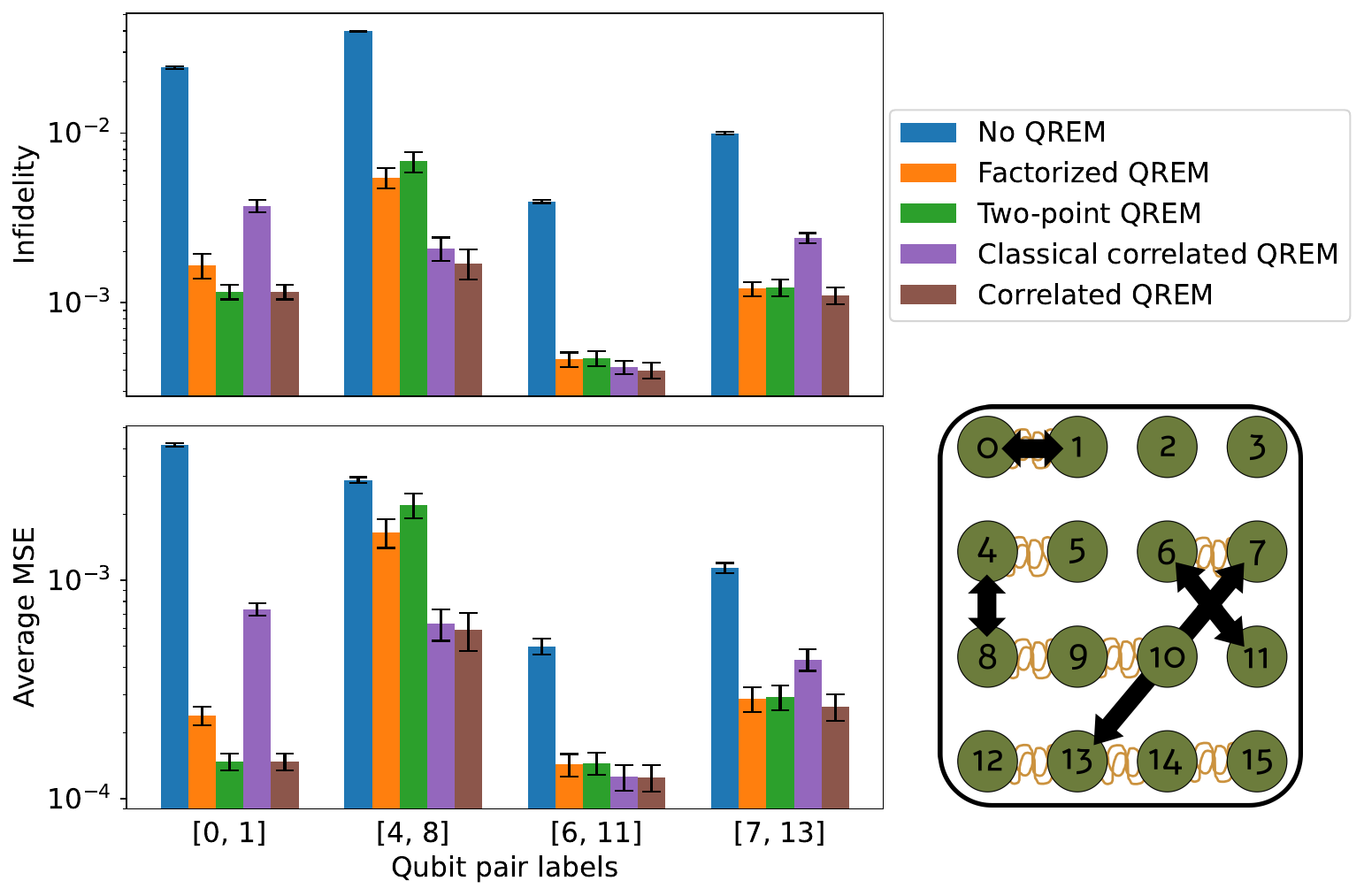}
 \caption{Average readout error-mitigated observables for selected two-qubit pairs. In a 16-qubit simulation, state infidelity and average mean-squared error are evaluated across five distinct implementations of the scalable QREM protocol. The measurement setup is the same as for Fig.~\ref{fig:16_qubit_example}. In the bottom right the experimental $n$-qubit POVMs applied are visualized by orange squiggly lines, indicating potential correlated readout errors.  The specific qubit pairs plotted are indicated by black bidirectional arrows. }
 \label{fig:16_qubit_correlator}
\end{figure*}

To simplify our analysis, we initially focus on a 16-qubit system to test the performance of the protocol. Five different implementations of the protocol are compared in scenarios where some are expected to perform optimally and where no advantage is anticipated. The protocols considered are as follows:
\begin{enumerate}
    \item \textbf{No QREM:} No readout error mitigation is applied.
    \item \textbf{Factorized QREM:} The REMST protocol in which each qubit is treated individually. The reconstructed noisy POVMs are tensor products of single-qubit POVMs. 
    \item \textbf{Two-qubit QREM:} The REMST protocol is applied to the two qubits that are part of the observable. The correlated readout errors are only mitigated for the two qubits participating in the observable. 
    \item \textbf{Classical correlated QREM:} The correlated QREM protocol is applied where only classical errors are corrected. The reconstructed noisy POVMs are produced using the full correlated QREM protocol, but discarding all off-diagonal terms in the POVMs.  
    \item \textbf{Correlated QREM:} The complete correlated QREM protocol is applied. 
\end{enumerate}
Additional details on the different protocols can be found in Appendix~\ref{app:details_on_other_protocols}.

All protocols were evaluated for $N_s = 10$ random realizations of the 16 qubits extracting correlators between $N_p = 20$ random pairs of qubits. 
The readout noise was sampled at random from a large set of experimentally extracted POVMs, see Sec.~\ref{Sec:Experimentally_extracted_POVMs}.  In Fig.~\ref{fig:16_qubit_example} the MSE is plotted as an average over all states, pairs, and observables. For a fair comparison, each protocol received the same measurements in both noise calibration, when applicable, and readout error-mitigated state reconstruction. A 3-local PHF was utilized for calibration, resulting in a total of $9 \times 10^6$ calibration measurements performed. For each requested observable, $\approx 10^5$ measurements were performed.
The correlated QREM protocol performed the best overall, indicating that both beyond-classical readout errors and correlated readout errors were present in the experimentally extracted POVMs. To better understand how the protocols work, specific pairs of qubits are examined in Fig.~\ref{fig:16_qubit_correlator}, where both state reconstruction infidelity and MSE are averaged only over randomly chosen target states. The $[0,1]$ pair featured correlated noise only between the qubits constituting the observable. Therefore, the factorized QREM protocol performed worse than the two-point and correlated QREM protocols, which worked equally well. The classical correlated QREM protocol performs slightly worse, which is expected, as small amounts of nonclassical errors are present in the experimental POVMs. For the pair $[4,8]$, there was no correlated readout noise between the two qubits appearing in the observable, but significant correlated readout noise between the adjacent qubits, specifically for $(4,5)$ and $(8,9,10)$. It is therefore expected that the two-point QREM protocol does not perform too well while both correlated QREM protocols perform well. For the two remaining cases, $[6,11]$ and $[7,13]$, we observe similar performance between all QREM protocols, indicating that there were no significant correlated readout errors.


\subsection{Strong artificial noise}
\label{sec:strong_artifical_noise}
In order to illustrate the applicability of the protocol to correlated and beyond-classical noise  sources (property \textbf{B}), two artificial noise models were examined in a 16-qubit simulation: one characterized by significant correlated errors and the other by pronounced coherent errors. Both MSE and state reconstruction infidelity were averaged over $N_s = 10$ randomly chosen target states and $N_p = 20$ random pairs. The number of calibration measurements and state reconstruction measurements is the same as in Sec.~\ref{sec:16-qubit_example}.  


\subsubsection{Strong correlated noise}
\label{sec:strong_correlation}
\begin{figure}
 \centering
 \includegraphics[width=0.95\linewidth]{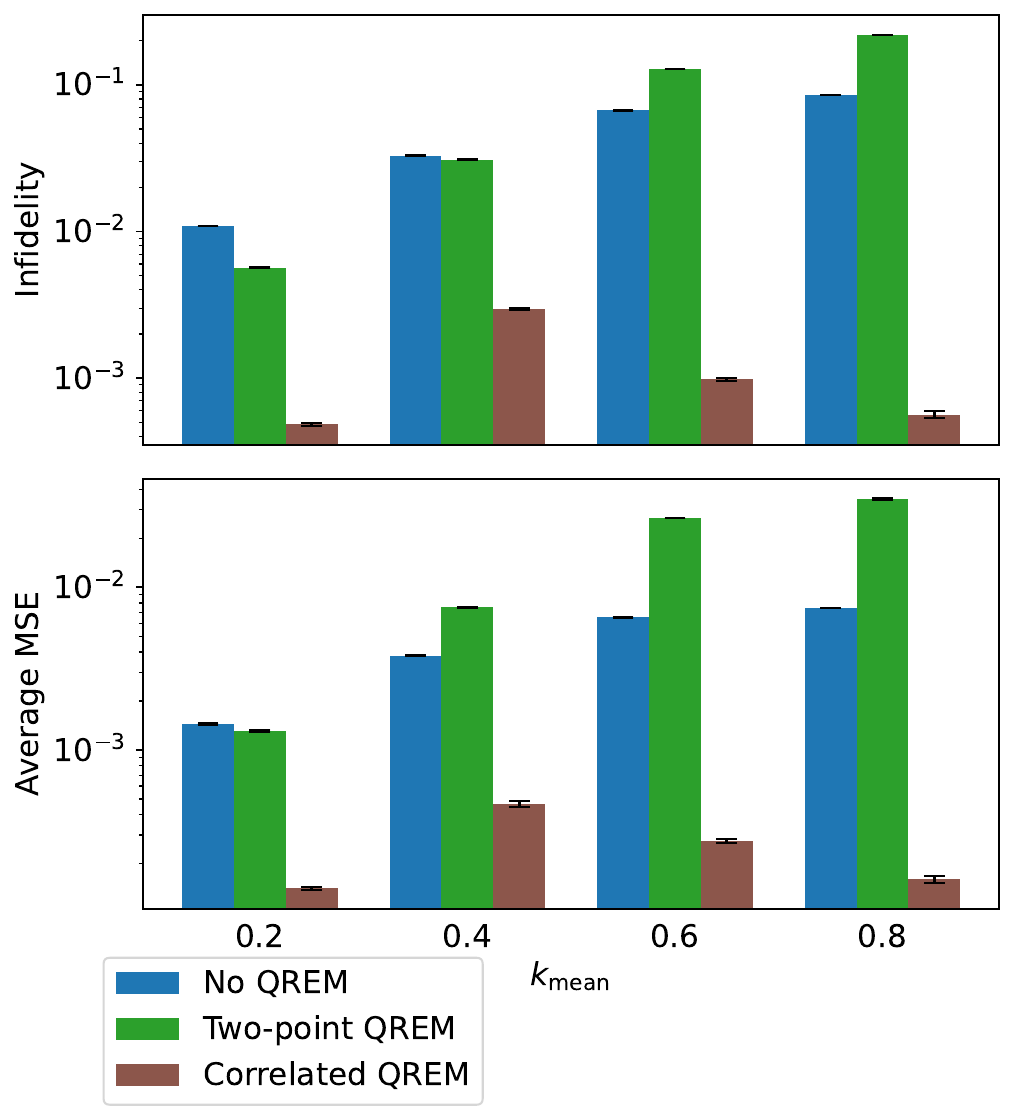}
 \caption{Comparing protocol performance for increasing correlated noise strength. Readout error-mitigated infidelity and averaged mean-squared error for all two-qubit Pauli observables are plotted for the two-point QREM and the correlated QREM protocol. The error bars are the standard deviation from 10 repetitions of the simulation with identical initial conditions.}
 \label{fig:iSWAP_example}
\end{figure}

Correlated noise can be generated by applying a noisy entangling gate between two qubits. The iSWAP gate is the native entangling gate for superconducting qubits and generates symmetric correlation between the two qubits, unlike the CNOT gate, which is directional. The simulated correlated noise is modeled as a probabilistic iSWAP channel
\begin{equation}
    \mathcal{E}(M_i) = (1-k) M_i + k U_{\text{iSWAP}}M_i U_{\text{iSWAP}}^\dagger,
\end{equation}
which can be understood as preserving the quantum state with probability $k$, while with probability $1-k$ an iSWAP gate is applied to the state, correlating the qubits. In the simulation, only neighboring pairs of qubits were correlated, i.e.,~pair $\{(0,1), (2,3), \dots ,(14, 15)\}$. For each pair of qubits, the noise strength $k$ was randomly drawn from a uniform distribution around a mean value $k\in [k_{\text{mean}} -0.1 ,k_{\text{mean}} +0.1]$. 

In Fig.~\ref{fig:iSWAP_example} correlated QREM is compared to no QREM and two-point QREM for different mean noise strengths $k_\text{mean}$. The two-point QREM protocol has weaker assumptions than the factorized QREM protocol, and is a stronger candidate to compare to the correlated QREM protocol. The two-point QREM protocol accounts for correlations between the qubits in the two-point observable, but not other qubits.

With increasing noise strength, both the two-point protocol and the no-mitigation protocol performed significantly worse than the correlated QREM protocol. In particular, the two-point QREM protocol yielded marginally poorer results compared to when the no-mitigation protocol. This can be attributed to the fact that most random pairs of qubits involved qubits that belonged to distinct noise clusters. By overlooking correlated noise, the effective environment for these qubits exhibits non-Markovian behavior, causing a foundational assumption of the two-point protocol to break down, rendering it ineffective. When the mean noise strength is at $k_{\text{mean}}=0.8$, the averaged MSE and the infidelity decrease for the correlated QREM protocol. This is likely because the reconstructed POVM is very close to being a swapping of the two qubits, which is a simpler error to mitigate than a more equal mixture of swapped and not swapped measurements. We note that the correlation coefficients generated by this model with $k_\text{mean} = 0.4$ are comparable to the highest correlation coefficients observed in state-of-the-art superconducting qubit systems \cite{Tuziemski2023}, see Fig.~\ref{fig:dendrogram_tuning} in Appendix~\ref{App:Hierarchical_clustering} for more details.

\subsubsection{Strong coherent noise}
\label{sec:strong_coherent}

\begin{figure}
 \centering
 \includegraphics[width=0.95\linewidth]{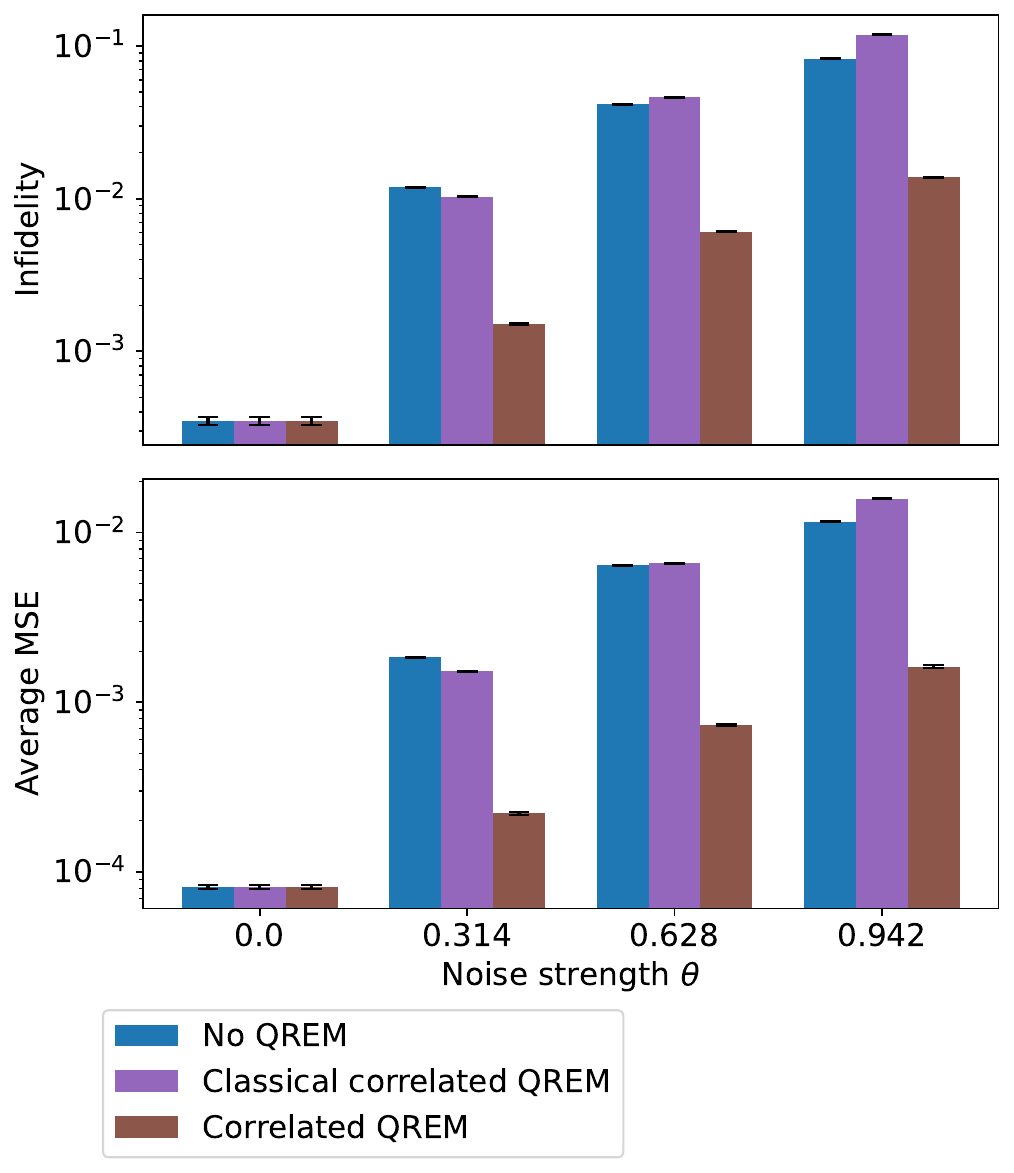}
 \caption{Comparing protocol performance for increasing coherent noise. 
 Readout error-mitigated infidelity and averaged mean-squared error for all two-qubit Pauli observables are plotted for the classical correlated QREM and the correlated QREM protocol. The error bars are the standard deviation from 10 repetitions of the simulation with identical initial conditions.}
 \label{fig:coherent_example}
\end{figure}
Coherent errors can occur when unitary gates are applied incorrectly. As an example for superconducting qubits, single-qubit unitary rotations are applied by sending a resonant driving pulse for a specified time $t$. Inaccuracies in how long this pulse is applied cause either underrotation or overrotation of the qubit. A combination of correlated qubit readout and coherent errors can be modeled by single-qubit overrotations and neighboring two-qubit $XX$ interactions \cite{Chen2021}. Such a noise model can be formulated as a $n-$qubit channel with noise strength $\theta$,
\begin{equation}
\label{eq:coherent_error}
    \mathcal{E}^n_\theta(M_i) = R^n_{x}(\theta)M_iR^{n \dagger}_{x}(\theta),
\end{equation}
where
\begin{equation}
    R^n_x(\theta) =
    \begin{cases}
        \exp(-\frac{i \theta}{2} \sum_{\langle j,k\rangle} X^j\otimes X^k ),& \text{when } n> 1\\
        \exp(-\frac{i \theta}{2} X), &\text{when } n=1.
    \end{cases}
\end{equation}
The 16-qubit correlation structure was reused from the example in Fig.~\ref{fig:16_qubit_example} (inset to the bottom right). For each noise cluster, the channel in Eq.~\ref{eq:coherent_error} was applied with the same noise strength $\theta$ for all clusters. Although this noise model creates a four-qubit noise cluster, we continue to limit the maximum size of the cluster to $n_\text{corr}=3$. This means that the correlated QREM protocol will not be able to capture all correlations in the system.

In Fig.~\ref{fig:coherent_example} both correlation-conscious versions of the protocol are compared to no QREM. Classical correlated QREM only considered a statistical redistribution of the outcomes within the same cluster and ignored any off-diagonal entries in the POVM reconstruction, which is equivalent to using a confusion matrix to mitigate readout errors \cite{Maciejewski2020}. Both classical correlated and no QREM performed similar across the board, while correlated QREM achieved about an order of magnitude better results. There is a clear trend for the correlated QREM protocol upwards, which is due to the four-qubit correlated noise that is not captured with the maximum cluster size $n_\text{corr}=3$. Given that this noise model is strong and unrealistic and is not anticipated in any plausible experiments, the performance of the protocol is considered satisfactory, as it still surpasses all other alternatives. The readout errors associated with a noise strength $\theta = 0.314$ are similar to the coherent errors observed between the qubit pair $[0, 1]$ shown in Fig.~\ref{fig:16_qubit_correlator}, when comparing classical correlated QREM to correlated QREM.

For this simulation a different hierarchical clustering method was used than for all the other simulations due to an artifact in how the correlation coefficients are computed. This is discussed in detail in Appendix~\ref{app:Limitation_coherent_error}.
In Appendix~\ref{app:addional_results_coherent} a simulation is shown with the same clustering method as the other simulations.

\subsection{100-qubit readout error mitigation with experimental noise}
\label{sec:100_qubit_example}

\begin{figure}
 \centering
 \includegraphics[width=\linewidth]{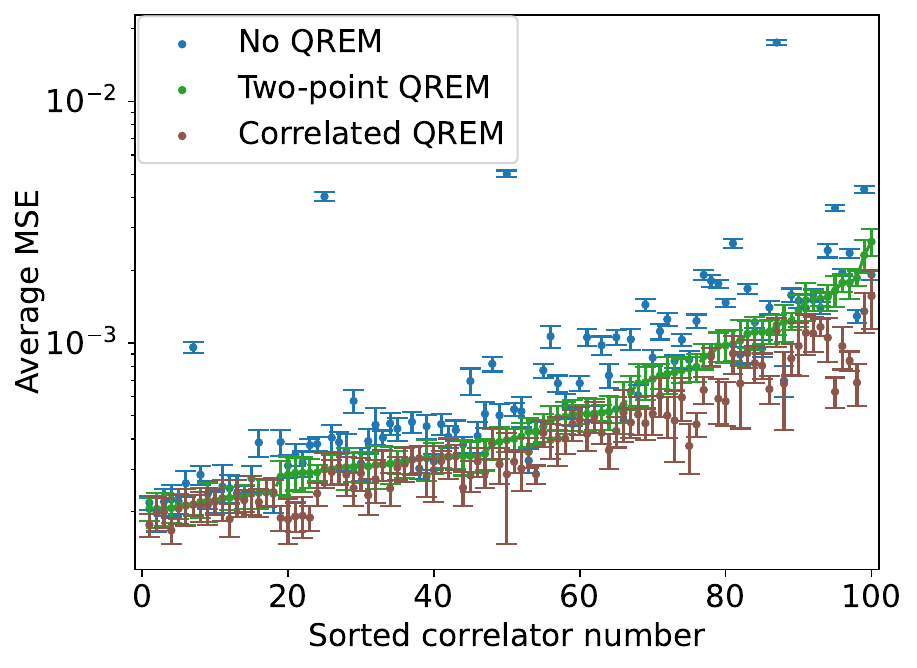}
 \caption{100-qubit example with POVMs extracted from superconducting qubits. The correlated QREM protocol is compared to the two-point QREM protocol. Random pairs are selected and the mean-squared error is computed over all two-qubit Pauli observables for each pair. The pairs are sorted by increasing MSE for the two-point QREM protocol. The error bars are the standard deviation from 10 repetitions of the simulation with identical initial conditions.}
 \label{fig:100_qubit_example}
\end{figure}

To demonstrate scalability (property \textbf{A}), $N_s = 10$ random 100-qubit simulations were performed with readout noise sampled from experimentally extracted POVMs. The calibration step used a total of $11.2 \times 10^6$ measurements, and $\approx 5\times10^4$ measurements were performed for each requested observable. The correlated QREM protocol is compared to the two-point QREM protocol and no QREM. Fig.~\ref{fig:100_qubit_example} shows $N_p = 100$ randomly selected pairs of qubits sorted by increasing two-point QREM averaged MSE. The correlated QREM model consistently shows performance that is on par with or superior to the two-point model across the data set. A more distinct separation between the three methods appears in the more correlated pairs on the right side, highlighting the respective utility of the three protocols as no mitigation, partial mitigation, and full mitigation. For a system of this size, perfect clustering is unattainable, as detailed in Appendix~\ref{App:Hierarchical_clustering}, leading to potential inaccuracies in identifying the appropriate noise cluster for some observables. Nonetheless, the correlated protocol consistently matches or exceeds the performance of other protocols. In Appendix~\ref{app:additional_results} we include two additional 100-qubit simulations, one with mixed target states and the other with single-qubit observables included in the average MSE.

\section{Conclusion and outlook}
We have presented a scalable readout error-mitigation protocol that takes into account the correlated readout structure. Correlated readout errors are efficiently characterized by combining overlapping detector tomography and hierarchical clustering to group the most correlated qubits into noise clusters.
For a desired few-body observable, the involved noise clusters are combined to a connected noise cluster, on which the REMST protocol \cite{Aasen2024} is executed to reconstruct observables with mitigated readout errors. We demonstrate the effectiveness and generality of our protocol by comparing it with previously proposed QREM protocols. Simulations of up to 100 qubits were performed with both experimentally obtained readout noise and simulated correlated readout noise. In every scenario tested, the correlated QREM protocol shows a performance that is, at worst, on par with all competing protocols, but in most cases it substantially exceeds them in terms of reduction of the mean-squared errors of the averaged observables.

Although the protocol's sample complexity is logarithmic in the number of qubits, it is not guaranteed to be optimal due to challenges in identifying minimal measurement strategies. Finding a better parallel measurement scheme would allow for more efficient characterization of correlated readout, and also a more practical ``measure first, ask questions later" approach, potentially by using covering arrays \cite{Kiara2024}. The correlation coefficients perform well as an indicator for correlated readout, but have certain drawbacks. When the strength of correlations within the system is large, it can assign correlation coefficients to uncorrelated qubits that are of the same order of magnitude as other weakly correlated qubits. These incorrect correlation coefficients can cause weakly correlated qubits to be incorrectly assigned to other groups. Although these incorrect assignments lower the average performance of the correlated QREM protocol, it still shows a significant improvement over other candidates. Finding better approaches to quantify the correlated readout would improve the correlated QREM protocol and is a potential future direction. 

Looking forward, we will further explore the application of this protocol to large-scale experimental setups. Notably, in tasks such as evaluating correlation propagation in large quantum spin systems, we expect our protocol to yield a precision gain leading to new insights on nonequilibrium quantum dynamics.



\section*{Code and data availability}
The code used to generate and analyze the presented results can be found, together with a tutorial notebook, in the GitHub repository available at: \url{https://github.com/AdrianAasen/CREMST}. The data used for this project are available in the code repository with a notebook used to display the results. 

\section*{Author Contributions}
A.A. developed the protocol and software with the supervision of M.G. A.D., H.R. and A.U. provided the experimental characterization of the measurement device used for sample generation.  A.A. wrote the manuscript with contributions from A.D and M.G. All
authors contributed to the finalization of the manuscript.

\section*{Acknowledgments}
A.A. thanks K. Hansenne for fruitful discussions on overlapping tomography and covering arrays, N. Euler for suggesting a cluster verification approach, R. Dougherty for help with generating the perfect hash families, and M. D'Achille for technical assistance. Some python packages were instrumental in this project, Numpy \cite{harris2020}, Scipy \cite{Virtanen2020}, joblib \cite{joblib}, Matplotlib \cite{Hunter2007} and Quantikz \cite{Kay2018}. The authors acknowledge support by the state of Baden-Württemberg through bwHPC.
This research is supported by funding from the German Research Foundation (DFG) under the project identifier Grant No. 398816777-SFB 1375 (NOA).

 \appendix



\section{Limitation of the correlation coefficients}
\label{App:Limitation_correlation_coeff}
Although the correlation coefficients provide a useful proxy for readout-induced correlations, they are not a strict correlation measure. They can and do assign non-trivial correlation coefficients to uncorrelated qubits. The origin of these unphysical correlations lies in the definition of the correlation coefficients. As emphasized in Sec.~\ref{sec:correlation_coefficients}, defining a reduced POVM from a larger subsystem in the presence of correlated readout is a complex task. The starting point for the correlation coefficients are two-qubit POVMs which are directly reconstructed from measurements. One is therefore implicitly tracing out all the other qubits in the system. Recalling the definition of traced-out POVMs from Eq.~\eqref{eq:traced_out_POVM}, constructing reduced POVMs directly assumes the reference state as the maximally mixed state $\rho_\text{B} =  \mathbb{1}/2^{N-2}$. This is because the calibration states on the traced subsystem are an equal mixture of all possible symmetric and informationally complete calibration states, see Eq.~\eqref{eq:QDT_calibration_states}, leading to reduced POVM,

\begin{equation}
    M^{\text{A}}_{i_\text{A}} = \frac{1}{2^{N-2}}\sum_{i_\text{B}} \Tr_\text{B}(M^{\text{AB}}_{i_\text{A}i_\text{B}}),
    \label{eq:traced_out_POVM_with_identity}
\end{equation}
where $N$ is total number of qubits in the system, and the reference state label is omitted on the reduced POVM.

This choice ignores the potential correlations that occur between the A and B subsystems. Therefore, if there were substantial correlations between these subsystems, the residual discrepancy would be attributed to the correlation between two qubits in subsystem A. This becomes apparent if there are groups of strongly correlated qubits and other groups of weakly correlated qubits which are mutually uncorrelated. The correlation coefficients between a qubit from the weakly correlated group and a qubit from the strongly correlated group can be larger than the correlation coefficients between two weakly correlated qubits. 
In effect, this causes an unphysical correlation to be assigned to pairs of qubits which may not share any correlation.



\section{Perfect hash families}
\label{App:perfect_hash_families}
The PHF $\Phi_{15,16,3}$ used for the 16-qubit simulations is presented in Table~\ref{tab:table_16_qubit}. In Ref.~\cite{Aasen2025} we provide a set of other PHF, specifically $\Phi_{36,16,4}, \Phi_{37,100,3}$, and $\Phi_{87,50,4}$.

\newcolumntype{C}{>{\centering\arraybackslash}X} 
\begin{table}[H]
    \centering
\begin{tabularx}{\textwidth/2}{l|*{16}{C}}
        \toprule
& \multicolumn{16}{c@{}}{$k$}\\
$N$  & 1 & 2 & 3 & 4 & 5 & 6 & 7 & 8 & 9 & 10 & 11 & 12 & 13 & 14 & 15 & 16 \\ 
\hline
1 & 0 & 1 & 2 & 0 & 1 & 2 & 0 & 1 & 2 & 0 & 1 & 2 & 0 & 1 & 2 & 0 \\ 
2 & 0 & 1 & 2 & 2 & 1 & 0 & 1 & 0 & 2 & 2 & 1 & 0 & 1 & 2 & 0 & 1 \\ 
3 & 0 & 1 & 2 & 2 & 0 & 1 & 2 & 1 & 0 & 1 & 2 & 0 & 0 & 2 & 1 & 1 \\ 
4 & 0 & 1 & 2 & 2 & 0 & 1 & 2 & 0 & 1 & 2 & 0 & 2 & 1 & 0 & 2 & 0 \\ 
5 & 0 & 1 & 2 & 1 & 0 & 2 & 0 & 2 & 1 & 2 & 1 & 0 & 1 & 0 & 1 & 2 \\ 
6 & 0 & 1 & 2 & 1 & 0 & 2 & 0 & 2 & 1 & 0 & 2 & 1 & 2 & 1 & 0 & 1 \\ 
7 & 0 & 1 & 2 & 2 & 1 & 0 & 2 & 1 & 0 & 1 & 2 & 0 & 0 & 0 & 0 & 0 \\ 
8 & 0 & 1 & 2 & 2 & 1 & 0 & 2 & 2 & 1 & 0 & 0 & 1 & 2 & 1 & 2 & 0 \\ 
9 & 0 & 1 & 2 & 1 & 0 & 2 & 2 & 0 & 0 & 1 & 2 & 2 & 0 & 1 & 1 & 0 \\ 
10 & 0 & 1 & 2 & 2 & 1 & 0 & 1 & 2 & 0 & 1 & 2 & 1 & 1 & 2 & 2 & 2 \\ 
11 & 0 & 1 & 2 & 0 & 2 & 0 & 0 & 0 & 1 & 1 & 0 & 0 & 2 & 1 & 1 & 0 \\ 
12 & 0 & 0 & 1 & 2 & 1 & 1 & 1 & 2 & 0 & 0 & 1 & 1 & 2 & 2 & 1 & 1 \\ 
13 & 0 & 0 & 1 & 2 & 1 & 0 & 2 & 2 & 2 & 1 & 1 & 1 & 2 & 2 & 2 & 0 \\ 
14 & 0 & 0 & 0 & 0 & 1 & 1 & 1 & 2 & 2 & 2 & 2 & 2 & 2 & 2 & 0 & 0 \\ 
15 & 0 & 0 & 0 & 0 & 0 & 0 & 0 & 0 & 1 & 1 & 1 & 1 & 1 & 2 & 2 & 2 \\
\bottomrule
\end{tabularx}
\caption{Perfect hash family $\Phi_{15,16,3}$. This perfect hash family was used to generate all samples used for the QDT step in Secs.~\ref{sec:16-qubit_example} and \ref{sec:strong_artifical_noise}. It was generated by using the density algorithm presented in Ref.~\cite{dougherty2019}.}
\label{tab:table_16_qubit}
\end{table}

\section{Hierarchical clustering algorithm and limitations}
\label{App:Hierarchical_clustering}

In this work, we used hierarchical agglomerative clustering algorithms to group correlated qubits. We will only give a brief and qualitative description of the clustering algorithm, for more details see Ref.~\cite{Muller2011}. In simple terms, a hierarchical agglomerative clustering algorithm uses the mutual distance between all data points to procedurally combine the data points with the smallest pairwise distance into clusters. This process continues until a singular cluster is achieved that includes all data points. For our purposes the data points are the qubits and the mutual distances are derived from the symmetric correlation coefficients, see Sec.~\ref{sec:correlation_coefficients}, and are defined as
\begin{equation}
    d^2_{ij} = 1 - c_{i \leftrightarrow j}.
    \label{eq:clustering_distance}
\end{equation}
This gives a distance where $d=1$ means completely uncorrelated and $d=0$ means maximally correlated. To limit how large the clusters can grow, a threshold value $\gamma$ can be chosen such that no clusters with intra-cluster distance larger than $\gamma$ will be grouped together. The threshold value is upper bounded by 1 and should be set at least below the expected statistical fluctuations $\gamma \leq 1- 1/\sqrt{n_\text{calib}}$, where $n_\text{calib}$ is the number of measurements used for each calibration state. If the readout noise is large, one may in addition manually tune this threshold if there are obvious separations of correlator scale. Examples of tuned and untuned dendrograms are shown in Fig.~\ref{fig:dendrogram_tuning}. A dendrogram is a visual representation of the clustering process. Read from the bottom upward, each qubit starts in their own node, indicated by each qubit index having it's own branch. Moving upward, the nodes are joined by horizontal lines at their characteristic intercluster distance. The threshold value indicated at what distance between clusters the clustering stops, and the noise clusters can be read off by following each branch downward. In Fig.~\ref{fig:dendrogram_tuning} each noise cluster is given a distinct color. If noise clusters exceed the maximum size $n_{\text{corr}}$ after applying the threshold, as shown in the top plot, they are iteratively split by following dendrogram branches downward to the first intersection point, dividing the cluster into two smaller ones. This does not affect the threshold for other clusters.

\begin{figure}
  \centering
  \includegraphics[width=0.95\linewidth]{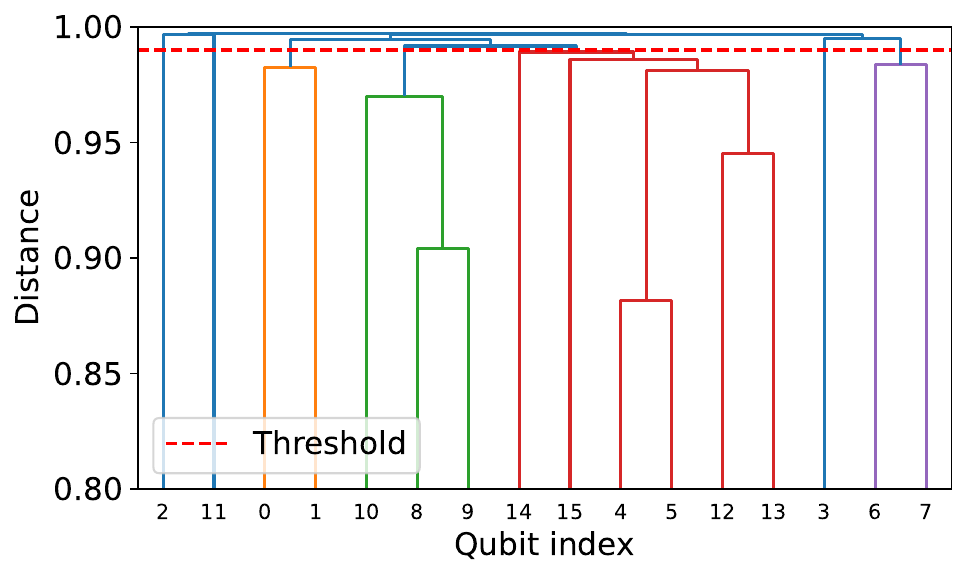}  \\
  \includegraphics[width=0.95\linewidth]{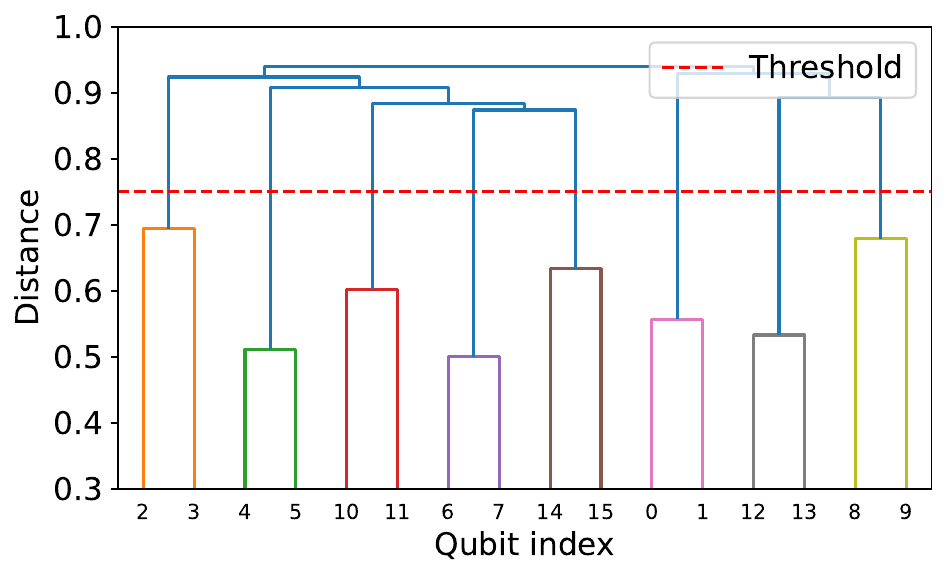} 
\caption{Examples of automatic and manual dendrogram tuning. Distance refers to $d^2$. \textbf{Top:} Dendrogram from Sec.~\ref{sec:16-qubit_example}. No tuning was performed and the automatic threshold $\gamma=1-1/\sqrt{n_\text{calib}}$ was used. The largest cluster was iteratively split into $\{(14), (15), (4,5), (12,13)\}$. \textbf{Bottom:} Dendrogram from Sec.~\ref{sec:strong_correlation} with $k=0.4$. Manual threshold was set due to clear separation of correlation structure. The characteristic value of the correlation coefficients for this noise model are comparable with the worst qubit correlations observed in Ref.~\cite{Tuziemski2023} (IBM Cusco $d^2 \approx 0.7$ and Rigetti Aspen-M-3 $d^2\approx 0.6$.).}
\label{fig:dendrogram_tuning}
\end{figure}

The numerical implementation used the SciPy hierarchy package with the ``complete" linkage method \cite{Virtanen2020}. The ``complete" linkage method groups together clusters by minimizing the largest distance between any of the qubits in the new joined cluster. In Sec.~\ref{sec:strong_coherent}, the ``average'' linkage was used instead due to a limitation of the correlation coefficients. The average linkage method instead groups together qubits based on the lowest average intercluster distances. We also found very good error-mitigation performance using the ``ward" method which generally produced dendrograms where it was easier to visually see the separation of correlation scale when correlations were large.  However, we advise against its use, as the Ward method is only guaranteed to work with Euclidean distances \cite{Muller2011}, while the distance described in Eq.~\eqref{eq:clustering_distance} is non-Euclidean.

\subsection{Limitation on the coherent error model} 
\label{app:Limitation_coherent_error}
The coherent error model used in Sec.~\ref{sec:strong_coherent} used the ``average'' linkage method, as opposed to the ``complete'' linkage method used for all other simulations. This choice stems from the way the noise model creates correlations among qubits within three- and four-qubit noise clusters. To demonstrate this problem, we will specifically analyze the case of three qubits.

The error unitary for a three-qubit system reads as
\begin{equation}
\begin{split}
        R^3_x(\theta) =& e^{-i (X\otimes X \otimes \mathbb{1}_2 + \mathbb{1}_2 \otimes X \otimes X) \theta/2}\\
        =& \cos^2\left(\tfrac{\theta}{2}\right) \mathbb{1}_8 + \sin^2\left(\tfrac{\theta}{2}\right) X \otimes \mathbb{1}_2 \otimes X \\
        & - i \cos\left(\tfrac{\theta}{2}\right) \sin\left(\tfrac{\theta}{2}\right) (X\otimes X \otimes \mathbb{1}_2 + \mathbb{1}_2 \otimes X \otimes X),
\end{split}
\end{equation}
where a matrix exponential expansion was performed. The second term, which correlates the first and third qubits, is suppressed by a factor $ \sin^2\left(\tfrac{\theta}{2}\right)$, which, at small $\theta$, makes the correlation coefficients for these two qubits small, while the correlation between the first and second and the second and third remain large. This causes issues using the ``complete'' linkage method, which optimizes the smallest largest distance between all qubits in the clusters that should be joined. Therefore, if two qubits have very weak correlations, they are very unlikely to be combined into a noise cluster, even though they are very correlated to a mutual qubit. This, combined with the drawbacks of the correlation coefficients discussed in Appendix~\ref{App:Limitation_correlation_coeff}, makes it very unlikely that all three qubits are clustered together. A solution to this was to change the linkage method to ``average'' which instead optimizes for the smallest average distance between qubits. 

\section{Necessity of connected noise clusters}
 \label{App:proving_obs_12}

This appendix demonstrates that the correlated QREM protocol, outlined in Sec.~\ref{sec:method}, is the minimal procedure required to mitigate both correlated and coherent readout errors. To begin with, it is convenient to introduce a practical representation for the readout error-mitigated quantum states. 
When it is necessary to denote the subsystem to which operators belong, superscripts with capital letters are used.
 
 \subsection{Representing the readout error-mitigated state}
A density matrix can be represented as a linear combination of any complete basis of Hermitian operators. An informationally complete POVM forms such a basis. Using the index sum notation, the density matrix can be decomposed as
 \begin{equation}
     \rho = c_i M_i,
     \label{eq:density_matix_representation}
 \end{equation}
 where $c_i$ is some complex coefficient.  An expression for $c_i$ can be found by using the probability of receiving outcome $i$ from a measurement,
\begin{equation}
    p_i = \Tr(M_i \rho) = \Tr(M_i M_j) c_j = T_{ij} c_j,
    \label{eq:transfer_matrix_def}
\end{equation}
where we have defined the overlap matrix $T_{ij} = \Tr(M_i M_j)$. Solving Eq.~\eqref{eq:transfer_matrix_def} for $c$ and inserting it into Eq.~\eqref{eq:density_matix_representation} one finds a linear inversion representation of the density matrix
\begin{equation}
\label{eq:density_matrix_POVM_rep}
    \rho = p_i \left(T^{-1}\right)_{ij} M_j.
\end{equation}
This representation is convenient as it directly relates the quantum state to the measurement operators and the measurement outcomes $p_i = n_i/N$, where $n_i$ is the recorded number outcome $i$, and $N$ is the total number of measurements. 

If $N$ measurements were performed by a noisy measurement device $\tilde{\mathbf{M}}$, a linear inversion estimator for the prepared state would be 
\begin{equation}
    \hat \rho = \frac{n_i}{N}\left(\tilde{T}^{-1}\right)_{ij} \tilde{M}_j.
    \label{eq:error_mitigated_state}
\end{equation}
Therefore, if the noisy POVM $\tilde{\mathbf{M}}$ is known, Eq.~\eqref{eq:error_mitigated_state} is an estimator of the readout error-mitigated state. 
In the limit of infinite measurements, linear inversion aligns with both the maximum likelihood estimator and the Bayesian mean estimator, which are used in our protocols. Therefore, in the remainder of this section, we assume $n_i/N \rightarrow p_i = \Tr(\rho \tilde{M}_i)$.



\subsection{Minimality of correlated QREM protocol}

The necessity of the correlated QREM protocol can be phrased as an observation, 
\begin{restatable}{observation}{entangled}
    If no assumptions about the prepared state can be made, then readout error mitigation must be performed on the connected noise cluster. 
    \label{theorem:entangled_QREM}
\end{restatable}
\noindent This observation states that an error-mitigated state can generally not be reduced beyond the connected noise clusters of an observable without losing accuracy. We will show this by considering a simplified three-qubit case illustrated in Fig.~\ref{fig:entangled_state_sketch}. The noisy correlated readout is described by two error channels $\mathcal{E}_1$ and $\mathcal{E}_2$ which operate on a one-qubit and two-qubit Hilbert space, respectively, and the prepared state is a general three-qubit state.  
The process can be written as follows,
\begin{equation}
\begin{split}
        \ket{000} \rightarrow & \rho^{\text{ABC}} =  U\ket{000}\bra{000}  U^\dagger \\
        \rightarrow  & \mathcal{E}_1\otimes\mathcal{E}_2(\rho^{\text{ABC}}).
\end{split}
\end{equation}

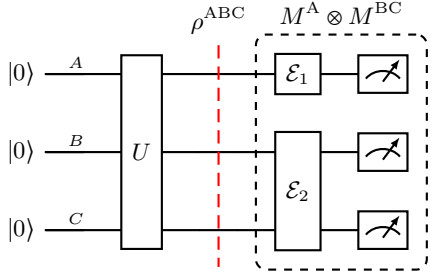
\begin{figure}
    \centering
    \begin{quantikz}
        \lstick{$\ket{0}$}&\wire[l][1]["A"{above,pos=0.2}]{a}&\gate[3]{U}&\slice{$\rho^{\text{ABC}}$}&&\gate{\mathcal{E}_1}\gategroup[3, steps=2,style={dashed,rounded
    corners}]{$M^\text{A}\otimes M^\text{BC}$}&\meter{}\\
        \lstick{$\ket{0}$}&\wire[l][1]["B"{above,pos=0.2}]{a}&&&&\gate[2]{\mathcal{E}_2}&\meter{}\\
        \lstick{$\ket{0}$}&\wire[l][1]["C"{above,pos=0.2}]{a}&&&&&\meter{}
    \end{quantikz}
    \caption{Example quantum circuit with a prepared general three-qubit sate between correlated qubit readout. }
     \label{fig:entangled_state_sketch}
\end{figure}


We want to show that if a state is sampled by a noisy POVM $\mathbf{M} = \{M^{\text{A}} \otimes M^{\text{BC}}\}$ (described by $\mathcal{E}_1$ and  $\mathcal{E}_2$), we will not obtain the correct readout error-mitigated state $\rho^{\text{AB}}$ if we do not first reconstruct the full state $\rho^{\text{ABC}}$, which spans the full connected noise cluster. We do this by comparing the reconstruction of $\Tr_C(\rho^{\text{ABC}})$ to a reconstruction of $\rho^{\text{AB}}$ that is aware of the measurement outcomes on the C subsystem.

It is useful to define a reduced measurement operator which retains the information about the outcome on the traced-out subsystem as 
\begin{equation}
 \Tr_\text{B}(M^\text{AB}_{ij}) = M^{\text{A}|\text{B}}_{ij},
\end{equation}
which is the measurement operator for outcome $i$ on subsystem A conditioned on outcome $j$ occurring on subsystem B.

These conditioned measurement operators occur naturally when the C subsystem is traced out from the fully reconstructed state on ABC from Fig.~\ref{fig:entangled_state_sketch}. To see this consider the full $\rho^{\text{ABC}}$,

 \begin{equation}
\label{eq:error_mititgated_density_matrix}
    \rho^{\text{ABC}} = p^\text{ABC}_{ijk}  \left({T^\text{ABC}}^{-1} \right)_{ijkmno} M^{\text{A}}_m \otimes M^\text{BC}_{no},
\end{equation}
 and then trace out the subsystem C,
\begin{equation}
\begin{split}
       \Tr_\text{C}(\rho^{\text{ABC}}) =& p^\text{ABC}_{ijk}  \left({T^\text{ABC}}^{-1} \right)_{ijkmno} M^{\text{A}}_m \otimes \Tr_\text{C}(M^{\text{BC}}_{no})\\
       =& p^\text{ABC}_{ijk}  \left({T^\text{ABC}}^{-1} \right)_{ijkmno} M^{\text{A}}_m \otimes M^{\text{B}|\text{C}}_{no}.
\label{eq:traced_down_full_QREM_state}
\end{split}
\end{equation}



If one instead starts directly with the conditioned measurement operators $\mathbf{M} = \{M^{\text{A}}\otimes M^{\text{B}|\text{C}} \}$, which primarily operates in the AB subsystem while keeping track of the outcomes that occurred in the environment, one has
\begin{equation}
    \rho^{\text{AB}|\text{C}} = p^\text{ABC}_{ijk}  \left({T^{\text{AB}|\text{C}}}^{-1} \right)_{ijkmno} M^{\text{A}}_m \otimes M^{\text{B}|\text{C}}_{no}.
    \label{eq:POVM_reduced_state}
\end{equation}
The only difference between Eqs.~\eqref{eq:traced_down_full_QREM_state} and \eqref{eq:POVM_reduced_state} is the inverse overlap matrix, which in general
\begin{equation}
    \begin{split}
        T^{\text{AB}|\text{C}}_{ijkmno} =& \Tr_\text{AB}\left((M^\text{A}_i \otimes M^{\text{B}|\text{C}}_{jk}) (M^\text{A}_m \otimes M^{\text{B}|\text{C}}_{no})\right)\\ \neq & T^{\text{ABC}}_{ijkmno} = \Tr_\text{ABC}\left((M^\text{A}_i \otimes M^{\text{BC}}_{jk}) (M^\text{A}_m \otimes M^{\text{BC}}_{no})\right)
    \end{split}
\end{equation} 
because, for general matrices $Q$ and $R$
\begin{equation}
    \Tr_{\text{AB}}(QR) \neq \Tr_{\text{A}}(\Tr_\text{B}(Q)\Tr_\text{B}(R)).
\end{equation}
Since the matrices are not equal, their inverse cannot be equal. 
Therefore, one cannot recover the reduced error-mitigated state if one does not consider the full state between the two error clusters.

Observation~\ref{theorem:entangled_QREM} can be understood as a worst-case state reconstruction requirement to get any error-mitigated observable. That is, if the largest cluster size allowed is $n_\text{corr}$ qubits, then the largest state that might need to be reconstructed for a two-qubit observables is a $2n_\text{corr}$-qubit state.

\subsection{Product state assumption}
Significant simplifications are possible if some structure is known about the prepared state. In particular, if the prepared state is a product between noise clusters, readout error mitigation can be performed on the noise clusters separately, which means that no connected noise clusters are necessary.

To show this, we need the following property of the overlap matrix. 
\begin{prop}
If a POVM is a product of two smaller POVMs, $\mathbf{M} = \{M^\text{A} \otimes M^\text{B}\}$, then
\begin{equation}
     \left({T^\text{AB}}^{-1}\right)_{ijmn} =  \left({T^\text{A}}^{-1}\right)_{im}  \left({T^\text{B}}^{-1}\right)_{jn}.
     \label{eq:separability_of_transfer_matrix}
\end{equation}
\end{prop}
\begin{proof}

The overlap matrix for a product POVM can be written out explicitly,
\begin{equation}
\begin{split}
T^{\text{AB}}_{ijmn} =& \Tr\left((M^\text{A}_{i} \otimes  M^\text{B}_{j}) (M^\text{A}_{m} \otimes  M^\text{B}_{n})\right)\\
=&\Tr(M_i^\text{A}M_m^\text{A})\Tr(M_j^\text{B}M_n^{\text{B}})\\
=&T_{im}^\text{A}T_{jn}^\text{B}.
\end{split}
\end{equation}
From this one can then directly conclude that 
\begin{equation}
   T^{\text{AB}}_{ijmn}\left({T^\text{A}}^{-1}\right)_{mk}  \left({T^\text{B}}^{-1}\right)_{nl} = \mathbb{1}_{ijkl}.
\end{equation}
\end{proof}



\begin{figure}
    \centering
    \begin{quantikz}
        \lstick{$\ket{0}$}&\wire[l][1]["A"{above,pos=0.2}]{a}&\gate{U_1}&\slice{$\rho^{\text{ABC}}$}&&\gate{\mathcal{E}_1}\gategroup[3, steps=2,style={dashed,rounded
    corners}]{$M^\text{A}\otimes M^\text{BC}$}&\meter{}\\
        \lstick{$\ket{0}$}&\wire[l][1]["B"{above,pos=0.2}]{a}&\gate[2]{U_2}&&&\gate[2]{\mathcal{E}_2}&\meter{}\\
        \lstick{$\ket{0}$}&\wire[l][1]["C"{above,pos=0.2}]{a}&&&&&\meter{}
    \end{quantikz}
    \caption{Example quantum circuit with a product state between correlated qubit readout. }
     \label{fig:factorized_state_sketch}
\end{figure}
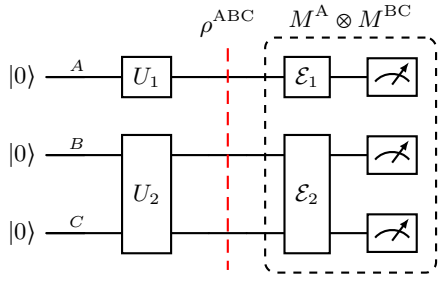

Consider the simplified three-qubit example in Fig.~\ref{fig:factorized_state_sketch}. It is known that the prepared state is a product state between the noise clusters. We want to show that reconstructing the states for each noise cluster individually, $\rho^\text{A} \otimes \rho^\text{BC}$, is equivalent to reconstructing the whole $\rho^\text{ABC}$. 

The two readout error-mitigated states can be written as follows,
\begin{equation}
\begin{split}
    \rho^{\text{A}} =& p_i^\text{A} \left({T^\text{A}}^{-1} \right)_{im} M_m^\text{A},\\
    \rho^{\text{BC}} =& p_{jk}^\text{BC} \left({T^\text{BC}}^{-1} \right)_{jkno} M_{no}^\text{A}.
\end{split}
\end{equation}
Using Eq.~\eqref{eq:separability_of_transfer_matrix}, we can combine the two 
\begin{equation}     
     \rho^\text{A} \otimes \rho^\text{BC} = p^\text{A}_{i}  p^\text{BC}_{jk}\left({T^\text{ABC}}^{-1} \right)_{ijkmno} M^{\text{A}}_m \otimes M^{\text{BC}}_{no}.
\end{equation}
Since the prepared state is a product state of subsystems A and BC, $p^\text{A}_{i}  p^\text{BC}_{jk} =  p^\text{ABC}_{ijk}$. Thus, $\rho^\text{A} \otimes \rho^\text{BC} = \rho^\text{ABC}$.

\section{Details on the protocol}
\label{App:protocol}
In this appendix, we provide implementation details on various parts of the protocol, including measurement simulations and POVM reconstruction. In addition, a pseudocode representation of the correlated QREM protocol in Fig.~\ref{fig:protocol} is provided.

\subsection{Simulating measurements}
The process of simulating the measurements draws inspiration from the way measurements are performed on superconducting qubit chips. The native measurement is the computational basis measurement. In order to achieve an informationally complete (IC) measurement for a single qubit, the qubit must be rotated into the $x$ basis and $y$ basis, as well as left undisturbed. This process effectively corresponds to the computational measurements performed along the $x, y,$ and $z$ axes, respectively. The natural way to simulate readout noise is to have it act only in the computational basis readout, and not include rotation gate errors. This closely resembles superconducting qubits, as individual qubit gates can be executed with greater precision than reading out the computational basis. Simulating measurements in this way simplifies the mitigation step, as performing detector tomography on the computational basis is sufficient. The mitigated IC POVM is created by applying ideal rotations to the mitigated computational basis POVM. Note that this does not mean that only classical errors are simulated, each basis is characterized and simulated with coherent errors.  

Due to noncommutativity of unitary operators, the application of noise channels such as the one in Sec.~\ref{sec:strong_coherent} to the computational basis does not yield the same IC POVM as applying the noise channel directly to an IC POVM such as the Pauli-6 POVM.

\subsection{State reconstruction}

Within the REMST protocol, quantum state tomography (QST) is essential to mitigate coherent errors. In the correlated QREM protocol, QST is performed for every connected noise cluster. The states are reconstructed using an iterative maximum-likelihood algorithm \cite{Lvovsky2004}. The likelihood function used for the connected noise clusters is defined as
\begin{equation}
    \mathcal{L}(\rho) \propto \Pi_i \Tr\left(\rho M^C_i\right)^{n_i},
\end{equation}
where $n_i$ is the observed number of outcomes $i$ and $\mathbf{M}^C$ is the reconstructed connected noise-cluster POVM.

The primary limitation affecting the protocol's numerical runtime is quantum state reconstruction. The maximum state size that can be potentially reconstructed is a $2n_\text{corr}$-qubit state, whereas the largest POVM involves only $n_\text{corr}$ qubits. Consequently, reconstructing the POVM is much quicker than reconstructing the state. Thus, to determine an adequate maximum cluster size, the runtimes of QST were examined. Table~\ref{tab:QST_runtimes} displays the numerical runtimes for state reconstruction in three to eight qubits. These runtimes were recorded on a 10-core laptop with a 2.5-GHz clock frequency. It can be concluded that setting $n_\text{corr}=3$ results in satisfactory runtimes suitable for the benchmarking performed in this study. In the existing implementation, choosing $n_\text{corr}=4$ is possible, albeit challenging. Substantial literature on more refined state reconstruction approaches, see e.g.~Refs.~\cite{Shang2017, Hou2016}, suggests these could considerably accelerate reconstruction times compared to the current setup; however, they probably would not permit much larger cluster sizes.

\begin{table}[H]
    \centering
    \begin{tabular}{c|c}
      Number of qubits   & Approximate runtime \\
      \hline
        3 &  1 s\\
        4 &  30 s\\
        5 & 2.5 min\\
        6 &  15 min\\
        7 & 30 min \\
        8 & 14 h
    \end{tabular}
    \caption{Recorded runtimes for connected cluster reconstructions for various numbers of qubits in the cluster. These measurements were conducted on a laptop with 10 cores, using iterative MLE \cite{Lvovsky2004}. }
    \label{tab:QST_runtimes}
\end{table}

\subsection{Detector tomography}
In this work we follow the prescription given by Ref.~\cite{Fiurek2001}.
Detector tomography for a $n$-qubit system is performed by preparing and measuring an informationally complete set of calibration states, $\rho_s$, $M$ times and solving the set of equations
\begin{equation}
    \frac{m_{si}}{M} = \Tr(\rho_s M_i),
\end{equation}
where $m_{si}$ is the number of times outcome $i$ was observed when measuring calibration state $\rho_s$. The calibration states are all possible combinations of tensor products of the single-qubit calibration states $\{\ket{\psi_i}\bra{\psi_i}\}_{i=1}^4$. Explicitly, we use pure states pointing to the corners of a tetrahedron in the Bloch sphere,
\begin{equation}
\label{eq:QDT_calibration_states}
\begin{split}
        \ket{\psi_1} &= \ket{0},\\
    \ket{\psi_2} &= \frac{1}{\sqrt{3}} \ket{0} + \sqrt{\frac{2}{3}} \ket{1},\\
    \ket{\psi_3} &= \frac{1}{\sqrt{3}} \ket{0} + \sqrt{\frac{2}{3}} e^{i \tfrac{2\pi}{3}}\ket{1},\\
    \ket{\psi_4} &= \frac{1}{\sqrt{3}} \ket{0} + \sqrt{\frac{2}{3}} e^{i \tfrac{4\pi}{3}}\ket{1}.\\
\end{split}    
\end{equation}
The four calibration states form a set of symmetric and informationally complete rank-1 projectors for a single qubit.

In order to have a fair comparison between all protocols, we have opted to use 3-local PHF to generate the calibration states. Using a 3-local PHF guarantees that any three-qubit POVM can be reconstructed, which is the largest used for the results presented in this work. As an example of how the calibration states are generated, Table~\ref{tab:Hashes_calib_example} shows the calibration states written explicitly for the first four qubits using the first hash function of the 3-local PHF in Table~\ref{tab:table_16_qubit}. Each hash function assigns one of three possible values to each qubit. An informationally complete set of calibration states is prepared between the three groups of qubits, where each qubit with label $(0,1,2)$ is assigned $\{(\psi_1,\psi_1,\psi_2), (\psi_1,\psi_1,\psi_3), (\psi_1,\psi_2,\psi_1), \dots , (\psi_3,\psi_3,\psi_2)\}$. The cases where all qubits receive the same calibration states are removed from each hash function and measured at the end.
\setlength{\tabcolsep}{4pt}

\begin{table}
    \centering
    \begin{tabular}{c|ccccc}
    \toprule
       Qubit index  & 1 & 2 & 3 & 4 \\
        \hline
        Hash value & 0 & 1 & 2 & 0\\
        \hline
        1 & $\ket{1}$ & $\ket{0}$ & $\ket{0}$ & $\ket{1}$ \\
        2 & $\ket{2}$ & $\ket{0}$ & $\ket{0}$ & $\ket{2}$ \\
        3 & $\ket{3}$ & $\ket{0}$ & $\ket{0}$ & $\ket{3}$ \\
        4 & $\ket{0}$ & $\ket{1}$ & $\ket{0}$ & $\ket{0}$ \\
        5 & $\ket{1}$ & $\ket{1}$ & $\ket{0}$ & $\ket{1}$ \\
        6 & $\ket{2}$ & $\ket{1}$ & $\ket{0}$ & $\ket{2}$ \\
        7 & $\ket{3}$ & $\ket{1}$ & $\ket{0}$ & $\ket{3}$ \\
        8 & $\ket{0}$ & $\ket{2}$ & $\ket{0}$ & $\ket{0}$ \\
        9 & $\ket{1}$ & $\ket{2}$ & $\ket{0}$ & $\ket{1}$ \\
        10 & $\ket{2}$ & $\ket{2}$ & $\ket{0}$ & $\ket{2}$ \\
    \vdots & \vdots & \vdots & \vdots & \vdots \\
        57 & $\ket{3}$ & $\ket{2}$ & $\ket{3}$ & $\ket{3}$ \\
        58 & $\ket{0}$ & $\ket{3}$ & $\ket{3}$ & $\ket{0}$ \\
        59 & $\ket{1}$ & $\ket{3}$ & $\ket{3}$ & $\ket{1}$ \\
        60 & $\ket{2}$ & $\ket{3}$ & $\ket{3}$ & $\ket{2}$ \\

        \bottomrule
    \end{tabular}
    \caption{Shortened list of the 60 calibration states generated from the first hash function in $\Phi_{15,16,3}$ PHF, focusing on the initial four qubits. The notation $\ket{i}$ is an abbreviation for the calibration state $\ket{\psi_i}$.}
    \label{tab:Hashes_calib_example}
\end{table}

\subsection{Details on the other protocols}
\label{app:details_on_other_protocols}
In Sec.~\ref{sec:results}, alongside the correlated QREM protocol, three alternative protocols were examined. We provide further details on how each of these protocols diverges from the full correlated QREM protocol. 

\subsubsection{Factorized QREM}
The facotrized QREM protocol only mitigated individual qubit readout errors. From the detector tomography measurements shared between all protocols, the noisy POVM is reconstructed for each qubit. For a requested two-qubit observable, the POVM for the two qubits entering the observable are tensored together and used as the reconstructed noisy POVM in the REMST protocol \cite{Aasen2024}. 

\subsubsection{Two-point QREM}
The two-point QREM protocol is a naive implementation of the REMST protocol \cite{Aasen2024}. For a given two-qubit observable, the two qubits entering the observable are isolated, and the noisy POVMs are reconstructed for the two-qubit system. In this way correlated readout between the two qubits is captured, but no other correlations the two qubits may have are mitigated. A drawback with using this method is that effects from other correlated qubits are incorrectly assumed to be due to correlations between the two qubits, which can cause the REMST protocol to perform worse than performing no QREM, see, for example, the right side of Fig.~\ref{fig:100_qubit_example}. 
\subsubsection{Classical correlated QREM}
The classical correlated QREM protocol operates exactly like the correlated QREM protocol, the only difference is that once the noise-cluster POVMs are reconstructed, the off-diagonal elements in the computational basis are set to zero. In this way, only statistical distributions of the different outcomes are captured, which is equivalent to using a confusion matrix \cite{Maciejewski2020}.

\subsection{Pseudocode for the correlated QREM protocol}
A pseudocode representation of the two primary components of the correlated QREM protocol is presented. Noise characterization is detailed in Algorithm~\ref{Alg:QREM_noise_characterization_pseudocode}, while observable reconstruction is described in Algorithm~\ref{Alg:QREM_observable_pseudocode}.

\begin{algorithm}[H]
  \caption{Correlated noise structure characterization}
  \textbf{Input:}\\
  \quad $\Phi$: 2-local perfect hash family\\
  \quad $\{\psi_1, \psi_2, \psi_3, \psi_4\}$: Single-qubit calibration states\\
  \textbf{Output:}\\
  \quad $s$: List of noise clusters\\
  \quad $\mathbf{M}$: List of reconstructed noise-cluster POVMs\\
  
  \begin{algorithmic}[1]
    \For{each hash function $\phi$ in $\Phi$}
      \For{each qubit $k$}
          \State  Assign label $\phi(k)$ to qubit $k$
        \EndFor
        \State  Assign calibration state $\{(\psi_1,\psi_2), (\psi_1,\psi_2) , \dots ,(\psi_4,\psi_3)\}$ to qubits with label $(1,2)$
        \State  Measure calibration states
    \EndFor
    \For{each calibration state $\psi_n$}
        \State Measure state $\psi_n$ on all qubits
    \EndFor
    \State Create list $\alpha$ of all possible qubit pairs
    \For{each pair $i$ in $\alpha$}
      \State Reconstruct POVM $\mathbf{M}_i$
      \State Compute correlation coefficient $c_i$ using $\mathbf{M}_i$
    \EndFor
    \State $s \gets$ hierarchical clustering using $c$
    \For{each noise cluster $j$ in $s$}
        \State Measure informationally complete set of calibration states on noise cluster $j$
      \State $\mathbf{M}_j\gets$ reconstruct POVM on noise cluster $j$ 
    \EndFor
  \end{algorithmic}
  \label{Alg:QREM_noise_characterization_pseudocode}
\end{algorithm}

\begin{algorithm}[H]
	\caption{Reconstructing readout error-mitigated observables} 
\textbf{Input}:\\ $O$: List of requested observables\\
$s$: List of noise clusters\\
$\mathbf{M}$: List of reconstructed noise-cluster POVMs \\
\textbf{Output}:\\
$O^{\text{QREM}}$: List of readout error-mitigated observables\\
$\rho^\text{QREM}$: List of readout error-mitigated states \\
\begin{algorithmic}[1]
    \For {each observable $O_i$ in $O$}
       \State Find connected noise cluster $C_i$ from $s$
       \State Create connected noise-cluster POVM $M^C_i$ using $\mathbf{M}$ and $C_i$
        \State $\rho^\text{QREM}_i \leftarrow$ Perform REMST \cite{Aasen2024} on $C_i$ using $M^C_i$ 
        \State  $O_i^{\text{QREM}} \leftarrow \Tr(\rho^\text{QREM}_i O_i)$
    \EndFor
\end{algorithmic}
\label{Alg:QREM_observable_pseudocode}
\end{algorithm}

\section{Additional results}
\label{app:additional_results}
\subsection{Coherent error model}
\label{app:addional_results_coherent}

\begin{figure}[H]
 \centering
 \includegraphics[width=\linewidth]{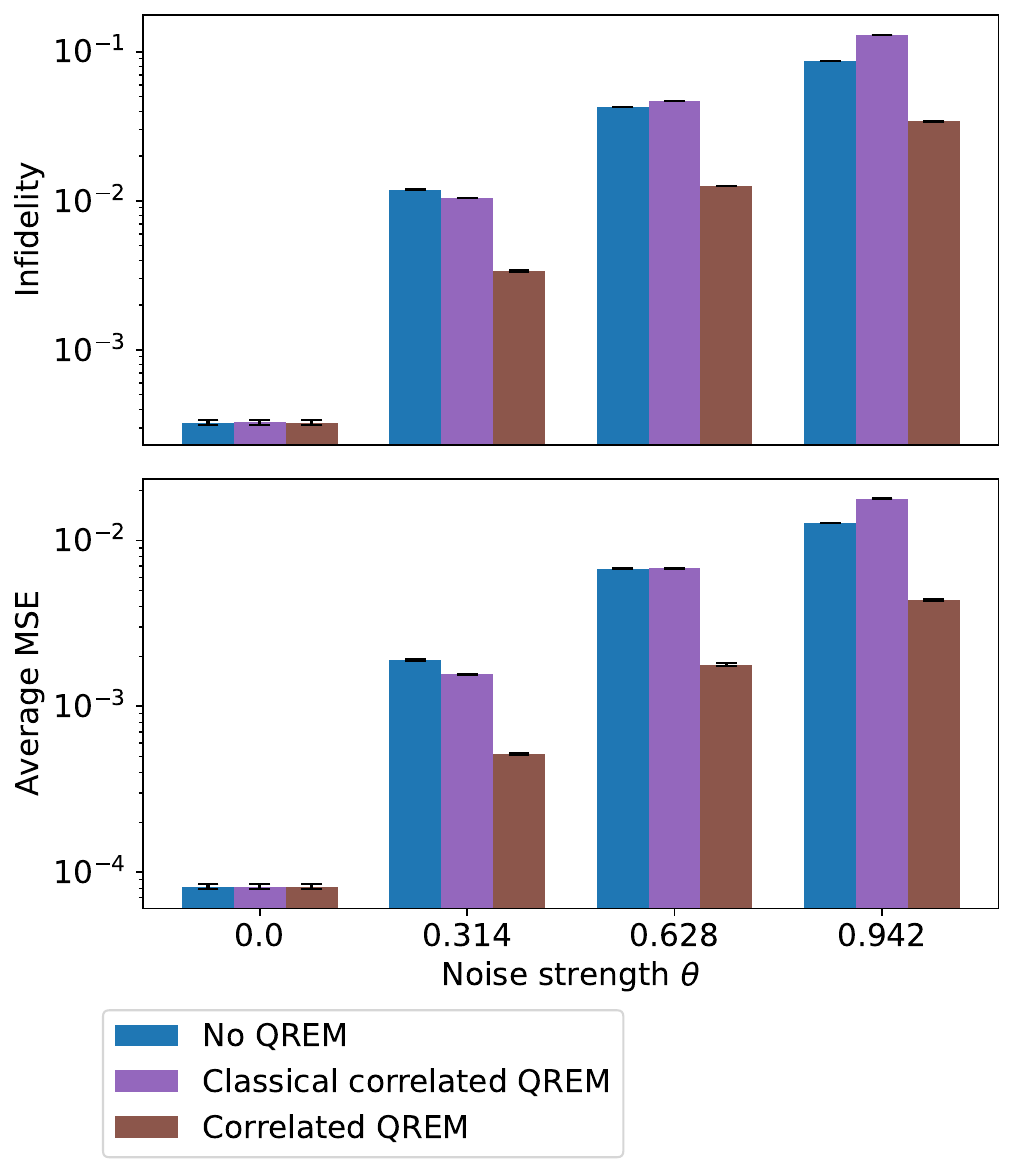}
 \caption{The coherent noise example from Sec.~\ref{sec:correlation_coefficients} with the ``complete'' linkage method. Everything else is identical to Fig.~\ref{fig:coherent_example}.}
\label{fig:coherent_qubit_example_with_complete_backend}
\end{figure}

Fig.~\ref{fig:coherent_qubit_example_with_complete_backend} illustrates another simulation from Sec.~\ref{sec:strong_coherent}, using the ``complete'' method for the hierarchical clustering linkage map. This method optimizes the clustering obtained based on the minimal ``largest distance" between the nodes that should be joined. As expected, the correlated QREM protocol performs significantly worse than when the ``average" linkage map is used, which optimized joined nodes with the smallest average distance.  This comes from the fact that only two-qubit correlated errors are correctly captured and three- and four-qubit correlations are missed. This issue is discussed in more detail in Appendix~\ref{app:Limitation_coherent_error}.

\subsection{100 qubits with mixed target states}
\label{app:100_qubits_mixed_states}
Previous simulations have used random pure states as target states. Although most experiments focus on generating pure states, circuit errors can result in output states being mixed even before noisy measurements are performed. To confirm that our protocol also works well for mixed states, we reran the simulations from Sec.~\ref{sec:100_qubit_example} using mixed target states, with the findings depicted in Fig.~\ref{fig:100_qubit_mixed_state_example}. The target states are generated from Haar-random pure states by applying a depolarizing channel. For each four-qubit pure state, the channel is applied with a strength $p$ uniformly selected from the range $[0.03, 0.07]$, leading to states with an average purity of approximately $0.9$. The four-qubit depolarizing channel is given by
\begin{equation}
    \mathcal{E}(\rho) = (1-p)\rho + \frac{p}{2^4}\mathbb{1}. 
\end{equation}
The overall performance for all readout error-mitigation methods considered is slightly better than with the pure target states in Fig.~\ref{fig:100_qubit_example}, but qualitatively the same. This is expected since pure states are harder to estimate than mixed states \cite{Mahler2013}, resulting in a higher overall state reconstruction accuracy for mixed states.

\begin{figure}
 \centering
 \includegraphics[width=\linewidth]{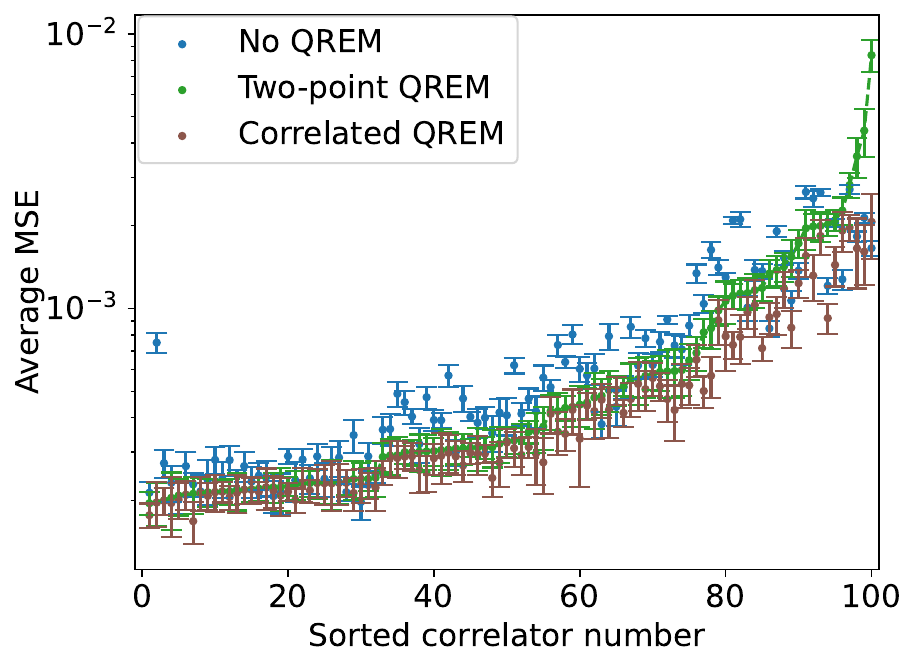}
 \caption{100-qubit example with mixed target states. Everything else is identical to Fig.~\ref{fig:100_qubit_example}.}
 \label{fig:100_qubit_mixed_state_example}
\end{figure}

\subsection{100-qubit simulation including single-qubit observables}
\label{app:additional_results_100_qubit}

\begin{figure}
 \centering
 \includegraphics[width=\linewidth]{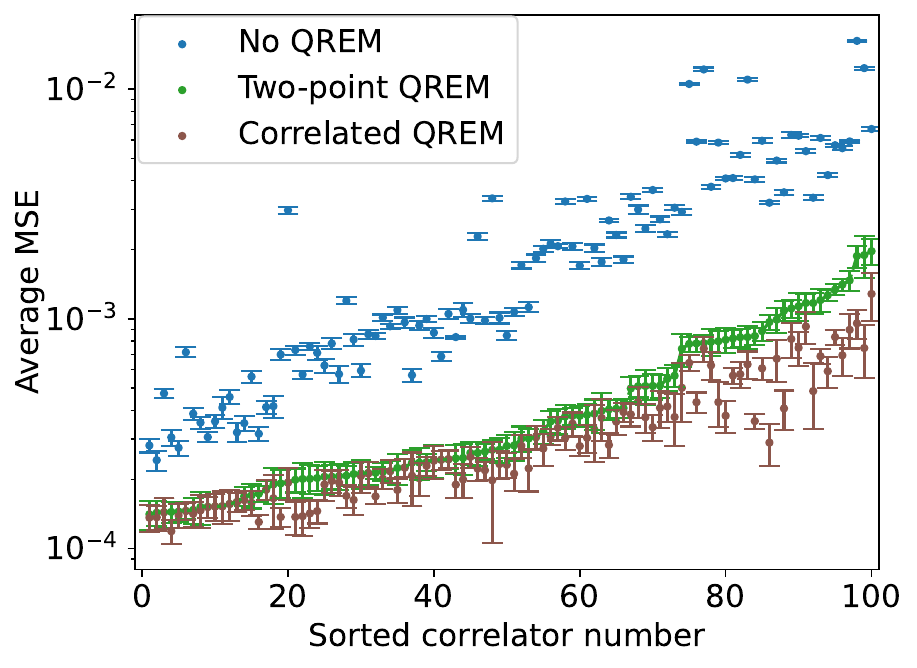}
 \caption{100-qubit example with the identity operator included in the list of Pauli operators. Everything else is identical to Fig.~\ref{fig:100_qubit_example}.}
 \label{fig:100_qubit_example_with_identity}
\end{figure}

Throughout the simulations presented in Sec.~\ref{sec:results}, the MSE was computed for all two-qubit Pauli observables and averaged to give representative protocol performance. In Fig.~\ref{fig:100_qubit_example_with_identity} the 100-quibt example is recomputed with the identity operator included in the list of Pauli observables.
The effect of including the identity operator is that the single-qubit observables are also averaged in the MSE. A slight decrease in the performance of the no QREM protocol can be seen, while both the two-point and correlated QREM protocols improve significantly on average, indicating that both protocols are also well suited to capture all lower-order observables.

\bibliography{refs}

\begin{thebibliography}{58}%
\makeatletter
\providecommand \@ifxundefined [1]{%
 \@ifx{#1\undefined}
}%
\providecommand \@ifnum [1]{%
 \ifnum #1\expandafter \@firstoftwo
 \else \expandafter \@secondoftwo
 \fi
}%
\providecommand \@ifx [1]{%
 \ifx #1\expandafter \@firstoftwo
 \else \expandafter \@secondoftwo
 \fi
}%
\providecommand \natexlab [1]{#1}%
\providecommand \enquote  [1]{``#1''}%
\providecommand \bibnamefont  [1]{#1}%
\providecommand \bibfnamefont [1]{#1}%
\providecommand \citenamefont [1]{#1}%
\providecommand \href@noop [0]{\@secondoftwo}%
\providecommand \href [0]{\begingroup \@sanitize@url \@href}%
\providecommand \@href[1]{\@@startlink{#1}\@@href}%
\providecommand \@@href[1]{\endgroup#1\@@endlink}%
\providecommand \@sanitize@url [0]{\catcode `\\12\catcode `\$12\catcode `\&12\catcode `\#12\catcode `\^12\catcode `\_12\catcode `\%12\relax}%
\providecommand \@@startlink[1]{}%
\providecommand \@@endlink[0]{}%
\providecommand \url  [0]{\begingroup\@sanitize@url \@url }%
\providecommand \@url [1]{\endgroup\@href {#1}{\urlprefix }}%
\providecommand \urlprefix  [0]{URL }%
\providecommand \Eprint [0]{\href }%
\providecommand \doibase [0]{https://doi.org/}%
\providecommand \selectlanguage [0]{\@gobble}%
\providecommand \bibinfo  [0]{\@secondoftwo}%
\providecommand \bibfield  [0]{\@secondoftwo}%
\providecommand \translation [1]{[#1]}%
\providecommand \BibitemOpen [0]{}%
\providecommand \bibitemStop [0]{}%
\providecommand \bibitemNoStop [0]{.\EOS\space}%
\providecommand \EOS [0]{\spacefactor3000\relax}%
\providecommand \BibitemShut  [1]{\csname bibitem#1\endcsname}%
\let\auto@bib@innerbib\@empty
\bibitem [{\citenamefont {Altman}\ \emph {et~al.}(2021)\citenamefont {Altman}, \citenamefont {Brown}, \citenamefont {Carleo}, \citenamefont {Carr}, \citenamefont {Demler}, \citenamefont {Chin}, \citenamefont {DeMarco}, \citenamefont {Economou}, \citenamefont {Eriksson}, \citenamefont {Fu}, \citenamefont {Greiner}, \citenamefont {Hazzard}, \citenamefont {Hulet}, \citenamefont {Kollár}, \citenamefont {Lev}, \citenamefont {Lukin}, \citenamefont {Ma}, \citenamefont {Mi}, \citenamefont {Misra}, \citenamefont {Monroe}, \citenamefont {Murch}, \citenamefont {Nazario}, \citenamefont {Ni}, \citenamefont {Potter}, \citenamefont {Roushan}, \citenamefont {Saffman}, \citenamefont {Schleier-Smith}, \citenamefont {Siddiqi}, \citenamefont {Simmonds}, \citenamefont {Singh}, \citenamefont {Spielman}, \citenamefont {Temme}, \citenamefont {Weiss}, \citenamefont {Vučković}, \citenamefont {Vuletić}, \citenamefont {Ye},\ and\ \citenamefont {Zwierlein}}]{Altman2021}%
  \BibitemOpen
  \bibfield  {author} {\bibinfo {author} {\bibfnamefont {E.}~\bibnamefont {Altman}}, \bibinfo {author} {\bibfnamefont {K.~R.}\ \bibnamefont {Brown}}, \bibinfo {author} {\bibfnamefont {G.}~\bibnamefont {Carleo}}, \bibinfo {author} {\bibfnamefont {L.~D.}\ \bibnamefont {Carr}}, \bibinfo {author} {\bibfnamefont {E.}~\bibnamefont {Demler}}, \bibinfo {author} {\bibfnamefont {C.}~\bibnamefont {Chin}}, \bibinfo {author} {\bibfnamefont {B.}~\bibnamefont {DeMarco}}, \bibinfo {author} {\bibfnamefont {S.~E.}\ \bibnamefont {Economou}}, \bibinfo {author} {\bibfnamefont {M.~A.}\ \bibnamefont {Eriksson}}, \bibinfo {author} {\bibfnamefont {K.-M.~C.}\ \bibnamefont {Fu}}, \bibinfo {author} {\bibfnamefont {M.}~\bibnamefont {Greiner}}, \bibinfo {author} {\bibfnamefont {K.~R.}\ \bibnamefont {Hazzard}}, \bibinfo {author} {\bibfnamefont {R.~G.}\ \bibnamefont {Hulet}}, \bibinfo {author} {\bibfnamefont {A.~J.}\ \bibnamefont {Kollár}}, \bibinfo {author} {\bibfnamefont {B.~L.}\ \bibnamefont {Lev}}, \bibinfo {author} {\bibfnamefont
  {M.~D.}\ \bibnamefont {Lukin}}, \bibinfo {author} {\bibfnamefont {R.}~\bibnamefont {Ma}}, \bibinfo {author} {\bibfnamefont {X.}~\bibnamefont {Mi}}, \bibinfo {author} {\bibfnamefont {S.}~\bibnamefont {Misra}}, \bibinfo {author} {\bibfnamefont {C.}~\bibnamefont {Monroe}}, \bibinfo {author} {\bibfnamefont {K.}~\bibnamefont {Murch}}, \bibinfo {author} {\bibfnamefont {Z.}~\bibnamefont {Nazario}}, \bibinfo {author} {\bibfnamefont {K.-K.}\ \bibnamefont {Ni}}, \bibinfo {author} {\bibfnamefont {A.~C.}\ \bibnamefont {Potter}}, \bibinfo {author} {\bibfnamefont {P.}~\bibnamefont {Roushan}}, \bibinfo {author} {\bibfnamefont {M.}~\bibnamefont {Saffman}}, \bibinfo {author} {\bibfnamefont {M.}~\bibnamefont {Schleier-Smith}}, \bibinfo {author} {\bibfnamefont {I.}~\bibnamefont {Siddiqi}}, \bibinfo {author} {\bibfnamefont {R.}~\bibnamefont {Simmonds}}, \bibinfo {author} {\bibfnamefont {M.}~\bibnamefont {Singh}}, \bibinfo {author} {\bibfnamefont {I.}~\bibnamefont {Spielman}}, \bibinfo {author} {\bibfnamefont {K.}~\bibnamefont
  {Temme}}, \bibinfo {author} {\bibfnamefont {D.~S.}\ \bibnamefont {Weiss}}, \bibinfo {author} {\bibfnamefont {J.}~\bibnamefont {Vučković}}, \bibinfo {author} {\bibfnamefont {V.}~\bibnamefont {Vuletić}}, \bibinfo {author} {\bibfnamefont {J.}~\bibnamefont {Ye}},\ and\ \bibinfo {author} {\bibfnamefont {M.}~\bibnamefont {Zwierlein}},\ }\bibfield  {title} {\bibinfo {title} {Quantum simulators: Architectures and opportunities},\ }\href {https://doi.org/10.1103/prxquantum.2.017003} {\bibfield  {journal} {\bibinfo  {journal} {PRX Quantum}\ }\textbf {\bibinfo {volume} {2}},\  (\bibinfo {year} {2021})}\BibitemShut {NoStop}%
\bibitem [{\citenamefont {Temme}\ \emph {et~al.}(2017)\citenamefont {Temme}, \citenamefont {Bravyi},\ and\ \citenamefont {Gambetta}}]{Temme2017}%
  \BibitemOpen
  \bibfield  {author} {\bibinfo {author} {\bibfnamefont {K.}~\bibnamefont {Temme}}, \bibinfo {author} {\bibfnamefont {S.}~\bibnamefont {Bravyi}},\ and\ \bibinfo {author} {\bibfnamefont {J.~M.}\ \bibnamefont {Gambetta}},\ }\bibfield  {title} {\bibinfo {title} {Error mitigation for short-depth quantum circuits},\ }\href {https://doi.org/10.1103/physrevlett.119.180509} {\bibfield  {journal} {\bibinfo  {journal} {Physical Review Letters}\ }\textbf {\bibinfo {volume} {119}},\  (\bibinfo {year} {2017})}\BibitemShut {NoStop}%
\bibitem [{\citenamefont {Bonet-Monroig}\ \emph {et~al.}(2018)\citenamefont {Bonet-Monroig}, \citenamefont {Sagastizabal}, \citenamefont {Singh},\ and\ \citenamefont {O’Brien}}]{BonetMonroig2018}%
  \BibitemOpen
  \bibfield  {author} {\bibinfo {author} {\bibfnamefont {X.}~\bibnamefont {Bonet-Monroig}}, \bibinfo {author} {\bibfnamefont {R.}~\bibnamefont {Sagastizabal}}, \bibinfo {author} {\bibfnamefont {M.}~\bibnamefont {Singh}},\ and\ \bibinfo {author} {\bibfnamefont {T.~E.}\ \bibnamefont {O’Brien}},\ }\bibfield  {title} {\bibinfo {title} {Low-cost error mitigation by symmetry verification},\ }\href {https://doi.org/10.1103/physreva.98.062339} {\bibfield  {journal} {\bibinfo  {journal} {Physical Review A}\ }\textbf {\bibinfo {volume} {98}},\  (\bibinfo {year} {2018})}\BibitemShut {NoStop}%
\bibitem [{\citenamefont {van~den Berg}\ \emph {et~al.}(2023)\citenamefont {van~den Berg}, \citenamefont {Minev}, \citenamefont {Kandala},\ and\ \citenamefont {Temme}}]{vandenBerg2023}%
  \BibitemOpen
  \bibfield  {author} {\bibinfo {author} {\bibfnamefont {E.}~\bibnamefont {van~den Berg}}, \bibinfo {author} {\bibfnamefont {Z.~K.}\ \bibnamefont {Minev}}, \bibinfo {author} {\bibfnamefont {A.}~\bibnamefont {Kandala}},\ and\ \bibinfo {author} {\bibfnamefont {K.}~\bibnamefont {Temme}},\ }\bibfield  {title} {\bibinfo {title} {Probabilistic error cancellation with sparse pauli–lindblad models on noisy quantum processors},\ }\href {https://doi.org/10.1038/s41567-023-02042-2} {\bibfield  {journal} {\bibinfo  {journal} {Nature Physics}\ }\textbf {\bibinfo {volume} {19}},\ \bibinfo {pages} {1116–1121} (\bibinfo {year} {2023})}\BibitemShut {NoStop}%
\bibitem [{\citenamefont {Cai}\ \emph {et~al.}(2023)\citenamefont {Cai}, \citenamefont {Babbush}, \citenamefont {Benjamin}, \citenamefont {Endo}, \citenamefont {Huggins}, \citenamefont {Li}, \citenamefont {McClean},\ and\ \citenamefont {O’Brien}}]{Cai2023}%
  \BibitemOpen
  \bibfield  {author} {\bibinfo {author} {\bibfnamefont {Z.}~\bibnamefont {Cai}}, \bibinfo {author} {\bibfnamefont {R.}~\bibnamefont {Babbush}}, \bibinfo {author} {\bibfnamefont {S.~C.}\ \bibnamefont {Benjamin}}, \bibinfo {author} {\bibfnamefont {S.}~\bibnamefont {Endo}}, \bibinfo {author} {\bibfnamefont {W.~J.}\ \bibnamefont {Huggins}}, \bibinfo {author} {\bibfnamefont {Y.}~\bibnamefont {Li}}, \bibinfo {author} {\bibfnamefont {J.~R.}\ \bibnamefont {McClean}},\ and\ \bibinfo {author} {\bibfnamefont {T.~E.}\ \bibnamefont {O’Brien}},\ }\bibfield  {title} {\bibinfo {title} {Quantum error mitigation},\ }\href {https://doi.org/10.1103/revmodphys.95.045005} {\bibfield  {journal} {\bibinfo  {journal} {Reviews of Modern Physics}\ }\textbf {\bibinfo {volume} {95}},\  (\bibinfo {year} {2023})}\BibitemShut {NoStop}%
\bibitem [{\citenamefont {Quek}\ \emph {et~al.}(2024)\citenamefont {Quek}, \citenamefont {Stilck~Fran\c{c}a}, \citenamefont {Khatri}, \citenamefont {Meyer},\ and\ \citenamefont {Eisert}}]{Quek2024}%
  \BibitemOpen
  \bibfield  {author} {\bibinfo {author} {\bibfnamefont {Y.}~\bibnamefont {Quek}}, \bibinfo {author} {\bibfnamefont {D.}~\bibnamefont {Stilck~Fran\c{c}a}}, \bibinfo {author} {\bibfnamefont {S.}~\bibnamefont {Khatri}}, \bibinfo {author} {\bibfnamefont {J.~J.}\ \bibnamefont {Meyer}},\ and\ \bibinfo {author} {\bibfnamefont {J.}~\bibnamefont {Eisert}},\ }\bibfield  {title} {\bibinfo {title} {Exponentially tighter bounds on limitations of quantum error mitigation},\ }\href {https://doi.org/10.1038/s41567-024-02536-7} {\bibfield  {journal} {\bibinfo  {journal} {Nature Physics}\ }\textbf {\bibinfo {volume} {20}},\ \bibinfo {pages} {1648–1658} (\bibinfo {year} {2024})}\BibitemShut {NoStop}%
\bibitem [{\citenamefont {Zimborás}\ \emph {et~al.}(2025)\citenamefont {Zimborás}, \citenamefont {Koczor}, \citenamefont {Holmes}, \citenamefont {Borrelli}, \citenamefont {Gilyén}, \citenamefont {Huang}, \citenamefont {Cai}, \citenamefont {Acín}, \citenamefont {Aolita}, \citenamefont {Banchi}, \citenamefont {Brandão}, \citenamefont {Cavalcanti}, \citenamefont {Cubitt}, \citenamefont {Filippov}, \citenamefont {García-Pérez}, \citenamefont {Goold}, \citenamefont {Kálmán}, \citenamefont {Kyoseva}, \citenamefont {Rossi}, \citenamefont {Sokolov}, \citenamefont {Tavernelli},\ and\ \citenamefont {Maniscalco}}]{Zimboras2025}%
  \BibitemOpen
  \bibfield  {author} {\bibinfo {author} {\bibfnamefont {Z.}~\bibnamefont {Zimborás}}, \bibinfo {author} {\bibfnamefont {B.}~\bibnamefont {Koczor}}, \bibinfo {author} {\bibfnamefont {Z.}~\bibnamefont {Holmes}}, \bibinfo {author} {\bibfnamefont {E.-M.}\ \bibnamefont {Borrelli}}, \bibinfo {author} {\bibfnamefont {A.}~\bibnamefont {Gilyén}}, \bibinfo {author} {\bibfnamefont {H.-Y.}\ \bibnamefont {Huang}}, \bibinfo {author} {\bibfnamefont {Z.}~\bibnamefont {Cai}}, \bibinfo {author} {\bibfnamefont {A.}~\bibnamefont {Acín}}, \bibinfo {author} {\bibfnamefont {L.}~\bibnamefont {Aolita}}, \bibinfo {author} {\bibfnamefont {L.}~\bibnamefont {Banchi}}, \bibinfo {author} {\bibfnamefont {F.~G. S.~L.}\ \bibnamefont {Brandão}}, \bibinfo {author} {\bibfnamefont {D.}~\bibnamefont {Cavalcanti}}, \bibinfo {author} {\bibfnamefont {T.}~\bibnamefont {Cubitt}}, \bibinfo {author} {\bibfnamefont {S.~N.}\ \bibnamefont {Filippov}}, \bibinfo {author} {\bibfnamefont {G.}~\bibnamefont {García-Pérez}}, \bibinfo {author} {\bibfnamefont
  {J.}~\bibnamefont {Goold}}, \bibinfo {author} {\bibfnamefont {O.}~\bibnamefont {Kálmán}}, \bibinfo {author} {\bibfnamefont {E.}~\bibnamefont {Kyoseva}}, \bibinfo {author} {\bibfnamefont {M.~A.~C.}\ \bibnamefont {Rossi}}, \bibinfo {author} {\bibfnamefont {B.}~\bibnamefont {Sokolov}}, \bibinfo {author} {\bibfnamefont {I.}~\bibnamefont {Tavernelli}},\ and\ \bibinfo {author} {\bibfnamefont {S.}~\bibnamefont {Maniscalco}},\ }\href {https://arxiv.org/abs/2501.05694} {\bibinfo {title} {Myths around quantum computation before full fault tolerance: What no-go theorems rule out and what they don't}} (\bibinfo {year} {2025}),\ \Eprint {https://arxiv.org/abs/2501.05694} {arXiv:2501.05694 [quant-ph]} \BibitemShut {NoStop}%
\bibitem [{\citenamefont {Tuziemski}\ \emph {et~al.}(2023)\citenamefont {Tuziemski}, \citenamefont {Maciejewski}, \citenamefont {Majsak}, \citenamefont {Słowik}, \citenamefont {Kotowski}, \citenamefont {Kowalczyk-Murynka}, \citenamefont {Podziemski},\ and\ \citenamefont {Oszmaniec}}]{Tuziemski2023}%
  \BibitemOpen
  \bibfield  {author} {\bibinfo {author} {\bibfnamefont {J.}~\bibnamefont {Tuziemski}}, \bibinfo {author} {\bibfnamefont {F.~B.}\ \bibnamefont {Maciejewski}}, \bibinfo {author} {\bibfnamefont {J.}~\bibnamefont {Majsak}}, \bibinfo {author} {\bibfnamefont {O.}~\bibnamefont {Słowik}}, \bibinfo {author} {\bibfnamefont {M.}~\bibnamefont {Kotowski}}, \bibinfo {author} {\bibfnamefont {K.}~\bibnamefont {Kowalczyk-Murynka}}, \bibinfo {author} {\bibfnamefont {P.}~\bibnamefont {Podziemski}},\ and\ \bibinfo {author} {\bibfnamefont {M.}~\bibnamefont {Oszmaniec}},\ }\href {https://arxiv.org/abs/2311.10661} {\bibinfo {title} {Efficient reconstruction, benchmarking and validation of cross-talk models in readout noise in near-term quantum devices}} (\bibinfo {year} {2023}),\ \Eprint {https://arxiv.org/abs/2311.10661} {arXiv:2311.10661 [quant-ph]} \BibitemShut {NoStop}%
\bibitem [{\citenamefont {Zhou}\ \emph {et~al.}(2023)\citenamefont {Zhou}, \citenamefont {Sitler}, \citenamefont {Oda}, \citenamefont {Schultz},\ and\ \citenamefont {Quiroz}}]{Zhou2023}%
  \BibitemOpen
  \bibfield  {author} {\bibinfo {author} {\bibfnamefont {Z.}~\bibnamefont {Zhou}}, \bibinfo {author} {\bibfnamefont {R.}~\bibnamefont {Sitler}}, \bibinfo {author} {\bibfnamefont {Y.}~\bibnamefont {Oda}}, \bibinfo {author} {\bibfnamefont {K.}~\bibnamefont {Schultz}},\ and\ \bibinfo {author} {\bibfnamefont {G.}~\bibnamefont {Quiroz}},\ }\bibfield  {title} {\bibinfo {title} {Quantum crosstalk robust quantum control},\ }\href {https://doi.org/10.1103/physrevlett.131.210802} {\bibfield  {journal} {\bibinfo  {journal} {Physical Review Letters}\ }\textbf {\bibinfo {volume} {131}},\  (\bibinfo {year} {2023})}\BibitemShut {NoStop}%
\bibitem [{\citenamefont {Geller}\ and\ \citenamefont {Sun}(2021)}]{Geller2021}%
  \BibitemOpen
  \bibfield  {author} {\bibinfo {author} {\bibfnamefont {M.~R.}\ \bibnamefont {Geller}}\ and\ \bibinfo {author} {\bibfnamefont {M.}~\bibnamefont {Sun}},\ }\bibfield  {title} {\bibinfo {title} {Toward efficient correction of multiqubit measurement errors: pair correlation method},\ }\href {https://doi.org/10.1088/2058-9565/abd5c9} {\bibfield  {journal} {\bibinfo  {journal} {Quantum Science and Technology}\ }\textbf {\bibinfo {volume} {6}},\ \bibinfo {pages} {025009} (\bibinfo {year} {2021})}\BibitemShut {NoStop}%
\bibitem [{\citenamefont {Sarovar}\ \emph {et~al.}(2020)\citenamefont {Sarovar}, \citenamefont {Proctor}, \citenamefont {Rudinger}, \citenamefont {Young}, \citenamefont {Nielsen},\ and\ \citenamefont {Blume-Kohout}}]{Sarovar2020}%
  \BibitemOpen
  \bibfield  {author} {\bibinfo {author} {\bibfnamefont {M.}~\bibnamefont {Sarovar}}, \bibinfo {author} {\bibfnamefont {T.}~\bibnamefont {Proctor}}, \bibinfo {author} {\bibfnamefont {K.}~\bibnamefont {Rudinger}}, \bibinfo {author} {\bibfnamefont {K.}~\bibnamefont {Young}}, \bibinfo {author} {\bibfnamefont {E.}~\bibnamefont {Nielsen}},\ and\ \bibinfo {author} {\bibfnamefont {R.}~\bibnamefont {Blume-Kohout}},\ }\bibfield  {title} {\bibinfo {title} {Detecting crosstalk errors in quantum information processors},\ }\href {https://doi.org/10.22331/q-2020-09-11-321} {\bibfield  {journal} {\bibinfo  {journal} {Quantum}\ }\textbf {\bibinfo {volume} {4}},\ \bibinfo {pages} {321} (\bibinfo {year} {2020})}\BibitemShut {NoStop}%
\bibitem [{\citenamefont {Heinsoo}\ \emph {et~al.}(2018)\citenamefont {Heinsoo}, \citenamefont {Andersen}, \citenamefont {Remm}, \citenamefont {Krinner}, \citenamefont {Walter}, \citenamefont {Salathé}, \citenamefont {Gasparinetti}, \citenamefont {Besse}, \citenamefont {Potočnik}, \citenamefont {Wallraff},\ and\ \citenamefont {Eichler}}]{Heinsoo2018}%
  \BibitemOpen
  \bibfield  {author} {\bibinfo {author} {\bibfnamefont {J.}~\bibnamefont {Heinsoo}}, \bibinfo {author} {\bibfnamefont {C.~K.}\ \bibnamefont {Andersen}}, \bibinfo {author} {\bibfnamefont {A.}~\bibnamefont {Remm}}, \bibinfo {author} {\bibfnamefont {S.}~\bibnamefont {Krinner}}, \bibinfo {author} {\bibfnamefont {T.}~\bibnamefont {Walter}}, \bibinfo {author} {\bibfnamefont {Y.}~\bibnamefont {Salathé}}, \bibinfo {author} {\bibfnamefont {S.}~\bibnamefont {Gasparinetti}}, \bibinfo {author} {\bibfnamefont {J.-C.}\ \bibnamefont {Besse}}, \bibinfo {author} {\bibfnamefont {A.}~\bibnamefont {Potočnik}}, \bibinfo {author} {\bibfnamefont {A.}~\bibnamefont {Wallraff}},\ and\ \bibinfo {author} {\bibfnamefont {C.}~\bibnamefont {Eichler}},\ }\bibfield  {title} {\bibinfo {title} {Rapid high-fidelity multiplexed readout of superconducting qubits},\ }\href {https://doi.org/10.1103/physrevapplied.10.034040} {\bibfield  {journal} {\bibinfo  {journal} {Physical Review Applied}\ }\textbf {\bibinfo {volume} {10}},\  (\bibinfo {year}
  {2018})}\BibitemShut {NoStop}%
\bibitem [{\citenamefont {de~Groot}\ \emph {et~al.}(2010)\citenamefont {de~Groot}, \citenamefont {van Loo}, \citenamefont {Lisenfeld}, \citenamefont {Schouten}, \citenamefont {Lupaşcu}, \citenamefont {Harmans},\ and\ \citenamefont {Mooij}}]{deGroot2010}%
  \BibitemOpen
  \bibfield  {author} {\bibinfo {author} {\bibfnamefont {P.~C.}\ \bibnamefont {de~Groot}}, \bibinfo {author} {\bibfnamefont {A.~F.}\ \bibnamefont {van Loo}}, \bibinfo {author} {\bibfnamefont {J.}~\bibnamefont {Lisenfeld}}, \bibinfo {author} {\bibfnamefont {R.~N.}\ \bibnamefont {Schouten}}, \bibinfo {author} {\bibfnamefont {A.}~\bibnamefont {Lupaşcu}}, \bibinfo {author} {\bibfnamefont {C.~J. P.~M.}\ \bibnamefont {Harmans}},\ and\ \bibinfo {author} {\bibfnamefont {J.~E.}\ \bibnamefont {Mooij}},\ }\bibfield  {title} {\bibinfo {title} {Low-crosstalk bifurcation detectors for coupled flux qubits},\ }\href {https://doi.org/10.1063/1.3367875} {\bibfield  {journal} {\bibinfo  {journal} {Applied Physics Letters}\ }\textbf {\bibinfo {volume} {96}},\  (\bibinfo {year} {2010})}\BibitemShut {NoStop}%
\bibitem [{\citenamefont {Lienhard}\ \emph {et~al.}(2022)\citenamefont {Lienhard}, \citenamefont {Veps\"al\"ainen}, \citenamefont {Govia}, \citenamefont {Hoffer}, \citenamefont {Qiu}, \citenamefont {Rist\`e}, \citenamefont {Ware}, \citenamefont {Kim}, \citenamefont {Winik}, \citenamefont {Melville}, \citenamefont {Niedzielski}, \citenamefont {Yoder}, \citenamefont {Ribeill}, \citenamefont {Ohki}, \citenamefont {Krovi}, \citenamefont {Orlando}, \citenamefont {Gustavsson},\ and\ \citenamefont {Oliver}}]{Lienhard2022}%
  \BibitemOpen
  \bibfield  {author} {\bibinfo {author} {\bibfnamefont {B.}~\bibnamefont {Lienhard}}, \bibinfo {author} {\bibfnamefont {A.}~\bibnamefont {Veps\"al\"ainen}}, \bibinfo {author} {\bibfnamefont {L.~C.}\ \bibnamefont {Govia}}, \bibinfo {author} {\bibfnamefont {C.~R.}\ \bibnamefont {Hoffer}}, \bibinfo {author} {\bibfnamefont {J.~Y.}\ \bibnamefont {Qiu}}, \bibinfo {author} {\bibfnamefont {D.}~\bibnamefont {Rist\`e}}, \bibinfo {author} {\bibfnamefont {M.}~\bibnamefont {Ware}}, \bibinfo {author} {\bibfnamefont {D.}~\bibnamefont {Kim}}, \bibinfo {author} {\bibfnamefont {R.}~\bibnamefont {Winik}}, \bibinfo {author} {\bibfnamefont {A.}~\bibnamefont {Melville}}, \bibinfo {author} {\bibfnamefont {B.}~\bibnamefont {Niedzielski}}, \bibinfo {author} {\bibfnamefont {J.}~\bibnamefont {Yoder}}, \bibinfo {author} {\bibfnamefont {G.~J.}\ \bibnamefont {Ribeill}}, \bibinfo {author} {\bibfnamefont {T.~A.}\ \bibnamefont {Ohki}}, \bibinfo {author} {\bibfnamefont {H.~K.}\ \bibnamefont {Krovi}}, \bibinfo {author} {\bibfnamefont {T.~P.}\
  \bibnamefont {Orlando}}, \bibinfo {author} {\bibfnamefont {S.}~\bibnamefont {Gustavsson}},\ and\ \bibinfo {author} {\bibfnamefont {W.~D.}\ \bibnamefont {Oliver}},\ }\bibfield  {title} {\bibinfo {title} {Deep-neural-network discrimination of multiplexed superconducting-qubit states},\ }\href {https://doi.org/10.1103/PhysRevApplied.17.014024} {\bibfield  {journal} {\bibinfo  {journal} {Phys. Rev. Appl.}\ }\textbf {\bibinfo {volume} {17}},\ \bibinfo {pages} {014024} (\bibinfo {year} {2022})}\BibitemShut {NoStop}%
\bibitem [{\citenamefont {Maciejewski}\ \emph {et~al.}(2020)\citenamefont {Maciejewski}, \citenamefont {Zimborás},\ and\ \citenamefont {Oszmaniec}}]{Maciejewski2020}%
  \BibitemOpen
  \bibfield  {author} {\bibinfo {author} {\bibfnamefont {F.~B.}\ \bibnamefont {Maciejewski}}, \bibinfo {author} {\bibfnamefont {Z.}~\bibnamefont {Zimborás}},\ and\ \bibinfo {author} {\bibfnamefont {M.}~\bibnamefont {Oszmaniec}},\ }\bibfield  {title} {\bibinfo {title} {Mitigation of readout noise in near-term quantum devices by classical post-processing based on detector tomography},\ }\href {https://doi.org/10.22331/q-2020-04-24-257} {\bibfield  {journal} {\bibinfo  {journal} {Quantum}\ }\textbf {\bibinfo {volume} {4}},\ \bibinfo {pages} {257} (\bibinfo {year} {2020})}\BibitemShut {NoStop}%
\bibitem [{\citenamefont {Nation}\ \emph {et~al.}(2021)\citenamefont {Nation}, \citenamefont {Kang}, \citenamefont {Sundaresan},\ and\ \citenamefont {Gambetta}}]{Nation2021}%
  \BibitemOpen
  \bibfield  {author} {\bibinfo {author} {\bibfnamefont {P.~D.}\ \bibnamefont {Nation}}, \bibinfo {author} {\bibfnamefont {H.}~\bibnamefont {Kang}}, \bibinfo {author} {\bibfnamefont {N.}~\bibnamefont {Sundaresan}},\ and\ \bibinfo {author} {\bibfnamefont {J.~M.}\ \bibnamefont {Gambetta}},\ }\bibfield  {title} {\bibinfo {title} {Scalable mitigation of measurement errors on quantum computers},\ }\href {https://doi.org/10.1103/prxquantum.2.040326} {\bibfield  {journal} {\bibinfo  {journal} {PRX Quantum}\ }\textbf {\bibinfo {volume} {2}},\  (\bibinfo {year} {2021})}\BibitemShut {NoStop}%
\bibitem [{\citenamefont {Nachman}\ \emph {et~al.}(2020)\citenamefont {Nachman}, \citenamefont {Urbanek}, \citenamefont {de~Jong},\ and\ \citenamefont {Bauer}}]{Nachman2020}%
  \BibitemOpen
  \bibfield  {author} {\bibinfo {author} {\bibfnamefont {B.}~\bibnamefont {Nachman}}, \bibinfo {author} {\bibfnamefont {M.}~\bibnamefont {Urbanek}}, \bibinfo {author} {\bibfnamefont {W.~A.}\ \bibnamefont {de~Jong}},\ and\ \bibinfo {author} {\bibfnamefont {C.~W.}\ \bibnamefont {Bauer}},\ }\bibfield  {title} {\bibinfo {title} {Unfolding quantum computer readout noise},\ }\href {https://doi.org/10.1038/s41534-020-00309-7} {\bibfield  {journal} {\bibinfo  {journal} {npj Quantum Information}\ }\textbf {\bibinfo {volume} {6}},\  (\bibinfo {year} {2020})}\BibitemShut {NoStop}%
\bibitem [{\citenamefont {Huang}\ \emph {et~al.}(2020)\citenamefont {Huang}, \citenamefont {Kueng},\ and\ \citenamefont {Preskill}}]{Huang2020}%
  \BibitemOpen
  \bibfield  {author} {\bibinfo {author} {\bibfnamefont {H.-Y.}\ \bibnamefont {Huang}}, \bibinfo {author} {\bibfnamefont {R.}~\bibnamefont {Kueng}},\ and\ \bibinfo {author} {\bibfnamefont {J.}~\bibnamefont {Preskill}},\ }\bibfield  {title} {\bibinfo {title} {Predicting many properties of a quantum system from very few measurements},\ }\href {https://doi.org/10.1038/s41567-020-0932-7} {\bibfield  {journal} {\bibinfo  {journal} {Nature Physics}\ }\textbf {\bibinfo {volume} {16}},\ \bibinfo {pages} {1050–1057} (\bibinfo {year} {2020})}\BibitemShut {NoStop}%
\bibitem [{\citenamefont {Chen}\ \emph {et~al.}(2021)\citenamefont {Chen}, \citenamefont {Yu}, \citenamefont {Zeng},\ and\ \citenamefont {Flammia}}]{Chen2021}%
  \BibitemOpen
  \bibfield  {author} {\bibinfo {author} {\bibfnamefont {S.}~\bibnamefont {Chen}}, \bibinfo {author} {\bibfnamefont {W.}~\bibnamefont {Yu}}, \bibinfo {author} {\bibfnamefont {P.}~\bibnamefont {Zeng}},\ and\ \bibinfo {author} {\bibfnamefont {S.~T.}\ \bibnamefont {Flammia}},\ }\bibfield  {title} {\bibinfo {title} {Robust shadow estimation},\ }\href {https://doi.org/10.1103/prxquantum.2.030348} {\bibfield  {journal} {\bibinfo  {journal} {PRX Quantum}\ }\textbf {\bibinfo {volume} {2}},\  (\bibinfo {year} {2021})}\BibitemShut {NoStop}%
\bibitem [{\citenamefont {Jnane}\ \emph {et~al.}(2024)\citenamefont {Jnane}, \citenamefont {Steinberg}, \citenamefont {Cai}, \citenamefont {Nguyen},\ and\ \citenamefont {Koczor}}]{Jnane2024}%
  \BibitemOpen
  \bibfield  {author} {\bibinfo {author} {\bibfnamefont {H.}~\bibnamefont {Jnane}}, \bibinfo {author} {\bibfnamefont {J.}~\bibnamefont {Steinberg}}, \bibinfo {author} {\bibfnamefont {Z.}~\bibnamefont {Cai}}, \bibinfo {author} {\bibfnamefont {H.~C.}\ \bibnamefont {Nguyen}},\ and\ \bibinfo {author} {\bibfnamefont {B.}~\bibnamefont {Koczor}},\ }\bibfield  {title} {\bibinfo {title} {Quantum error mitigated classical shadows},\ }\href {https://doi.org/10.1103/prxquantum.5.010324} {\bibfield  {journal} {\bibinfo  {journal} {PRX Quantum}\ }\textbf {\bibinfo {volume} {5}},\  (\bibinfo {year} {2024})}\BibitemShut {NoStop}%
\bibitem [{\citenamefont {Koh}\ and\ \citenamefont {Grewal}(2022)}]{Koh2022}%
  \BibitemOpen
  \bibfield  {author} {\bibinfo {author} {\bibfnamefont {D.~E.}\ \bibnamefont {Koh}}\ and\ \bibinfo {author} {\bibfnamefont {S.}~\bibnamefont {Grewal}},\ }\bibfield  {title} {\bibinfo {title} {Classical shadows with noise},\ }\href {https://doi.org/10.22331/q-2022-08-16-776} {\bibfield  {journal} {\bibinfo  {journal} {Quantum}\ }\textbf {\bibinfo {volume} {6}},\ \bibinfo {pages} {776} (\bibinfo {year} {2022})}\BibitemShut {NoStop}%
\bibitem [{\citenamefont {Arrasmith}\ \emph {et~al.}(2023)\citenamefont {Arrasmith}, \citenamefont {Patterson}, \citenamefont {Boughton},\ and\ \citenamefont {Paini}}]{Arrasmith2023}%
  \BibitemOpen
  \bibfield  {author} {\bibinfo {author} {\bibfnamefont {A.}~\bibnamefont {Arrasmith}}, \bibinfo {author} {\bibfnamefont {A.}~\bibnamefont {Patterson}}, \bibinfo {author} {\bibfnamefont {A.}~\bibnamefont {Boughton}},\ and\ \bibinfo {author} {\bibfnamefont {M.}~\bibnamefont {Paini}},\ }\href {https://arxiv.org/abs/2303.17741} {\bibinfo {title} {Development and demonstration of an efficient readout error mitigation technique for use in nisq algorithms}} (\bibinfo {year} {2023}),\ \Eprint {https://arxiv.org/abs/2303.17741} {arXiv:2303.17741 [quant-ph]} \BibitemShut {NoStop}%
\bibitem [{\citenamefont {Hu}\ \emph {et~al.}(2025)\citenamefont {Hu}, \citenamefont {Gu}, \citenamefont {Majumder}, \citenamefont {Ren}, \citenamefont {Zhang}, \citenamefont {Wang}, \citenamefont {You}, \citenamefont {Minev}, \citenamefont {Yelin},\ and\ \citenamefont {Seif}}]{Hu2025}%
  \BibitemOpen
  \bibfield  {author} {\bibinfo {author} {\bibfnamefont {H.-Y.}\ \bibnamefont {Hu}}, \bibinfo {author} {\bibfnamefont {A.}~\bibnamefont {Gu}}, \bibinfo {author} {\bibfnamefont {S.}~\bibnamefont {Majumder}}, \bibinfo {author} {\bibfnamefont {H.}~\bibnamefont {Ren}}, \bibinfo {author} {\bibfnamefont {Y.}~\bibnamefont {Zhang}}, \bibinfo {author} {\bibfnamefont {D.~S.}\ \bibnamefont {Wang}}, \bibinfo {author} {\bibfnamefont {Y.-Z.}\ \bibnamefont {You}}, \bibinfo {author} {\bibfnamefont {Z.}~\bibnamefont {Minev}}, \bibinfo {author} {\bibfnamefont {S.~F.}\ \bibnamefont {Yelin}},\ and\ \bibinfo {author} {\bibfnamefont {A.}~\bibnamefont {Seif}},\ }\bibfield  {title} {\bibinfo {title} {Demonstration of robust and efficient quantum property learning with shallow shadows},\ }\href {https://doi.org/10.1038/s41467-025-57349-w} {\bibfield  {journal} {\bibinfo  {journal} {Nature Communications}\ }\textbf {\bibinfo {volume} {16}} (\bibinfo {year} {2025})}\BibitemShut {NoStop}%
\bibitem [{\citenamefont {Onorati}\ \emph {et~al.}(2024)\citenamefont {Onorati}, \citenamefont {Kitzinger}, \citenamefont {Helsen}, \citenamefont {Ioannou}, \citenamefont {Werner}, \citenamefont {Roth},\ and\ \citenamefont {Eisert}}]{Onorati2024}%
  \BibitemOpen
  \bibfield  {author} {\bibinfo {author} {\bibfnamefont {E.}~\bibnamefont {Onorati}}, \bibinfo {author} {\bibfnamefont {J.}~\bibnamefont {Kitzinger}}, \bibinfo {author} {\bibfnamefont {J.}~\bibnamefont {Helsen}}, \bibinfo {author} {\bibfnamefont {M.}~\bibnamefont {Ioannou}}, \bibinfo {author} {\bibfnamefont {A.~H.}\ \bibnamefont {Werner}}, \bibinfo {author} {\bibfnamefont {I.}~\bibnamefont {Roth}},\ and\ \bibinfo {author} {\bibfnamefont {J.}~\bibnamefont {Eisert}},\ }\href {https://arxiv.org/abs/2403.04751} {\bibinfo {title} {Noise-mitigated randomized measurements and self-calibrating shadow estimation}} (\bibinfo {year} {2024}),\ \Eprint {https://arxiv.org/abs/2403.04751} {arXiv:2403.04751 [quant-ph]} \BibitemShut {NoStop}%
\bibitem [{\citenamefont {Aasen}\ \emph {et~al.}(2024)\citenamefont {Aasen}, \citenamefont {Di~Giovanni}, \citenamefont {Rotzinger}, \citenamefont {Ustinov},\ and\ \citenamefont {G\"{a}rttner}}]{Aasen2024}%
  \BibitemOpen
  \bibfield  {author} {\bibinfo {author} {\bibfnamefont {A.~S.}\ \bibnamefont {Aasen}}, \bibinfo {author} {\bibfnamefont {A.}~\bibnamefont {Di~Giovanni}}, \bibinfo {author} {\bibfnamefont {H.}~\bibnamefont {Rotzinger}}, \bibinfo {author} {\bibfnamefont {A.~V.}\ \bibnamefont {Ustinov}},\ and\ \bibinfo {author} {\bibfnamefont {M.}~\bibnamefont {G\"{a}rttner}},\ }\bibfield  {title} {\bibinfo {title} {Readout error mitigated quantum state tomography tested on superconducting qubits},\ }\href {https://doi.org/10.1038/s42005-024-01790-8} {\bibfield  {journal} {\bibinfo  {journal} {Communications Physics}\ }\textbf {\bibinfo {volume} {7}},\  (\bibinfo {year} {2024})}\BibitemShut {NoStop}%
\bibitem [{\citenamefont {Ding}\ \emph {et~al.}(2023)\citenamefont {Ding}, \citenamefont {Xu}, \citenamefont {Niu}, \citenamefont {Zhang}, \citenamefont {Bao},\ and\ \citenamefont {Huang}}]{Ding2023}%
  \BibitemOpen
  \bibfield  {author} {\bibinfo {author} {\bibfnamefont {C.}~\bibnamefont {Ding}}, \bibinfo {author} {\bibfnamefont {X.-Y.}\ \bibnamefont {Xu}}, \bibinfo {author} {\bibfnamefont {Y.-F.}\ \bibnamefont {Niu}}, \bibinfo {author} {\bibfnamefont {S.}~\bibnamefont {Zhang}}, \bibinfo {author} {\bibfnamefont {W.-S.}\ \bibnamefont {Bao}},\ and\ \bibinfo {author} {\bibfnamefont {H.-L.}\ \bibnamefont {Huang}},\ }\bibfield  {title} {\bibinfo {title} {Noise-resistant quantum state compression readout},\ }\href {https://doi.org/10.1007/s11433-022-2005-x} {\bibfield  {journal} {\bibinfo  {journal} {Science China Physics, Mechanics \& Astronomy}\ }\textbf {\bibinfo {volume} {66}} (\bibinfo {year} {2023})}\BibitemShut {NoStop}%
\bibitem [{\citenamefont {Huang}\ and\ \citenamefont {Ding}(2025)}]{Huang2025}%
  \BibitemOpen
  \bibfield  {author} {\bibinfo {author} {\bibfnamefont {H.-L.}\ \bibnamefont {Huang}}\ and\ \bibinfo {author} {\bibfnamefont {C.}~\bibnamefont {Ding}},\ }\bibfield  {title} {\bibinfo {title} {Quantum state compression shadow},\ }\href {https://doi.org/10.1140/epjs/s11734-025-01715-8} {\bibfield  {journal} {\bibinfo  {journal} {The European Physical Journal Special Topics}\ ,\ } (\bibinfo {year} {2025})}\BibitemShut {NoStop}%
\bibitem [{\citenamefont {Lundeen}\ \emph {et~al.}(2008)\citenamefont {Lundeen}, \citenamefont {Feito}, \citenamefont {Coldenstrodt-Ronge}, \citenamefont {Pregnell}, \citenamefont {Silberhorn}, \citenamefont {Ralph}, \citenamefont {Eisert}, \citenamefont {Plenio},\ and\ \citenamefont {Walmsley}}]{Lundeen2008}%
  \BibitemOpen
  \bibfield  {author} {\bibinfo {author} {\bibfnamefont {J.~S.}\ \bibnamefont {Lundeen}}, \bibinfo {author} {\bibfnamefont {A.}~\bibnamefont {Feito}}, \bibinfo {author} {\bibfnamefont {H.}~\bibnamefont {Coldenstrodt-Ronge}}, \bibinfo {author} {\bibfnamefont {K.~L.}\ \bibnamefont {Pregnell}}, \bibinfo {author} {\bibfnamefont {C.}~\bibnamefont {Silberhorn}}, \bibinfo {author} {\bibfnamefont {T.~C.}\ \bibnamefont {Ralph}}, \bibinfo {author} {\bibfnamefont {J.}~\bibnamefont {Eisert}}, \bibinfo {author} {\bibfnamefont {M.~B.}\ \bibnamefont {Plenio}},\ and\ \bibinfo {author} {\bibfnamefont {I.~A.}\ \bibnamefont {Walmsley}},\ }\bibfield  {title} {\bibinfo {title} {Tomography of quantum detectors},\ }\href {https://doi.org/10.1038/nphys1133} {\bibfield  {journal} {\bibinfo  {journal} {Nature Physics}\ }\textbf {\bibinfo {volume} {5}},\ \bibinfo {pages} {27–30} (\bibinfo {year} {2008})}\BibitemShut {NoStop}%
\bibitem [{\citenamefont {Cotler}\ and\ \citenamefont {Wilczek}(2020)}]{Cotler2020}%
  \BibitemOpen
  \bibfield  {author} {\bibinfo {author} {\bibfnamefont {J.}~\bibnamefont {Cotler}}\ and\ \bibinfo {author} {\bibfnamefont {F.}~\bibnamefont {Wilczek}},\ }\bibfield  {title} {\bibinfo {title} {Quantum overlapping tomography},\ }\href {https://doi.org/10.1103/physrevlett.124.100401} {\bibfield  {journal} {\bibinfo  {journal} {Physical Review Letters}\ }\textbf {\bibinfo {volume} {124}},\  (\bibinfo {year} {2020})}\BibitemShut {NoStop}%
\bibitem [{\citenamefont {Maciejewski}\ \emph {et~al.}(2021)\citenamefont {Maciejewski}, \citenamefont {Baccari}, \citenamefont {Zimborás},\ and\ \citenamefont {Oszmaniec}}]{Maciejewski2021}%
  \BibitemOpen
  \bibfield  {author} {\bibinfo {author} {\bibfnamefont {F.~B.}\ \bibnamefont {Maciejewski}}, \bibinfo {author} {\bibfnamefont {F.}~\bibnamefont {Baccari}}, \bibinfo {author} {\bibfnamefont {Z.}~\bibnamefont {Zimborás}},\ and\ \bibinfo {author} {\bibfnamefont {M.}~\bibnamefont {Oszmaniec}},\ }\bibfield  {title} {\bibinfo {title} {Modeling and mitigation of cross-talk effects in readout noise with applications to the quantum approximate optimization algorithm},\ }\href {https://doi.org/10.22331/q-2021-06-01-464} {\bibfield  {journal} {\bibinfo  {journal} {Quantum}\ }\textbf {\bibinfo {volume} {5}},\ \bibinfo {pages} {464} (\bibinfo {year} {2021})}\BibitemShut {NoStop}%
\bibitem [{\citenamefont {Georgescu}\ \emph {et~al.}(2014)\citenamefont {Georgescu}, \citenamefont {Ashhab},\ and\ \citenamefont {Nori}}]{Georgescu2014}%
  \BibitemOpen
  \bibfield  {author} {\bibinfo {author} {\bibfnamefont {I.}~\bibnamefont {Georgescu}}, \bibinfo {author} {\bibfnamefont {S.}~\bibnamefont {Ashhab}},\ and\ \bibinfo {author} {\bibfnamefont {F.}~\bibnamefont {Nori}},\ }\bibfield  {title} {\bibinfo {title} {Quantum simulation},\ }\href {https://doi.org/10.1103/revmodphys.86.153} {\bibfield  {journal} {\bibinfo  {journal} {Reviews of Modern Physics}\ }\textbf {\bibinfo {volume} {86}},\ \bibinfo {pages} {153–185} (\bibinfo {year} {2014})}\BibitemShut {NoStop}%
\bibitem [{\citenamefont {Keesling}\ \emph {et~al.}(2019)\citenamefont {Keesling}, \citenamefont {Omran}, \citenamefont {Levine}, \citenamefont {Bernien}, \citenamefont {Pichler}, \citenamefont {Choi}, \citenamefont {Samajdar}, \citenamefont {Schwartz}, \citenamefont {Silvi}, \citenamefont {Sachdev}, \citenamefont {Zoller}, \citenamefont {Endres}, \citenamefont {Greiner}, \citenamefont {Vuletić},\ and\ \citenamefont {Lukin}}]{Keesling2019}%
  \BibitemOpen
  \bibfield  {author} {\bibinfo {author} {\bibfnamefont {A.}~\bibnamefont {Keesling}}, \bibinfo {author} {\bibfnamefont {A.}~\bibnamefont {Omran}}, \bibinfo {author} {\bibfnamefont {H.}~\bibnamefont {Levine}}, \bibinfo {author} {\bibfnamefont {H.}~\bibnamefont {Bernien}}, \bibinfo {author} {\bibfnamefont {H.}~\bibnamefont {Pichler}}, \bibinfo {author} {\bibfnamefont {S.}~\bibnamefont {Choi}}, \bibinfo {author} {\bibfnamefont {R.}~\bibnamefont {Samajdar}}, \bibinfo {author} {\bibfnamefont {S.}~\bibnamefont {Schwartz}}, \bibinfo {author} {\bibfnamefont {P.}~\bibnamefont {Silvi}}, \bibinfo {author} {\bibfnamefont {S.}~\bibnamefont {Sachdev}}, \bibinfo {author} {\bibfnamefont {P.}~\bibnamefont {Zoller}}, \bibinfo {author} {\bibfnamefont {M.}~\bibnamefont {Endres}}, \bibinfo {author} {\bibfnamefont {M.}~\bibnamefont {Greiner}}, \bibinfo {author} {\bibfnamefont {V.}~\bibnamefont {Vuletić}},\ and\ \bibinfo {author} {\bibfnamefont {M.~D.}\ \bibnamefont {Lukin}},\ }\bibfield  {title} {\bibinfo {title} {Quantum
  kibble–zurek mechanism and critical dynamics on a programmable rydberg simulator},\ }\href {https://doi.org/10.1038/s41586-019-1070-1} {\bibfield  {journal} {\bibinfo  {journal} {Nature}\ }\textbf {\bibinfo {volume} {568}},\ \bibinfo {pages} {207–211} (\bibinfo {year} {2019})}\BibitemShut {NoStop}%
\bibitem [{\citenamefont {Zhang}\ \emph {et~al.}(2025)\citenamefont {Zhang}, \citenamefont {Cantú}, \citenamefont {Liu}, \citenamefont {Bylinskii}, \citenamefont {Braverman}, \citenamefont {Huber}, \citenamefont {Amato-Grill}, \citenamefont {Lukin}, \citenamefont {Gemelke}, \citenamefont {Keesling}, \citenamefont {Wang}, \citenamefont {Meurice},\ and\ \citenamefont {Tsai}}]{Zhang2025}%
  \BibitemOpen
  \bibfield  {author} {\bibinfo {author} {\bibfnamefont {J.}~\bibnamefont {Zhang}}, \bibinfo {author} {\bibfnamefont {S.~H.}\ \bibnamefont {Cantú}}, \bibinfo {author} {\bibfnamefont {F.}~\bibnamefont {Liu}}, \bibinfo {author} {\bibfnamefont {A.}~\bibnamefont {Bylinskii}}, \bibinfo {author} {\bibfnamefont {B.}~\bibnamefont {Braverman}}, \bibinfo {author} {\bibfnamefont {F.}~\bibnamefont {Huber}}, \bibinfo {author} {\bibfnamefont {J.}~\bibnamefont {Amato-Grill}}, \bibinfo {author} {\bibfnamefont {A.}~\bibnamefont {Lukin}}, \bibinfo {author} {\bibfnamefont {N.}~\bibnamefont {Gemelke}}, \bibinfo {author} {\bibfnamefont {A.}~\bibnamefont {Keesling}}, \bibinfo {author} {\bibfnamefont {S.-T.}\ \bibnamefont {Wang}}, \bibinfo {author} {\bibfnamefont {Y.}~\bibnamefont {Meurice}},\ and\ \bibinfo {author} {\bibfnamefont {S.-W.}\ \bibnamefont {Tsai}},\ }\bibfield  {title} {\bibinfo {title} {Probing quantum floating phases in rydberg atom arrays},\ }\href {https://doi.org/10.1038/s41467-025-55947-2} {\bibfield  {journal}
  {\bibinfo  {journal} {Nature Communications}\ }\textbf {\bibinfo {volume} {16}},\  (\bibinfo {year} {2025})}\BibitemShut {NoStop}%
\bibitem [{\citenamefont {Schuckert}\ \emph {et~al.}(2025)\citenamefont {Schuckert}, \citenamefont {Katz}, \citenamefont {Feng}, \citenamefont {Crane}, \citenamefont {De}, \citenamefont {Hafezi}, \citenamefont {Gorshkov},\ and\ \citenamefont {Monroe}}]{Schuckert2025}%
  \BibitemOpen
  \bibfield  {author} {\bibinfo {author} {\bibfnamefont {A.}~\bibnamefont {Schuckert}}, \bibinfo {author} {\bibfnamefont {O.}~\bibnamefont {Katz}}, \bibinfo {author} {\bibfnamefont {L.}~\bibnamefont {Feng}}, \bibinfo {author} {\bibfnamefont {E.}~\bibnamefont {Crane}}, \bibinfo {author} {\bibfnamefont {A.}~\bibnamefont {De}}, \bibinfo {author} {\bibfnamefont {M.}~\bibnamefont {Hafezi}}, \bibinfo {author} {\bibfnamefont {A.~V.}\ \bibnamefont {Gorshkov}},\ and\ \bibinfo {author} {\bibfnamefont {C.}~\bibnamefont {Monroe}},\ }\bibfield  {title} {\bibinfo {title} {Observation of a finite-energy phase transition in a one-dimensional quantum simulator},\ }\href {https://doi.org/10.1038/s41567-024-02751-2} {\bibfield  {journal} {\bibinfo  {journal} {Nature Physics}\ }\textbf {\bibinfo {volume} {21}},\ \bibinfo {pages} {374–379} (\bibinfo {year} {2025})}\BibitemShut {NoStop}%
\bibitem [{\citenamefont {Richerme}\ \emph {et~al.}(2014)\citenamefont {Richerme}, \citenamefont {Gong}, \citenamefont {Lee}, \citenamefont {Senko}, \citenamefont {Smith}, \citenamefont {Foss-Feig}, \citenamefont {Michalakis}, \citenamefont {Gorshkov},\ and\ \citenamefont {Monroe}}]{Richerme2014}%
  \BibitemOpen
  \bibfield  {author} {\bibinfo {author} {\bibfnamefont {P.}~\bibnamefont {Richerme}}, \bibinfo {author} {\bibfnamefont {Z.-X.}\ \bibnamefont {Gong}}, \bibinfo {author} {\bibfnamefont {A.}~\bibnamefont {Lee}}, \bibinfo {author} {\bibfnamefont {C.}~\bibnamefont {Senko}}, \bibinfo {author} {\bibfnamefont {J.}~\bibnamefont {Smith}}, \bibinfo {author} {\bibfnamefont {M.}~\bibnamefont {Foss-Feig}}, \bibinfo {author} {\bibfnamefont {S.}~\bibnamefont {Michalakis}}, \bibinfo {author} {\bibfnamefont {A.~V.}\ \bibnamefont {Gorshkov}},\ and\ \bibinfo {author} {\bibfnamefont {C.}~\bibnamefont {Monroe}},\ }\bibfield  {title} {\bibinfo {title} {Non-local propagation of correlations in quantum systems with long-range interactions},\ }\href {https://doi.org/10.1038/nature13450} {\bibfield  {journal} {\bibinfo  {journal} {Nature}\ }\textbf {\bibinfo {volume} {511}},\ \bibinfo {pages} {198–201} (\bibinfo {year} {2014})}\BibitemShut {NoStop}%
\bibitem [{\citenamefont {Jurcevic}\ \emph {et~al.}(2014)\citenamefont {Jurcevic}, \citenamefont {Lanyon}, \citenamefont {Hauke}, \citenamefont {Hempel}, \citenamefont {Zoller}, \citenamefont {Blatt},\ and\ \citenamefont {Roos}}]{Jurcevic2014}%
  \BibitemOpen
  \bibfield  {author} {\bibinfo {author} {\bibfnamefont {P.}~\bibnamefont {Jurcevic}}, \bibinfo {author} {\bibfnamefont {B.~P.}\ \bibnamefont {Lanyon}}, \bibinfo {author} {\bibfnamefont {P.}~\bibnamefont {Hauke}}, \bibinfo {author} {\bibfnamefont {C.}~\bibnamefont {Hempel}}, \bibinfo {author} {\bibfnamefont {P.}~\bibnamefont {Zoller}}, \bibinfo {author} {\bibfnamefont {R.}~\bibnamefont {Blatt}},\ and\ \bibinfo {author} {\bibfnamefont {C.~F.}\ \bibnamefont {Roos}},\ }\bibfield  {title} {\bibinfo {title} {Quasiparticle engineering and entanglement propagation in a quantum many-body system},\ }\href {https://doi.org/10.1038/nature13461} {\bibfield  {journal} {\bibinfo  {journal} {Nature}\ }\textbf {\bibinfo {volume} {511}},\ \bibinfo {pages} {202–205} (\bibinfo {year} {2014})}\BibitemShut {NoStop}%
\bibitem [{\citenamefont {Di~Giovanni}\ \emph {et~al.}(tion)\citenamefont {Di~Giovanni}, \citenamefont {Aasen}, \citenamefont {Lisenfeld}, \citenamefont {G\"{a}rttner}, \citenamefont {Rotzinger},\ and\ \citenamefont {Ustinov}}]{DiGiovanni2025B}%
  \BibitemOpen
  \bibfield  {author} {\bibinfo {author} {\bibfnamefont {A.}~\bibnamefont {Di~Giovanni}}, \bibinfo {author} {\bibfnamefont {A.~S.}\ \bibnamefont {Aasen}}, \bibinfo {author} {\bibfnamefont {J.}~\bibnamefont {Lisenfeld}}, \bibinfo {author} {\bibfnamefont {M.}~\bibnamefont {G\"{a}rttner}}, \bibinfo {author} {\bibfnamefont {H.}~\bibnamefont {Rotzinger}},\ and\ \bibinfo {author} {\bibfnamefont {A.~V.}\ \bibnamefont {Ustinov}},\ }\bibfield  {title} {\bibinfo {title} {Correlation-conscious optimization of multiplexed qubit readout},\ }\href@noop {} {\  (\bibinfo {year} {in preparation})}\BibitemShut {NoStop}%
\bibitem [{\citenamefont {Fiurášek}(2001)}]{Fiurek2001}%
  \BibitemOpen
  \bibfield  {author} {\bibinfo {author} {\bibfnamefont {J.}~\bibnamefont {Fiurášek}},\ }\bibfield  {title} {\bibinfo {title} {Maximum-likelihood estimation of quantum measurement},\ }\href {https://doi.org/10.1103/physreva.64.024102} {\bibfield  {journal} {\bibinfo  {journal} {Physical Review A}\ }\textbf {\bibinfo {volume} {64}},\  (\bibinfo {year} {2001})}\BibitemShut {NoStop}%
\bibitem [{\citenamefont {Blume-Kohout}(2010)}]{BlumeKohout2010}%
  \BibitemOpen
  \bibfield  {author} {\bibinfo {author} {\bibfnamefont {R.}~\bibnamefont {Blume-Kohout}},\ }\bibfield  {title} {\bibinfo {title} {Optimal, reliable estimation of quantum states},\ }\href {https://doi.org/10.1088/1367-2630/12/4/043034} {\bibfield  {journal} {\bibinfo  {journal} {New Journal of Physics}\ }\textbf {\bibinfo {volume} {12}},\ \bibinfo {pages} {043034} (\bibinfo {year} {2010})}\BibitemShut {NoStop}%
\bibitem [{\citenamefont {Lvovsky}(2004)}]{Lvovsky2004}%
  \BibitemOpen
  \bibfield  {author} {\bibinfo {author} {\bibfnamefont {A.~I.}\ \bibnamefont {Lvovsky}},\ }\bibfield  {title} {\bibinfo {title} {Iterative maximum-likelihood reconstruction in quantum homodyne tomography},\ }\href {https://doi.org/10.1088/1464-4266/6/6/014} {\bibfield  {journal} {\bibinfo  {journal} {Journal of Optics B: Quantum and Semiclassical Optics}\ }\textbf {\bibinfo {volume} {6}},\ \bibinfo {pages} {S556–S559} (\bibinfo {year} {2004})}\BibitemShut {NoStop}%
\bibitem [{\citenamefont {García-Pérez}\ \emph {et~al.}(2020)\citenamefont {García-Pérez}, \citenamefont {Rossi}, \citenamefont {Sokolov}, \citenamefont {Borrelli},\ and\ \citenamefont {Maniscalco}}]{GarcaPrez2020}%
  \BibitemOpen
  \bibfield  {author} {\bibinfo {author} {\bibfnamefont {G.}~\bibnamefont {García-Pérez}}, \bibinfo {author} {\bibfnamefont {M.~A.~C.}\ \bibnamefont {Rossi}}, \bibinfo {author} {\bibfnamefont {B.}~\bibnamefont {Sokolov}}, \bibinfo {author} {\bibfnamefont {E.-M.}\ \bibnamefont {Borrelli}},\ and\ \bibinfo {author} {\bibfnamefont {S.}~\bibnamefont {Maniscalco}},\ }\bibfield  {title} {\bibinfo {title} {Pairwise tomography networks for many-body quantum systems},\ }\href {https://doi.org/10.1103/physrevresearch.2.023393} {\bibfield  {journal} {\bibinfo  {journal} {Physical Review Research}\ }\textbf {\bibinfo {volume} {2}},\  (\bibinfo {year} {2020})}\BibitemShut {NoStop}%
\bibitem [{\citenamefont {Bonet-Monroig}\ \emph {et~al.}(2020)\citenamefont {Bonet-Monroig}, \citenamefont {Babbush},\ and\ \citenamefont {O’Brien}}]{BonetMonroig2020}%
  \BibitemOpen
  \bibfield  {author} {\bibinfo {author} {\bibfnamefont {X.}~\bibnamefont {Bonet-Monroig}}, \bibinfo {author} {\bibfnamefont {R.}~\bibnamefont {Babbush}},\ and\ \bibinfo {author} {\bibfnamefont {T.~E.}\ \bibnamefont {O’Brien}},\ }\bibfield  {title} {\bibinfo {title} {Nearly optimal measurement scheduling for partial tomography of quantum states},\ }\href {https://doi.org/10.1103/physrevx.10.031064} {\bibfield  {journal} {\bibinfo  {journal} {Physical Review X}\ }\textbf {\bibinfo {volume} {10}},\  (\bibinfo {year} {2020})}\BibitemShut {NoStop}%
\bibitem [{\citenamefont {Walker~II}\ and\ \citenamefont {Colbourn}(2007)}]{WalkerII2007}%
  \BibitemOpen
  \bibfield  {author} {\bibinfo {author} {\bibfnamefont {R.~A.}\ \bibnamefont {Walker~II}}\ and\ \bibinfo {author} {\bibfnamefont {C.~J.}\ \bibnamefont {Colbourn}},\ }\bibfield  {title} {\bibinfo {title} {Perfect hash families: Constructions and existence},\ }\href {https://doi.org/10.1515/jmc.2007.008} {\bibfield  {journal} {\bibinfo  {journal} {Journal of Mathematical Cryptology}\ }\textbf {\bibinfo {volume} {1}},\  (\bibinfo {year} {2007})}\BibitemShut {NoStop}%
\bibitem [{\citenamefont {Dougherty}(2019)}]{dougherty2019}%
  \BibitemOpen
  \bibfield  {author} {\bibinfo {author} {\bibfnamefont {R.}~\bibnamefont {Dougherty}},\ }\href {https://hdl.handle.net/2286/R.I.55479} {\emph {\bibinfo {title} {Hash families and applications to t-restrictions}}}\ (\bibinfo  {publisher} {Arizona State University},\ \bibinfo {year} {2019})\BibitemShut {NoStop}%
\bibitem [{\citenamefont {Aasen}(2025)}]{Aasen2025}%
  \BibitemOpen
  \bibfield  {author} {\bibinfo {author} {\bibfnamefont {A.}~\bibnamefont {Aasen}},\ }\href@noop {} {\bibinfo {title} {{CREMST}}},\ \bibinfo {howpublished} {\url{https://github.com/AdrianAasen/CREMST}} (\bibinfo {year} {2025})\BibitemShut {NoStop}%
\bibitem [{\citenamefont {Hansenne}\ \emph {et~al.}(2024)\citenamefont {Hansenne}, \citenamefont {Qu}, \citenamefont {Weinbrenner}, \citenamefont {de~Gois}, \citenamefont {Wang}, \citenamefont {Ming}, \citenamefont {Yang}, \citenamefont {Horodecki}, \citenamefont {Gao},\ and\ \citenamefont {Gühne}}]{Kiara2024}%
  \BibitemOpen
  \bibfield  {author} {\bibinfo {author} {\bibfnamefont {K.}~\bibnamefont {Hansenne}}, \bibinfo {author} {\bibfnamefont {R.}~\bibnamefont {Qu}}, \bibinfo {author} {\bibfnamefont {L.~T.}\ \bibnamefont {Weinbrenner}}, \bibinfo {author} {\bibfnamefont {C.}~\bibnamefont {de~Gois}}, \bibinfo {author} {\bibfnamefont {H.}~\bibnamefont {Wang}}, \bibinfo {author} {\bibfnamefont {Y.}~\bibnamefont {Ming}}, \bibinfo {author} {\bibfnamefont {Z.}~\bibnamefont {Yang}}, \bibinfo {author} {\bibfnamefont {P.}~\bibnamefont {Horodecki}}, \bibinfo {author} {\bibfnamefont {W.}~\bibnamefont {Gao}},\ and\ \bibinfo {author} {\bibfnamefont {O.}~\bibnamefont {Gühne}},\ }\href {https://arxiv.org/abs/2408.05730} {\bibinfo {title} {Optimal overlapping tomography}} (\bibinfo {year} {2024}),\ \Eprint {https://arxiv.org/abs/2408.05730} {arXiv:2408.05730 [quant-ph]} \BibitemShut {NoStop}%
\bibitem [{\citenamefont {Jerger}\ \emph {et~al.}(2012)\citenamefont {Jerger}, \citenamefont {Poletto}, \citenamefont {Macha}, \citenamefont {H\"{u}bner}, \citenamefont {Il’ichev},\ and\ \citenamefont {Ustinov}}]{Jerger2012}%
  \BibitemOpen
  \bibfield  {author} {\bibinfo {author} {\bibfnamefont {M.}~\bibnamefont {Jerger}}, \bibinfo {author} {\bibfnamefont {S.}~\bibnamefont {Poletto}}, \bibinfo {author} {\bibfnamefont {P.}~\bibnamefont {Macha}}, \bibinfo {author} {\bibfnamefont {U.}~\bibnamefont {H\"{u}bner}}, \bibinfo {author} {\bibfnamefont {E.}~\bibnamefont {Il’ichev}},\ and\ \bibinfo {author} {\bibfnamefont {A.~V.}\ \bibnamefont {Ustinov}},\ }\bibfield  {title} {\bibinfo {title} {Frequency division multiplexing readout and simultaneous manipulation of an array of flux qubits},\ }\href {https://doi.org/10.1063/1.4739454} {\bibfield  {journal} {\bibinfo  {journal} {Applied Physics Letters}\ }\textbf {\bibinfo {volume} {101}},\ \bibinfo {pages} {042604} (\bibinfo {year} {2012})}\BibitemShut {NoStop}%
\bibitem [{\citenamefont {Krantz}\ \emph {et~al.}(2019)\citenamefont {Krantz}, \citenamefont {Kjaergaard}, \citenamefont {Yan}, \citenamefont {Orlando}, \citenamefont {Gustavsson},\ and\ \citenamefont {Oliver}}]{Krantz2019}%
  \BibitemOpen
  \bibfield  {author} {\bibinfo {author} {\bibfnamefont {P.}~\bibnamefont {Krantz}}, \bibinfo {author} {\bibfnamefont {M.}~\bibnamefont {Kjaergaard}}, \bibinfo {author} {\bibfnamefont {F.}~\bibnamefont {Yan}}, \bibinfo {author} {\bibfnamefont {T.~P.}\ \bibnamefont {Orlando}}, \bibinfo {author} {\bibfnamefont {S.}~\bibnamefont {Gustavsson}},\ and\ \bibinfo {author} {\bibfnamefont {W.~D.}\ \bibnamefont {Oliver}},\ }\bibfield  {title} {\bibinfo {title} {A quantum engineer’s guide to superconducting qubits},\ }\href {https://doi.org/10.1063/1.5089550} {\bibfield  {journal} {\bibinfo  {journal} {Applied Physics Reviews}\ }\textbf {\bibinfo {volume} {6}},\  (\bibinfo {year} {2019})}\BibitemShut {NoStop}%
\bibitem [{\citenamefont {Giovanni}\ \emph {et~al.}(2025)\citenamefont {Giovanni}, \citenamefont {Aasen}, \citenamefont {Lisenfeld}, \citenamefont {Gärttner}, \citenamefont {Rotzinger},\ and\ \citenamefont {Ustinov}}]{DiGiovanni2025A}%
  \BibitemOpen
  \bibfield  {author} {\bibinfo {author} {\bibfnamefont {A.~D.}\ \bibnamefont {Giovanni}}, \bibinfo {author} {\bibfnamefont {A.~S.}\ \bibnamefont {Aasen}}, \bibinfo {author} {\bibfnamefont {J.}~\bibnamefont {Lisenfeld}}, \bibinfo {author} {\bibfnamefont {M.}~\bibnamefont {Gärttner}}, \bibinfo {author} {\bibfnamefont {H.}~\bibnamefont {Rotzinger}},\ and\ \bibinfo {author} {\bibfnamefont {A.~V.}\ \bibnamefont {Ustinov}},\ }\href {https://arxiv.org/abs/2502.08589} {\bibinfo {title} {Multiplexed qubit readout quality metric beyond assignment fidelity}} (\bibinfo {year} {2025}),\ \Eprint {https://arxiv.org/abs/2502.08589} {arXiv:2502.08589 [quant-ph]} \BibitemShut {NoStop}%
\bibitem [{\citenamefont {Harris}\ \emph {et~al.}(2020)\citenamefont {Harris}, \citenamefont {Millman}, \citenamefont {van~der Walt}, \citenamefont {Gommers}, \citenamefont {Virtanen}, \citenamefont {Cournapeau}, \citenamefont {Wieser}, \citenamefont {Taylor}, \citenamefont {Berg}, \citenamefont {Smith}, \citenamefont {Kern}, \citenamefont {Picus}, \citenamefont {Hoyer}, \citenamefont {van Kerkwijk}, \citenamefont {Brett}, \citenamefont {Haldane}, \citenamefont {del R{\'{i}}o}, \citenamefont {Wiebe}, \citenamefont {Peterson}, \citenamefont {G{\'{e}}rard-Marchant}, \citenamefont {Sheppard}, \citenamefont {Reddy}, \citenamefont {Weckesser}, \citenamefont {Abbasi}, \citenamefont {Gohlke},\ and\ \citenamefont {Oliphant}}]{harris2020}%
  \BibitemOpen
  \bibfield  {author} {\bibinfo {author} {\bibfnamefont {C.~R.}\ \bibnamefont {Harris}}, \bibinfo {author} {\bibfnamefont {K.~J.}\ \bibnamefont {Millman}}, \bibinfo {author} {\bibfnamefont {S.~J.}\ \bibnamefont {van~der Walt}}, \bibinfo {author} {\bibfnamefont {R.}~\bibnamefont {Gommers}}, \bibinfo {author} {\bibfnamefont {P.}~\bibnamefont {Virtanen}}, \bibinfo {author} {\bibfnamefont {D.}~\bibnamefont {Cournapeau}}, \bibinfo {author} {\bibfnamefont {E.}~\bibnamefont {Wieser}}, \bibinfo {author} {\bibfnamefont {J.}~\bibnamefont {Taylor}}, \bibinfo {author} {\bibfnamefont {S.}~\bibnamefont {Berg}}, \bibinfo {author} {\bibfnamefont {N.~J.}\ \bibnamefont {Smith}}, \bibinfo {author} {\bibfnamefont {R.}~\bibnamefont {Kern}}, \bibinfo {author} {\bibfnamefont {M.}~\bibnamefont {Picus}}, \bibinfo {author} {\bibfnamefont {S.}~\bibnamefont {Hoyer}}, \bibinfo {author} {\bibfnamefont {M.~H.}\ \bibnamefont {van Kerkwijk}}, \bibinfo {author} {\bibfnamefont {M.}~\bibnamefont {Brett}}, \bibinfo {author} {\bibfnamefont
  {A.}~\bibnamefont {Haldane}}, \bibinfo {author} {\bibfnamefont {J.~F.}\ \bibnamefont {del R{\'{i}}o}}, \bibinfo {author} {\bibfnamefont {M.}~\bibnamefont {Wiebe}}, \bibinfo {author} {\bibfnamefont {P.}~\bibnamefont {Peterson}}, \bibinfo {author} {\bibfnamefont {P.}~\bibnamefont {G{\'{e}}rard-Marchant}}, \bibinfo {author} {\bibfnamefont {K.}~\bibnamefont {Sheppard}}, \bibinfo {author} {\bibfnamefont {T.}~\bibnamefont {Reddy}}, \bibinfo {author} {\bibfnamefont {W.}~\bibnamefont {Weckesser}}, \bibinfo {author} {\bibfnamefont {H.}~\bibnamefont {Abbasi}}, \bibinfo {author} {\bibfnamefont {C.}~\bibnamefont {Gohlke}},\ and\ \bibinfo {author} {\bibfnamefont {T.~E.}\ \bibnamefont {Oliphant}},\ }\bibfield  {title} {\bibinfo {title} {Array programming with {NumPy}},\ }\href {https://doi.org/10.1038/s41586-020-2649-2} {\bibfield  {journal} {\bibinfo  {journal} {Nature}\ }\textbf {\bibinfo {volume} {585}},\ \bibinfo {pages} {357} (\bibinfo {year} {2020})}\BibitemShut {NoStop}%
\bibitem [{\citenamefont {Virtanen}\ \emph {et~al.}(2020)\citenamefont {Virtanen}, \citenamefont {Gommers}, \citenamefont {Oliphant}, \citenamefont {Haberland}, \citenamefont {Reddy}, \citenamefont {Cournapeau}, \citenamefont {Burovski}, \citenamefont {Peterson}, \citenamefont {Weckesser}, \citenamefont {Bright}, \citenamefont {van~der Walt}, \citenamefont {Brett}, \citenamefont {Wilson}, \citenamefont {Millman}, \citenamefont {Mayorov}, \citenamefont {Nelson}, \citenamefont {Jones}, \citenamefont {Kern}, \citenamefont {Larson}, \citenamefont {Carey}, \citenamefont {Polat}, \citenamefont {Feng}, \citenamefont {Moore}, \citenamefont {VanderPlas}, \citenamefont {Laxalde}, \citenamefont {Perktold}, \citenamefont {Cimrman}, \citenamefont {Henriksen}, \citenamefont {Quintero}, \citenamefont {Harris}, \citenamefont {Archibald}, \citenamefont {Ribeiro}, \citenamefont {Pedregosa}, \citenamefont {van Mulbregt}, \citenamefont {Vijaykumar}, \citenamefont {Bardelli}, \citenamefont {Rothberg}, \citenamefont {Hilboll},
  \citenamefont {Kloeckner}, \citenamefont {Scopatz}, \citenamefont {Lee}, \citenamefont {Rokem}, \citenamefont {Woods}, \citenamefont {Fulton}, \citenamefont {Masson}, \citenamefont {H\"{a}ggstr\"{o}m}, \citenamefont {Fitzgerald}, \citenamefont {Nicholson}, \citenamefont {Hagen}, \citenamefont {Pasechnik}, \citenamefont {Olivetti}, \citenamefont {Martin}, \citenamefont {Wieser}, \citenamefont {Silva}, \citenamefont {Lenders}, \citenamefont {Wilhelm}, \citenamefont {Young}, \citenamefont {Price}, \citenamefont {Ingold}, \citenamefont {Allen}, \citenamefont {Lee}, \citenamefont {Audren}, \citenamefont {Probst}, \citenamefont {Dietrich}, \citenamefont {Silterra}, \citenamefont {Webber}, \citenamefont {Slavič}, \citenamefont {Nothman}, \citenamefont {Buchner}, \citenamefont {Kulick}, \citenamefont {Sch\"{o}nberger}, \citenamefont {de~Miranda~Cardoso}, \citenamefont {Reimer}, \citenamefont {Harrington}, \citenamefont {Rodríguez}, \citenamefont {Nunez-Iglesias}, \citenamefont {Kuczynski}, \citenamefont {Tritz},
  \citenamefont {Thoma}, \citenamefont {Newville}, \citenamefont {K\"{u}mmerer}, \citenamefont {Bolingbroke}, \citenamefont {Tartre}, \citenamefont {Pak}, \citenamefont {Smith}, \citenamefont {Nowaczyk}, \citenamefont {Shebanov}, \citenamefont {Pavlyk}, \citenamefont {Brodtkorb}, \citenamefont {Lee}, \citenamefont {McGibbon}, \citenamefont {Feldbauer}, \citenamefont {Lewis}, \citenamefont {Tygier}, \citenamefont {Sievert}, \citenamefont {Vigna}, \citenamefont {Peterson}, \citenamefont {More}, \citenamefont {Pudlik}, \citenamefont {Oshima}, \citenamefont {Pingel}, \citenamefont {Robitaille}, \citenamefont {Spura}, \citenamefont {Jones}, \citenamefont {Cera}, \citenamefont {Leslie}, \citenamefont {Zito}, \citenamefont {Krauss}, \citenamefont {Upadhyay}, \citenamefont {Halchenko},\ and\ \citenamefont {Vázquez-Baeza}}]{Virtanen2020}%
  \BibitemOpen
  \bibfield  {author} {\bibinfo {author} {\bibfnamefont {P.}~\bibnamefont {Virtanen}}, \bibinfo {author} {\bibfnamefont {R.}~\bibnamefont {Gommers}}, \bibinfo {author} {\bibfnamefont {T.~E.}\ \bibnamefont {Oliphant}}, \bibinfo {author} {\bibfnamefont {M.}~\bibnamefont {Haberland}}, \bibinfo {author} {\bibfnamefont {T.}~\bibnamefont {Reddy}}, \bibinfo {author} {\bibfnamefont {D.}~\bibnamefont {Cournapeau}}, \bibinfo {author} {\bibfnamefont {E.}~\bibnamefont {Burovski}}, \bibinfo {author} {\bibfnamefont {P.}~\bibnamefont {Peterson}}, \bibinfo {author} {\bibfnamefont {W.}~\bibnamefont {Weckesser}}, \bibinfo {author} {\bibfnamefont {J.}~\bibnamefont {Bright}}, \bibinfo {author} {\bibfnamefont {S.~J.}\ \bibnamefont {van~der Walt}}, \bibinfo {author} {\bibfnamefont {M.}~\bibnamefont {Brett}}, \bibinfo {author} {\bibfnamefont {J.}~\bibnamefont {Wilson}}, \bibinfo {author} {\bibfnamefont {K.~J.}\ \bibnamefont {Millman}}, \bibinfo {author} {\bibfnamefont {N.}~\bibnamefont {Mayorov}}, \bibinfo {author} {\bibfnamefont
  {A.~R.~J.}\ \bibnamefont {Nelson}}, \bibinfo {author} {\bibfnamefont {E.}~\bibnamefont {Jones}}, \bibinfo {author} {\bibfnamefont {R.}~\bibnamefont {Kern}}, \bibinfo {author} {\bibfnamefont {E.}~\bibnamefont {Larson}}, \bibinfo {author} {\bibfnamefont {C.~J.}\ \bibnamefont {Carey}}, \bibinfo {author} {\bibfnamefont {I.}~\bibnamefont {Polat}}, \bibinfo {author} {\bibfnamefont {Y.}~\bibnamefont {Feng}}, \bibinfo {author} {\bibfnamefont {E.~W.}\ \bibnamefont {Moore}}, \bibinfo {author} {\bibfnamefont {J.}~\bibnamefont {VanderPlas}}, \bibinfo {author} {\bibfnamefont {D.}~\bibnamefont {Laxalde}}, \bibinfo {author} {\bibfnamefont {J.}~\bibnamefont {Perktold}}, \bibinfo {author} {\bibfnamefont {R.}~\bibnamefont {Cimrman}}, \bibinfo {author} {\bibfnamefont {I.}~\bibnamefont {Henriksen}}, \bibinfo {author} {\bibfnamefont {E.~A.}\ \bibnamefont {Quintero}}, \bibinfo {author} {\bibfnamefont {C.~R.}\ \bibnamefont {Harris}}, \bibinfo {author} {\bibfnamefont {A.~M.}\ \bibnamefont {Archibald}}, \bibinfo {author}
  {\bibfnamefont {A.~H.}\ \bibnamefont {Ribeiro}}, \bibinfo {author} {\bibfnamefont {F.}~\bibnamefont {Pedregosa}}, \bibinfo {author} {\bibfnamefont {P.}~\bibnamefont {van Mulbregt}}, \bibinfo {author} {\bibfnamefont {A.}~\bibnamefont {Vijaykumar}}, \bibinfo {author} {\bibfnamefont {A.~P.}\ \bibnamefont {Bardelli}}, \bibinfo {author} {\bibfnamefont {A.}~\bibnamefont {Rothberg}}, \bibinfo {author} {\bibfnamefont {A.}~\bibnamefont {Hilboll}}, \bibinfo {author} {\bibfnamefont {A.}~\bibnamefont {Kloeckner}}, \bibinfo {author} {\bibfnamefont {A.}~\bibnamefont {Scopatz}}, \bibinfo {author} {\bibfnamefont {A.}~\bibnamefont {Lee}}, \bibinfo {author} {\bibfnamefont {A.}~\bibnamefont {Rokem}}, \bibinfo {author} {\bibfnamefont {C.~N.}\ \bibnamefont {Woods}}, \bibinfo {author} {\bibfnamefont {C.}~\bibnamefont {Fulton}}, \bibinfo {author} {\bibfnamefont {C.}~\bibnamefont {Masson}}, \bibinfo {author} {\bibfnamefont {C.}~\bibnamefont {H\"{a}ggstr\"{o}m}}, \bibinfo {author} {\bibfnamefont {C.}~\bibnamefont {Fitzgerald}},
  \bibinfo {author} {\bibfnamefont {D.~A.}\ \bibnamefont {Nicholson}}, \bibinfo {author} {\bibfnamefont {D.~R.}\ \bibnamefont {Hagen}}, \bibinfo {author} {\bibfnamefont {D.~V.}\ \bibnamefont {Pasechnik}}, \bibinfo {author} {\bibfnamefont {E.}~\bibnamefont {Olivetti}}, \bibinfo {author} {\bibfnamefont {E.}~\bibnamefont {Martin}}, \bibinfo {author} {\bibfnamefont {E.}~\bibnamefont {Wieser}}, \bibinfo {author} {\bibfnamefont {F.}~\bibnamefont {Silva}}, \bibinfo {author} {\bibfnamefont {F.}~\bibnamefont {Lenders}}, \bibinfo {author} {\bibfnamefont {F.}~\bibnamefont {Wilhelm}}, \bibinfo {author} {\bibfnamefont {G.}~\bibnamefont {Young}}, \bibinfo {author} {\bibfnamefont {G.~A.}\ \bibnamefont {Price}}, \bibinfo {author} {\bibfnamefont {G.-L.}\ \bibnamefont {Ingold}}, \bibinfo {author} {\bibfnamefont {G.~E.}\ \bibnamefont {Allen}}, \bibinfo {author} {\bibfnamefont {G.~R.}\ \bibnamefont {Lee}}, \bibinfo {author} {\bibfnamefont {H.}~\bibnamefont {Audren}}, \bibinfo {author} {\bibfnamefont {I.}~\bibnamefont {Probst}},
  \bibinfo {author} {\bibfnamefont {J.~P.}\ \bibnamefont {Dietrich}}, \bibinfo {author} {\bibfnamefont {J.}~\bibnamefont {Silterra}}, \bibinfo {author} {\bibfnamefont {J.~T.}\ \bibnamefont {Webber}}, \bibinfo {author} {\bibfnamefont {J.}~\bibnamefont {Slavič}}, \bibinfo {author} {\bibfnamefont {J.}~\bibnamefont {Nothman}}, \bibinfo {author} {\bibfnamefont {J.}~\bibnamefont {Buchner}}, \bibinfo {author} {\bibfnamefont {J.}~\bibnamefont {Kulick}}, \bibinfo {author} {\bibfnamefont {J.~L.}\ \bibnamefont {Sch\"{o}nberger}}, \bibinfo {author} {\bibfnamefont {J.~V.}\ \bibnamefont {de~Miranda~Cardoso}}, \bibinfo {author} {\bibfnamefont {J.}~\bibnamefont {Reimer}}, \bibinfo {author} {\bibfnamefont {J.}~\bibnamefont {Harrington}}, \bibinfo {author} {\bibfnamefont {J.~L.~C.}\ \bibnamefont {Rodríguez}}, \bibinfo {author} {\bibfnamefont {J.}~\bibnamefont {Nunez-Iglesias}}, \bibinfo {author} {\bibfnamefont {J.}~\bibnamefont {Kuczynski}}, \bibinfo {author} {\bibfnamefont {K.}~\bibnamefont {Tritz}}, \bibinfo {author}
  {\bibfnamefont {M.}~\bibnamefont {Thoma}}, \bibinfo {author} {\bibfnamefont {M.}~\bibnamefont {Newville}}, \bibinfo {author} {\bibfnamefont {M.}~\bibnamefont {K\"{u}mmerer}}, \bibinfo {author} {\bibfnamefont {M.}~\bibnamefont {Bolingbroke}}, \bibinfo {author} {\bibfnamefont {M.}~\bibnamefont {Tartre}}, \bibinfo {author} {\bibfnamefont {M.}~\bibnamefont {Pak}}, \bibinfo {author} {\bibfnamefont {N.~J.}\ \bibnamefont {Smith}}, \bibinfo {author} {\bibfnamefont {N.}~\bibnamefont {Nowaczyk}}, \bibinfo {author} {\bibfnamefont {N.}~\bibnamefont {Shebanov}}, \bibinfo {author} {\bibfnamefont {O.}~\bibnamefont {Pavlyk}}, \bibinfo {author} {\bibfnamefont {P.~A.}\ \bibnamefont {Brodtkorb}}, \bibinfo {author} {\bibfnamefont {P.}~\bibnamefont {Lee}}, \bibinfo {author} {\bibfnamefont {R.~T.}\ \bibnamefont {McGibbon}}, \bibinfo {author} {\bibfnamefont {R.}~\bibnamefont {Feldbauer}}, \bibinfo {author} {\bibfnamefont {S.}~\bibnamefont {Lewis}}, \bibinfo {author} {\bibfnamefont {S.}~\bibnamefont {Tygier}}, \bibinfo {author}
  {\bibfnamefont {S.}~\bibnamefont {Sievert}}, \bibinfo {author} {\bibfnamefont {S.}~\bibnamefont {Vigna}}, \bibinfo {author} {\bibfnamefont {S.}~\bibnamefont {Peterson}}, \bibinfo {author} {\bibfnamefont {S.}~\bibnamefont {More}}, \bibinfo {author} {\bibfnamefont {T.}~\bibnamefont {Pudlik}}, \bibinfo {author} {\bibfnamefont {T.}~\bibnamefont {Oshima}}, \bibinfo {author} {\bibfnamefont {T.~J.}\ \bibnamefont {Pingel}}, \bibinfo {author} {\bibfnamefont {T.~P.}\ \bibnamefont {Robitaille}}, \bibinfo {author} {\bibfnamefont {T.}~\bibnamefont {Spura}}, \bibinfo {author} {\bibfnamefont {T.~R.}\ \bibnamefont {Jones}}, \bibinfo {author} {\bibfnamefont {T.}~\bibnamefont {Cera}}, \bibinfo {author} {\bibfnamefont {T.}~\bibnamefont {Leslie}}, \bibinfo {author} {\bibfnamefont {T.}~\bibnamefont {Zito}}, \bibinfo {author} {\bibfnamefont {T.}~\bibnamefont {Krauss}}, \bibinfo {author} {\bibfnamefont {U.}~\bibnamefont {Upadhyay}}, \bibinfo {author} {\bibfnamefont {Y.~O.}\ \bibnamefont {Halchenko}},\ and\ \bibinfo {author}
  {\bibfnamefont {Y.}~\bibnamefont {Vázquez-Baeza}},\ }\bibfield  {title} {\bibinfo {title} {Scipy 1.0: fundamental algorithms for scientific computing in python},\ }\href {https://doi.org/10.1038/s41592-019-0686-2} {\bibfield  {journal} {\bibinfo  {journal} {Nature Methods}\ }\textbf {\bibinfo {volume} {17}},\ \bibinfo {pages} {261–272} (\bibinfo {year} {2020})}\BibitemShut {NoStop}%
\bibitem [{\citenamefont {{Joblib Development Team}}(2021)}]{joblib}%
  \BibitemOpen
  \bibfield  {author} {\bibinfo {author} {\bibnamefont {{Joblib Development Team}}},\ }\href {https://joblib.readthedocs.io/} {\bibinfo {title} {Joblib: running python functions as pipeline jobs}} (\bibinfo {year} {2021})\BibitemShut {NoStop}%
\bibitem [{\citenamefont {Hunter}(2007)}]{Hunter2007}%
  \BibitemOpen
  \bibfield  {author} {\bibinfo {author} {\bibfnamefont {J.~D.}\ \bibnamefont {Hunter}},\ }\bibfield  {title} {\bibinfo {title} {Matplotlib: A 2d graphics environment},\ }\href {https://doi.org/10.1109/MCSE.2007.55} {\bibfield  {journal} {\bibinfo  {journal} {Computing in Science \& Engineering}\ }\textbf {\bibinfo {volume} {9}},\ \bibinfo {pages} {90} (\bibinfo {year} {2007})}\BibitemShut {NoStop}%
\bibitem [{\citenamefont {Kay}(2023)}]{Kay2018}%
  \BibitemOpen
  \bibfield  {author} {\bibinfo {author} {\bibfnamefont {A.}~\bibnamefont {Kay}},\ }\href {https://arxiv.org/abs/1809.03842} {\bibinfo {title} {Tutorial on the quantikz package}} (\bibinfo {year} {2023}),\ \Eprint {https://arxiv.org/abs/1809.03842} {arXiv:1809.03842 [quant-ph]} \BibitemShut {NoStop}%
\bibitem [{\citenamefont {Müllner}(2011)}]{Muller2011}%
  \BibitemOpen
  \bibfield  {author} {\bibinfo {author} {\bibfnamefont {D.}~\bibnamefont {Müllner}},\ }\href {https://arxiv.org/abs/1109.2378} {\bibinfo {title} {Modern hierarchical, agglomerative clustering algorithms}} (\bibinfo {year} {2011}),\ \Eprint {https://arxiv.org/abs/1109.2378} {arXiv:1109.2378 [stat.ML]} \BibitemShut {NoStop}%
\bibitem [{\citenamefont {Shang}\ \emph {et~al.}(2017)\citenamefont {Shang}, \citenamefont {Zhang},\ and\ \citenamefont {Ng}}]{Shang2017}%
  \BibitemOpen
  \bibfield  {author} {\bibinfo {author} {\bibfnamefont {J.}~\bibnamefont {Shang}}, \bibinfo {author} {\bibfnamefont {Z.}~\bibnamefont {Zhang}},\ and\ \bibinfo {author} {\bibfnamefont {H.~K.}\ \bibnamefont {Ng}},\ }\bibfield  {title} {\bibinfo {title} {Superfast maximum-likelihood reconstruction for quantum tomography},\ }\href {https://doi.org/10.1103/physreva.95.062336} {\bibfield  {journal} {\bibinfo  {journal} {Physical Review A}\ }\textbf {\bibinfo {volume} {95}},\  (\bibinfo {year} {2017})}\BibitemShut {NoStop}%
\bibitem [{\citenamefont {Hou}\ \emph {et~al.}(2016)\citenamefont {Hou}, \citenamefont {Zhong}, \citenamefont {Tian}, \citenamefont {Dong}, \citenamefont {Qi}, \citenamefont {Li}, \citenamefont {Wang}, \citenamefont {Nori}, \citenamefont {Xiang}, \citenamefont {Li},\ and\ \citenamefont {Guo}}]{Hou2016}%
  \BibitemOpen
  \bibfield  {author} {\bibinfo {author} {\bibfnamefont {Z.}~\bibnamefont {Hou}}, \bibinfo {author} {\bibfnamefont {H.-S.}\ \bibnamefont {Zhong}}, \bibinfo {author} {\bibfnamefont {Y.}~\bibnamefont {Tian}}, \bibinfo {author} {\bibfnamefont {D.}~\bibnamefont {Dong}}, \bibinfo {author} {\bibfnamefont {B.}~\bibnamefont {Qi}}, \bibinfo {author} {\bibfnamefont {L.}~\bibnamefont {Li}}, \bibinfo {author} {\bibfnamefont {Y.}~\bibnamefont {Wang}}, \bibinfo {author} {\bibfnamefont {F.}~\bibnamefont {Nori}}, \bibinfo {author} {\bibfnamefont {G.-Y.}\ \bibnamefont {Xiang}}, \bibinfo {author} {\bibfnamefont {C.-F.}\ \bibnamefont {Li}},\ and\ \bibinfo {author} {\bibfnamefont {G.-C.}\ \bibnamefont {Guo}},\ }\bibfield  {title} {\bibinfo {title} {Full reconstruction of a 14-qubit state within four hours},\ }\href {https://doi.org/10.1088/1367-2630/18/8/083036} {\bibfield  {journal} {\bibinfo  {journal} {New Journal of Physics}\ }\textbf {\bibinfo {volume} {18}},\ \bibinfo {pages} {083036} (\bibinfo {year} {2016})}\BibitemShut
  {NoStop}%
\bibitem [{\citenamefont {Mahler}\ \emph {et~al.}(2013)\citenamefont {Mahler}, \citenamefont {Rozema}, \citenamefont {Darabi}, \citenamefont {Ferrie}, \citenamefont {Blume-Kohout},\ and\ \citenamefont {Steinberg}}]{Mahler2013}%
  \BibitemOpen
  \bibfield  {author} {\bibinfo {author} {\bibfnamefont {D.~H.}\ \bibnamefont {Mahler}}, \bibinfo {author} {\bibfnamefont {L.~A.}\ \bibnamefont {Rozema}}, \bibinfo {author} {\bibfnamefont {A.}~\bibnamefont {Darabi}}, \bibinfo {author} {\bibfnamefont {C.}~\bibnamefont {Ferrie}}, \bibinfo {author} {\bibfnamefont {R.}~\bibnamefont {Blume-Kohout}},\ and\ \bibinfo {author} {\bibfnamefont {A.~M.}\ \bibnamefont {Steinberg}},\ }\bibfield  {title} {\bibinfo {title} {Adaptive quantum state tomography improves accuracy quadratically},\ }\href {https://doi.org/10.1103/physrevlett.111.183601} {\bibfield  {journal} {\bibinfo  {journal} {Physical Review Letters}\ }\textbf {\bibinfo {volume} {111}},\  (\bibinfo {year} {2013})}\BibitemShut {NoStop}%
\end{thebibliography}%

\end{document}